\documentclass{amsart}
\usepackage{amsmath, amsfonts, amssymb, amsthm, array, stmaryrd, graphicx, hyperref}
\usepackage{textcomp, verbatim, color, mathrsfs, mathabx, braket}
\newtheorem{thm}{Theorem}[section]
\newtheorem{lemma}{Lemma}[section]
\newtheorem{prop}[lemma]{Proposition}
\newtheorem{cor}[lemma]{Corollary}

\newtheorem{rmk}[lemma]{Remark}

\newtheorem*{acklg}{Acknowledgements}
\numberwithin{equation}{section}
\allowdisplaybreaks

\title{Derivation of a nonlinear Schr\"{o}dinger equation with a general power-type nonlinearity}

\author{ Zhihui Xie }
\address{Department of Mathematics \newline\indent The University of Texas at Austin}
\email{zxie@math.utexas.edu}

\date{\today}

\begin{document}

\begin{abstract}
In this paper we study the derivation of a certain type of NLS from many-body interactions of bosonic particles. We consider a model with a finite linear combination of $n$-body interactions, where $n \geq 2$ is an integer. We show that the $k$-particle marginal density of the BBGKY hierarchy converges when particle number goes to infinity, and the limit solves a corresponding infinite Gross-Pitaevskii hierarchy. We prove the uniqueness of factorized solution to the Gross-Pitaevskii hierarchy based on a priori space time estimates. The convergence is established by adapting the arguments originated or developed in \cite{ESY}, \cite{KSS} and \cite{CPquintic}. For the uniqueness part, we expand the procedure followed in \cite{KM} by introducing a different board game argument to handle the new contraction operator. This new board game argument helps us obtain a good estimate on the Duhamel terms. In \cite{KM}, the relevant space time estimates are assumed to be true, while we give a prove for it. 
\end{abstract}

\maketitle

\section{Introduction}
The nonlinear Schr\"odinger equation (NLS) is a macroscopic model for a quantum mechanical system, with different type of nonlinearities depending on the way we model the interaction potential (cubic, quintic, Hartree, etc.) in a quantum many body system. A derivation of the corresponding PDE that governs the system is a hot topic in mathematical physics that has been drawing much attention during the past decade. Some of the references in this direction include \cite{Spohn},\cite{EY2001},\cite{EESY},\cite{ESY06}, \cite{ESY},\cite{KM},\cite{KSS},\cite{CPquintic},\cite{GMM1},\cite{GMM2},\cite{xuwen} etc. In particular, the sequence of crucial works by Elgart, Erd\"os, Schlein and Yau \cite{EY2001}, \cite{EESY}, \cite{ESY06}, \cite{ESY} studied a model of Bose gas in $\mathbb{R}^3$ with pairwise interactions and rescaled potentials $V^{(p)}_N$ approaching a delta function. They proved that the $k$-particle density matrix for BBGKY hierarchy converges to that of the infinite hierarchy (GP hierarchy), 
which is actually governed by the solution of the cubic non-linear Schr\"odinger equation. In their work, uniqueness of 
solutions to the GP hierarchy is established via Feynman diagrams. 


In this paper, we derive a nonlinear Schr\"odinger equation with a linear combination of power type nonlinearities. Our work is motivated by \cite{KSS} and \cite{CPquintic}, in which the authors consider a quantum model with $2$-body interactions \cite{KSS} and $3$-body interactions \cite{CPquintic} respectively and obtain cubic and quintic NLS correspondingly that correctly describes the system. It is also worth mentioning that in \cite{CPquintic}, Chen and Pavlovi\'c predict that, if both $2$-body and $3$-body interactions are present in a quantum model, then that would lead (via Gross-Pitaevskii limit) to a NLS with a linear combination of cubic and quintic nonlinearities. In this paper, we will give a proof of that claim. Actually, we generalize the prediction from \cite{CPquintic} and derive the NLS with a finite linear combination of power nonlinearities. We also note that a particular example of such kind of NLS was studied by Tao-Visan-Zhang in \cite{TVZ}, in which local and global wellposedness and 
related questions are explored.

\subsection{BBGKY hierarchy}

We consider a quantum mechanical system of $N$ bosonic particles in ${\mathbb{R}}^d$, with $d \in \{1,2\}$. Let $p$ and $p_0$ be positive integers, fixed $p_0$, $1\leq p\leq p_0$. The time evolution of the $N$ particle wave function $\psi_N \in L^2_s(\mathbb{R}^{dN})$ is governed by the Schr\"odinger equation
\begin{equation}   \label{Hamiltonian equation}
 i\partial_t \psi_{N,t} = H_N \psi_{N,t}
\end{equation}
with the Hamiltonian
\small
\begin{equation}  \label{Hamiltonian}
 H_N:=\sum_{i=1}^{N}(-\Delta_{x_i})+\sum_{p=1}^{p_0}\frac{1}{N^p}\sum_{1\leq i_1<\cdots<i_{p+1}\leq N}N^{pd\beta}V^{(p)}\big(N^{\beta}(x_{i_1}-x_{i_2}),\cdots,N^{\beta}(x_{i_1}-x_{i_{p+1}})\big)
\end{equation}
\normalsize
on Hilbert space $L^2_s(\mathbb{R}^{dN})$, which is the subspace of $L^2(\mathbb{R}^{dN})$ consisting of all functions satisfying 
$$\psi_N(x_{\sigma(1)},x_{\sigma(2)}...,x_{\sigma(N)})=\psi_N(x_1,x_2...,x_N),$$
for any permutation $\sigma \in S_N$ and $0<\beta<\frac{1}{2dp+2}$. Also we assume that for all $1\leq p\leq p_0$ the $(p+1)$-body interaction potential $V^{(p)} \in W^{p,\infty}(\mathbb{R}^{pd})$ is a non-negative function with sufficient regularity and it is translation-invariant so that it can be written in the above form. For instance, when $p=2$, we have that $V^{(2)}(x_1-x_2, x_2-x_3, x_1-x_3)=V^{(2)}(x_1-x_2 ,-(x_1-x_2)+(x_1-x_3), x_1-x_3)=V^{(2)}(x_1-x_2, x_1-x_3).$ The first part of the Hamiltonian represents the kinetic energy part, while the second is the sum of interaction potentials involving $p+1$ particles. 
 
Note that \eqref{Hamiltonian equation} is linear, which together with the fact that $H_N$ is a self-adjoint operator implies that global in time solutions can be 
written by means of the unitary group generated by $H_N$ as
\begin{equation}
 \psi_{N,t}=e^{-iH_N t} \psi_N,  \qquad \forall t \in \mathbb{R}
\end{equation}

As in previous works on derivation of NLS from many body quantum dynamical systems \cite{ESY, KSS, CPquintic},  
we define the corresponding \emph{$N$-particle density matrix} as follows: 
\begin{equation}
\gamma_{N,t}(t,\textbf{x}_N;\textbf{x}'_N)=\psi_{N,t}(\textbf{x}_N) \bar\psi_{N,t}(\textbf{x}'_N)
\end{equation}
$\bar\psi_{N,t}$ denotes the complex conjugate of $\psi_{N,t}$.
Then \eqref{Hamiltonian equation} implies that
\begin{equation}      \label{Heisenberg equation}
 i\partial_t \gamma_{N,t}=[H_N, \gamma_{N,t}],
\end{equation}
where the Heisenberg commutator is given as usual $[A,B]:=AB-BA$. The $L^2$-normalization of $\psi_{N,t}$ implies that $Tr\gamma_{N,t}=1$.
By taking partial trace of $\gamma_{N,t}$ over the last $N-k$ particles we define the \emph{k-particle marginal density}:
\begin{equation}
 \gamma_{N,t}^{(k)}(t,\textbf{x}_k;\textbf{x}'_k)=\int \gamma_{N,t}(t,\textbf{x}_k,\textbf{x}_{N-k};\textbf{x}'_k,\textbf{x}_{N-k})d{\textbf{x}_{N-k}}
\end{equation}
where $\textbf{x}_k=(x_1,\cdots,x_k), \textbf{x}_{N-k}=(x_{k+1},\cdots,x_N), \; k=1,...,N$. 

Let $V^{(p)}_N(x_1,x_2,\cdots,x_p):=N^{pd\beta}V^{(p)}(N^{\beta}x_1,N^{\beta}x_2,\cdots,N^{\beta}x_{p})$. We can verify that the marginal densities satisfy the so called BBGKY hierarchy
\small
\begin{equation} \label{BBGKY}
 \begin{split}
  & i\partial_t \gamma_{N,t}^{(k)} =  \\
  & \quad \sum_{i=1}^{k}[-\Delta_{x_i}, \gamma_{N,t}^{(k)}]   \\
  & +\sum_{p=1}^{p_0}\Bigg\{\frac{1}{N^{p}}\sum_{1\leq i_1<\cdots<i_{p+1}\leq k}[V^{(p)}_N(x_{i_1}-x_{i_2},\cdots,x_{i_1}-x_{i_{p+1}}),\gamma_{N,t}^{(k)}]  \\
  & \qquad\quad +\frac{N-k}{N^{p}}\sum_{1\leq i_1<\cdots<i_{p} \leq k} Tr_{k+1}[V^{(p)}_N(x_{i_1}-x_{i_2},\cdots,x_{i_1}-x_{i_{p}},x_{i_1}-x_{k+1}),\gamma_{N,t}^{(k+1)}]  \\
  & \qquad\quad +\frac{(N-k)(N-k-1)}{N^{p}} \sum_{1\leq i_1<\cdots<i_{p-1}\leq k}Tr_{k+1}Tr_{k+2}  \\
  & \qquad \qquad \qquad [V^{(p)}_N(x_{i_1}-x_{i_2},\cdots,x_{i_1}-x_{i_{p-1}},x_{i_1}-x_{k+1},x_{i_2}-x_{k+2}),\gamma_{N,t}^{(k+2)}]  \\
  & \qquad\quad +\cdots  \\
  & \qquad\quad +\frac{(N-k)(N-k-1)\cdots(N-k-p+1)}{N^{p}}\sum_{1 \leq i_1 \leq k} Tr_{k+1}Tr_{k+2} \cdots Tr_{k+p}  \\
  & \qquad \qquad \qquad [V^{(p)}_N(x_{i_1}-x_{k+1},x_{i_1}-x_{k+2},\cdots,x_{i_1}-x_{k+p}),\gamma_{N,t}^{(k+p)}]\Bigg\}. 
\end{split}
\end{equation}
\normalsize
Here we use the convention that $\gamma_{N,t}^{(k)}=0$, whenever $k>N$. The symbol $Tr_{k+j}$ denotes the partial trace over the $m$-th particle, i.e, the kernel of the $k$-particle operator  $Tr_{k+1}[V^{(p)}_N(x_{i_1}-x_{i_2},\cdots,x_{i_1}-x_{i_{p}},x_{i_1}-x_{k+1}),\gamma_{N,t}^{(k+1)}]$ is given by
\small
\begin{equation}
  \begin{split}
  & \big(Tr_{k+1}[V^{(p)}_N(x_{i_1}-x_{i_2},\cdots,x_{i_1}-x_{i_{p}},x_{i_1}-x_{k+1}),\gamma_{N,t}^{(k+1)}]\big)(\textbf{x}_k;\textbf{x}'_k)  \\
  & = \int V^{(p)}_N(x_{i_1}-x_{i_2},\cdots,x_{i_1}-x_{i_{p}},x_{i_1}-x_{k+1})\gamma^{(k+1)}(\textbf{x}_k,x_{k+1};\textbf{x}'_k,x_{k+1})dx_{k+1}  \\
  &\quad -\int V^{(p)}_N(x'_{i_1}-x'_{i_2},\cdots,x'_{i_1}-x'_{i_{p}},x'_{i_1}-x_{k+1})\gamma^{(k+1)}(\textbf{x}_k,x_{k+1};\textbf{x}'_k,x_{k+1}) dx_{k+1}
  \end{split} 
\end{equation}
\normalsize
Let us present a heuristic argument on what one expects when taking $N \to \infty$. 
In particular, we note that all the terms in \eqref{BBGKY}, except the first term on the RHS and the last term in the bracket, are expected to vanish for fixed k and sufficiently small $\beta$, because $\frac{1}{N^{p}} \to 0$, $\frac{\prod_{i=0}^{j}(N-k-i)}{N^{p}} \to 0$, $\forall 0 \leq j \leq p-2$. The last interaction term on the RHS is expected to 
survive thanks to $\frac{\prod_{i=0}^{p-1}(N-k-i)}{N^{p}} \to 1$. 
Indeed, one can make this heuristic precise and prove existence of a weak sequential limit of \eqref{BBGKY} in 
the same topology that was originally used in \cite{ESY}, and subsequently in \cite{KSS, CPquintic}. Details are presented in Section \ref{section: convergence}. In such a way one shows that the corresponding infinite(GP) hierarchy is a weak sequential limit of \eqref{BBGKY}. 

\subsection{GP hierarchy} 
Following the convention in \cite{CPcauchy}, we formally write down the limit of \eqref{BBGKY} as $N \to \infty$, as follows: 
\begin{equation} \label{p-GP hierarchy}
 i\partial_t \gamma^{(k)}_{\infty,t}=\sum_{j=1}^{k}\big(-\Delta_{x_j}+\Delta_{x'_j}\big)\gamma^{(k)}_{\infty,t}+ \sum_{p=1}^{p_0} b_0^{(p)}\sum_{j=1}^{k}B_{j;k+1,\dots,k+p}\gamma^{(k+p)}_{\infty,t}
\end{equation}
for any $k\geq 1$. We call \eqref{p-GP hierarchy} \emph{cubic Gross-Pitaevskii (GP) hierarchy} if $p=1$; \emph{quintic GP hierarchy} if $p=2$ and \emph{septic GP hierarchy} if $p=3$, and so on. 
Here $b_0^{(p)}$ is the $L^1$ norm of the non-negative potential: $b_0^{(p)}=\int_{\mathbb{R}^{pd}} V^{(p)}(x_1,\cdots,x_p)dx_1\cdots dx_p$. 

The \emph{contraction operator} is given via 
\begin{equation}  \label{definition of Bj}
B_{j;k+1,\cdots,k+p}:=B^{+}_{j;k+1,\cdots,k+p}-B^{-}_{j;k+1,\cdots,k+p} 
\end{equation}
where
\begin{equation}    \label{defintion of Bj+}
 \begin{split}
  & \quad \Big(B^{+}_{j;k+1,\cdots,k+p}\gamma^{(k+p)}_{\infty,t}\Big)(t,\textbf{x}_k,\textbf{x}'_k) \\
  & :=\int \delta(x_j-x_{k+1}) \delta (x_j-x'_{k+1})\cdots \delta(x_j-x_{k+p}) \delta (x_j-x'_{k+p}) \\
  & \qquad \times \gamma^{(k+p)}_{\infty,t}(t,x_1,\cdots,x_{k+p};x'_1,\cdots,x'_{k+p}) dx_{k+1}dx'_{k+1}\cdots dx_{k+p}dx'_{k+p}
 \end{split}
\end{equation}
and
\begin{equation}     \label{defintion of Bj-}
  \begin{split}
  & \quad \Big(B^{-}_{j;k+1,\cdots,k+p}\gamma^{(k+p)}_{\infty,t}\Big)(t,\textbf{x}_k,\textbf{x}'_k) \\  
  & := \int \delta(x'_j-x_{k+1}) \delta (x'_j-x'_{k+1})\cdots \delta(x'_j-x_{k+p}) \delta (x'_j-x'_{k+p})   \\
  & \qquad \times \gamma^{(k+p)}_{\infty,t}(t,x_1,\cdots,x_{k+p};x'_1,\cdots,x'_{k+p}) dx_{k+1}dx'_{k+1}\cdots dx_{k+p}dx'_{k+p}
 \end{split}
\end{equation}

We can check that	
\begin{equation}      \label{factorized solution}
 \gamma_{\infty,t}^{(k)}=\Ket{\phi_t} \Bra{\phi_t}^{\otimes k}=\prod_{j=1}^{k} \phi_t(x_j)\bar\phi_t(x'_j)
\end{equation}
is a solution to \eqref{p-GP hierarchy} if $\phi_t$ is a solution to the nonlinear Schr\"odinger equation
\begin{equation}
 i\partial_t \phi_t=-\Delta \phi_t+\sum_{p=1}^{p_0}b_0^{(p)}|\phi_t|^{2p} \phi_t.
\end{equation}
We hope to establish the uniqueness on solutions of the Gross-Pitaevskii hierarchy, and build the following convergence under appropriate topology:
\begin{equation}
 \gamma_{N,t}^{(k)} \to \gamma_{\infty,t}^{(k)}, \qquad as\ \ N\to \infty, \quad \forall k \geq 1
\end{equation}
The uniqueness of solution to cubic GP hierarchy is proved in \cite{ESY} by Erd\"os-Schlein-Yau in a suitable space. By use of s sophisticated Feynman graph expansions. Fourier integrals associated to these graphs take most of the efforts in their analysis. Then later in \cite{KM}, a new method has been developed by Klainerman and Machedon to deal with the uniqueness part in a different space of density matrices. This approach also uses the expansion introduced in \cite{ESY}, but the authors take advantage of the space-time estimate obtained from free evolving Schr\"odinger equations, and thus yielding a comparatively simpler analysis on the contributions of expansion terms. Subsequent works like \cite{KSS} by Kirkpatrick, Schlein and Staffilani, and \cite{CPquintic} by Chen and Pavlovi\'{c} proceeded along their lines when considering the Bose gas with pair and three-body interactions respectively, and the solutions obtained in both \cite{KSS} and \cite{CPquintic} satisfies the Klainerman-Machedon bounds.

In this paper, we prove the following result:
\begin{thm}  \label{thm:main result}
 Let $p_0\geq 1$ be a fixed integer. Suppose that for all $1\leq p\leq p_0$ the potential $V^{(p)}\in W^{p,\infty}(\mathbb{R}^{dp})$ and $V^{(p)}\geq 0$ is translation-invariant.  Let $d\in \{1,2\}$ and $0<\beta<\frac{1}{2dp_0+2}$. $\{\psi_N\}_{N\geq 1}$ is a family of functions that satisfy
\begin{equation}     \label{condition for Nk order growth of Hamiltonian}
 \sup_{N}\frac{1}{N}\langle \psi_N, H_N\psi_N\rangle <\infty
\end{equation}
and assume $\{\psi_N\}_{N\geq 1}$ exhibits asymptotic factorization: $\exists \phi\in L^2(\mathbb{R}^d)$ such that $Tr\big|\gamma_N^{(1)}-\Ket{\phi} \Bra{\phi}\big|\to 0$ as $N\to 0$. $\gamma_N^{(1)}$ is the 1-particle marginal density associated with $\psi_N$.  \\
Then we have
\begin{equation}     \label{main convergence}
 Tr\big|\gamma_{N,t}^{(k)}-\Ket{\phi_t} \Bra{\phi_t}^{\otimes k}\big|\to 0 \quad as \ \ N\to \infty
\end{equation}
Here $\gamma_{N,t}^{(k)}$ is the k-particle marginal density associated to $\psi_{N,t}=e^{-iH_N t} \psi_N$, and $\phi_t$ solves the nonlinear Schr\"odinger equation: $ i\partial_t \phi_t=-\Delta \phi_t+\sum_{p=1}^{p_0} b_0^{(p)}|\phi_t|^{2p} \phi_t
$ with initial condition $\phi_0=\phi$ and potential constant $b_0^{(p)}=\int_{\mathbb{R}^{pd}}V^{(p)}(x)dx<\infty$.
\end{thm}

The bulk of this  paper is devoted to the proof of Theorem \ref{thm:main result}. The strategy we follow is to identify the limit of $\Gamma_{N,t}=\{\gamma_{N,t}^{(k)}\}_{k=1}^N$ as the unique solution to \eqref{p-GP hierarchy}; or in other words, every limit (under suitable topology) of $\Gamma_{N,t}$ solves \eqref{p-GP hierarchy} uniquely, since \eqref{factorized solution} is a solution, then \eqref{main convergence} follows by compactness.  

The idea to prove uniqueness of the infinite hierarchy in \cite{KM} consists of the following three major steps. First, we express each solution $\gamma^{(k)}$ in terms of the future iterates $\gamma^{(k+p_0)},..., \gamma^{(k+n p_0)}$ using Duhamel formula (we choose all $p$ to be $p_0$ for a upper bound of the number of terms). Since for each $p_0$, the operator $B^{k}_{k+p_0}=\sum_{j=1}^k B_{j, k+1,\cdots,k+p_0 n}$ is a sum of $k$ operators, the iterated Duhamel formula involves up to $k(k+p_0)\cdots\big(k+p_0(n-1)\big) \sim n!$ terms (see $J^k$ in \eqref{gamma expansion}). Then in the second step, we use a combinatorial argument to group these iterated terms into equivalence classes that we can bound. Finally, we treat each equivalence class with the Strichartz type estimate \eqref{free evolving bound}. 

Compared to \cite{CPquintic}, the main novelties are: 
\begin{itemize} 
\item  in the proof of an a priori energy bound (Proposition \ref{prop:finite priori energy bound}) which had to be carefully done due to presence of many terms in the interacting potential;
\item in the combinatorial argument, since in the case considered in this paper the matrices associated with iterated Duhamel terms reflect a combination of different interactions. 
\end{itemize} 
 
\subsection*{Organization of the paper} 
In section 2, we prove a priori energy bound for the BBGKY solutions and summarize mains steps on establishing the convergence of $k$-particle marginals to the infinite hierarchy. In section 3, we obtain two space-time estimates for the limiting hierarchy. In section 4, a free evolving bound on the limiting hierarchy is presented, which is later used to prove the uniqueness of solutions in 2D case. Sections 5-7 are devoted to the proof of uniqueness of solutions to the limiting hierarchy. We prove 1D case in section 5. In section 6, we obtain the results from board game arguments (first introduced in \cite{KM}), which, combined with the bounds in section 3 and 4 lead to the uniqueness in 2D case. Finally, two technical lemmas are included in the appendix sections.

\section{Convergence}   \label{section: convergence}

\subsection{A Priori Energy Bounds}
From the energy estimates, following \cite{KSS},\cite{CPquintic},\cite{ESY},\cite{EESY},\cite{EY2001}, we will be able to obtain the priori bounds below. 

\begin{prop}   \label{prop:finite priori energy bound}
 Suppose $0<\beta<\frac{1}{2dp_0+2}$, then there exists a constant $C$ (depends on $p_0, V^{(p)}, d$), such that for every $k$, there exists $N_0(k)$ such that 
\begin{equation}   \label{finite priori energy bound}
 \langle \psi, (H_N+N)^k\psi \rangle \geq C^k N^k \langle \psi, (1-\Delta_{x_1})\cdots(1-\Delta_{x_k})\psi \rangle
\end{equation}
for all $N\geq N_0(k)$, $\psi \in L^2_s(\mathbb{R}^{dN})$. The Hamiltonian $H_N$ is defined as in \eqref{Hamiltonian}.
\end{prop}

\begin{proof}
 We adapt the proof in \cite{KSS},\cite{CPquintic} to the current case. It's a two-step induction over $k\geq 0$. For $k=0$ the statement is trivial and for $k=1$ the statement follows from $V^{(p)}_N \geq 0$. In order to illustrate the techniques here, we check one more case before the running of induction. Write $S_i=(1-\Delta_{x_i})^{\frac{1}{2}}$ and the interactions in two groups $h_1$ and $h_2$, such that $H_N+N=h_1+h_2$:
\begin{align*}
  & h_1=\sum\limits_{j=n+1}^{N} S_j^2    \\
  & h_2=\sum_{j=1}^{n} S_j^2+\sum_{p=1}^{p_0}\sum_{i_1<i_2<\cdots<i_{p+1}} N^{-p}V^{(p)}_N(x_{i_1}-x_{i_2},\cdots,x_{i_1}-x_{i_{1+p}})
\end{align*}
 For $k=2$, let $h_1=\sum\limits_{j=1}^{N} S_j^2$ and $h_2=\sum_p \sum N^{-p}V^{(p)}_N(x_{i_1}-x_{i_2},\cdots,x_{i_1}-x_{i_{1+p}})$, then since $h_2^2\geq 0$, 
\begin{equation}    \label{k=2}
 \begin{split}
  & \langle \psi, (H_N+N)^2\psi \rangle   \\
  &=\langle \psi, h_1^2\psi \rangle + \langle \psi, h_1h_2\psi \rangle+ \langle \psi, h_2h_1\psi \rangle +\langle \psi, h_2^2\psi \rangle   \\
  &\geq \langle \psi, h_1^2\psi \rangle + \langle \psi, h_1h_2\psi \rangle+ \langle \psi, h_2h_1\psi \rangle   \\
  &=\boxed{N(N-1)\langle\psi, S_1^2S_2^2\psi\rangle + N\langle\psi, S_1^4\psi\rangle}  \quad  (\text{``leading terms''}) \\
  &\quad + \boxed{\sum_{p=1}^{p_0} N\sum N^{-p}(\langle\psi, S_1^2 V^{(p)}_N(x_{i_1}-x_{i_2},\cdots,x_{i_1}-x_{i_{1+p}})\psi\rangle+c.c )} \quad (\text{``error terms''})
\end{split}
\end{equation}
where c.c denotes ``complex conjugate``.  We keep the ``leading terms'' in RHS of \eqref{k=2} and look for a lower bound of the terms in the last line (``error terms''). As in \cite{CPquintic}, let $\dot{S}_j=(\dot{S}_{j,i})_{i=1}^d:=i\nabla_{x_j}$, then $S_j^2=1+\dot{S}_j^2=1-\Delta_{x_j}$. For sufficiently large $N$, by the permutation symmetry of $\psi$:
\begin{align}
  &N\sum N^{-p}(\langle\psi, S_1^2 V^{(p)}_N(x_{i_1}-x_{i_2},\cdots,x_{i_1}-x_{i_{1+p}})\psi\rangle+c.c )     \notag \\
  &= N^{1-p}(N-1)\cdots(N-p-1)(\langle\psi, S_1^2V^{(p)}_N(x_2-x_3,\cdots,x_2-x_{2+p})\psi\rangle +c.c )  \notag \\
  &\quad +N^{1-p}(N-1)\cdots(N-p)(\langle\psi, S_1^2V^{(p)}_N(x_1-x_2,\cdots,x_1-x_{1+p})\psi\rangle+c.c )    \notag \\
  &\geq C^2N^2(\langle\psi, S_1^2V^{(p)}_N(x_2-x_3,\cdots,x_2-x_{2+p})\psi\rangle+c.c )  \notag \\
  &\quad +CN(\langle\psi, (1+\dot{S}_1^2)V^{(p)}_N(x_1-x_2,\cdots,x_1-x_{1+p})\psi\rangle+c.c )   \notag \\
  &\geq CN(\langle\psi, \dot{S}_1^2 V^{(p)}_N(x_1-x_2,\cdots,x_1-x_{1+p})\psi\rangle+c.c )   \label{drop positive potential} \\
  &\geq -CN\big|\langle\psi, \dot{S}_1\big(\nabla_{x_1}V^{(p)}_N(x_1-x_2,\cdots,x_1-x_{1+p})\big)\psi\rangle\big|    \notag \\
  &\geq -CN\rho\big|\langle\psi, S_1^2\psi\rangle\big|-\frac{CN}{\rho}\big|\langle\psi, |\nabla_{x_1}V^{(p)}_N|^2\psi\big|   \notag  \\
  &\geq -CN\rho\big|\langle\psi, S_1^2\psi\rangle\big|-\frac{CN}{\rho}\|\nabla V^{(p)}_N\|^2_{L^{\infty}(\mathbb{R}^{dp})}\langle\psi, S_1^2S_2^2\psi\rangle   \label{apply lemma sobolev inequality}   \\
  &=-CN\rho\big|\langle\psi, S_1^2\psi\rangle\big|-\frac{CN^{1+(2pd+2)\beta}}{\rho}\|\nabla V^{(p)}\|^2_{L^{\infty}(\mathbb{R}^{dp})}\langle\psi, S_1^2S_2^2\psi\rangle   \notag
\end{align}
To obtain \eqref{drop positive potential}, we dropped positive terms using the positivity of $V^{(p)}_N$; $\rho>0$ is arbitrary and we've applied Lemma \ref{lem: sobolev inequality} to obtain \eqref{apply lemma sobolev inequality}. Thus 
\small
\begin{equation*}
 \begin{split}
 &\quad \langle \psi, (H_N+N)^2\psi \rangle    \\
 &\geq  N(N-1)\langle\psi, S_1^2S_2^2\psi\rangle + N\langle\psi, S_1^4\psi\rangle-\sum_{p=1}^{p_0}\Big(CN\rho\big|\langle\psi, S_1^2\psi\rangle\big|+\frac{CN^{1+(2pd+2)\beta}}{\rho}\langle\psi, S_1^2S_2^2\psi\rangle \Big)   \\
 &\geq C^2N^2\langle\psi, S_1^2S_2^2\psi\rangle,  \qquad \text{for big $N$ with $\beta<\frac{1}{2dp_0+2}$.}
\end{split}
\end{equation*}
\normalsize

The basic idea in the proof is to derive a lower bound of the ``error terms'' which is further dominated by the ``leading terms''. Now assume \eqref{finite priori energy bound} is true for all $k\leq n$, then we prove it holds for $k=n+2$. 
For big enough $N$, by the induction assumption, we have (since $H_N+N$ is self-adjoint):
\begin{equation}      \label{by induction}
 \langle \psi, (H_N+N)^{n+2}\psi \rangle  \geq C^nN^n\langle \psi, (H_N+N)S_1^2 \cdots S_n^2(H_N+N)\psi \rangle
\end{equation}
Then it follows that
\begin{equation}      \label{two group identity}
 \begin{split}
  & \langle \psi, (H_N+N)S_1^2 \cdots S_n^2(H_N+N)\psi \rangle   \\
  & =\langle \psi, h_1S_1^2 \cdots S_n^2h_1\psi \rangle + \langle \psi, h_1S_1^2 \cdots S_n^2h_2\psi \rangle   \\
  & \quad + \langle \psi, h_2S_1^2 \cdots S_n^2h_1\psi \rangle+\langle \psi, h_2S_1^2 \cdots S_n^2h_2\psi \rangle   \\
\end{split}
\end{equation}
Note that $h_2S_1^2 \cdots S_n^2h_2 \geq 0$. Combine \eqref{by induction} and \eqref{two group identity} and use the permutation symmetry of $\psi$ to get:
\small
\begin{equation}    \label{n+2 case}
 \begin{split}
  & \langle \psi, (H_N+N)^{n+2}\psi \rangle    \\
  & \geq C^nN^n\big(\langle \psi, h_1S_1^2 \cdots S_n^2h_1\psi \rangle + \langle \psi, h_1S_1^2 \cdots S_n^2h_2\psi \rangle +\langle \psi, h_2S_1^2 \cdots S_n^2h_1\psi \rangle \big)\\
  &\geq C^nN^n(N-n)(N-n-1)\langle \psi, S_1^2 \cdots S_{n+2}^2\psi \rangle + C^nN^n(N-n)n\langle \psi, S_1^4S_2^2 \cdots S_{n+1}^2\psi \rangle    \\
  & \quad + \sum_{p=1}^{p_0} C^nN^n\frac{(N-n)}{N^{p}}\sum_{i_1<\cdots<i_{p}}\big(\langle \psi, S_1^2 \cdots S_{n+1}^2V^{(p)}_N(x_{i_1}-x_{i_2},\cdots,x_{i_1}-x_{i_{1+p}})\psi \rangle+c.c \big)
 \end{split}
\end{equation}
\normalsize
The last term above is the error term we want to control. Again by permutation symmetry of $\psi$, we can further break down the interactions of the last term in \eqref{n+2 case} for big enough $N$:
\small
\begin{align}
  & \langle \psi, (H_N+N)^{n+2}\psi \rangle    \notag  \\
  & \geq  C^{n+2}N^{n+2} \langle \psi, S_1^2 \cdots S_{n+2}^2\psi \rangle + C^{n+1}N^{n+1}\big( \langle \psi, S_1^4 \cdots S_{n+1}^2\psi \rangle    \label{leading terms}   \\
  & \quad +\sum_{p=1}^{p_0} C^nN^{n-p}(N-n)(N-n-1)\cdots(N-n-1-p)  \label{group h1 interactions}  \\
  & \qquad \times \langle \psi, S_1^2 \cdots S_{n+1}^2V^{(p)}_N(x_{n+2}-x_{n+3},\cdots,x_{n+2}-x_{n+2+p})\psi \rangle  \notag \\
  & \quad + \sum_{p=1}^{p_0}\sum_{j=2}^{1+p}C^nN^{n-p}(N-n)(N-n-1)\cdots(N-n-p+j-2)(n+1)n\cdots(n+3-j)   \label{intergroup interactions}  \\
  & \qquad\times \langle \psi, S_1^2 \cdots S_{n+1}^2V^{(p)}_N(x_1-x_2,\cdots,x_1-x_{j-1},x_1-x_{n+2},\cdots,x_1-x_{n+3+p-j})\psi \rangle     \notag    \\
  & \quad + \sum_{p=1}^{p_0} C^nN^{n-p}(N-n)(n+1)n\cdots(n+1-p)   \label{group h2 interactions} \\
  & \qquad \times \langle \psi, S_1^2 \cdots S_{n+1}^2V^{(p)}_N(x_1-x_2,x_1-x_3,\cdots,x_1-x_{n+1},\cdots,x_1-x_{1+p})\psi \rangle     \notag 
\end{align}
\normalsize
We split terms as follows: \eqref{group h1 interactions}-\eqref{group h2 interactions}: we put the ``first'' $n$ particles in group $h_2$ and the ``rest'' in group $h_1$. Then the term \eqref{group h1 interactions} comes exclusively from group $h_1$ interactions; and term \eqref{group h2 interactions} is contributed purely by group $h_2$ interactions; \eqref{intergroup interactions} are mixture of inter-group and inner-group ($h_2$) interactions. We will handle each of these terms individually.

Our goal is to show that \eqref{group h1 interactions}-\eqref{group h2 interactions} are dominated by \eqref{leading terms}. Since $p_0$ is a finite number and $N$ can be arbitrarily large, thus it suffices to show the goal for a single $p$ with $1\leq p\leq p_0$.

First of all, term \eqref{group h1 interactions} is non-negative and thus can be dropped for purpose of a lower bound. To see this, note $V^{(p)}_N \geq 0$ and commutes with all derivatives $S_1, S_2, \cdots,S_{n+1}$, we have
\begin{align*}
 & \langle \psi, S_1^2 \cdots S_{n+1}^2V^{(p)}_N(x_{n+2}-x_{n+3},\cdots,x_{n+2}-x_{n+2+p})\psi \rangle    \\
 & \quad =\int d\textbf{x}_N V^{(p)}_N(x_{n+2}-x_{n+3},\cdots,x_{n+2}-x_{n+2+p})\big|(S_1\cdots S_{n+1}\psi)(\textbf{x}_N)\big|^2 \geq 0
\end{align*}
For \eqref{intergroup interactions}, the sum over $j$ consists of $p$ terms (if $1+p>n+1$, \eqref{intergroup interactions} is a sum of $n$ terms, and \eqref{group h2 interactions} vanishes). Consider the first term which corresponding to $j=2$:
\small
\begin{align}
 & \langle \psi, S_1^2 \cdots S_{n+1}^2V^{(p)}_N(x_1-x_{n+2},x_1-x_{n+3},\cdots,x_1-x_{n+1+p})\psi \rangle  \notag  \\
 & \geq \langle \psi, S_{n+1}\cdots S_2V^{(p)}_N(x_1-x_{n+2},x_1-x_{n+3},\cdots,x_1-x_{n+1+p})S_2\cdots S_{n+1}\psi \rangle    \notag  \\
 & \quad - \big|\langle \psi, S_{n+1}\cdots S_2\dot{S}_1\big(\nabla_{x_1}V^{(p)}_N(x_1-x_{n+2},x_1-x_{n+3},\cdots,x_1-x_{n+1+p})\big)S_2\cdots S_{n+1}\psi \rangle \big|     \notag  \\
 & \geq - \big|\langle \psi, S_{n+1}\cdots S_2\dot{S}_1\big(\nabla_{x_1}V^{(p)}_N(x_1-x_{n+2},x_1-x_{n+3},\cdots,x_1-x_{n+1+p})\big)S_2\cdots S_{n+1}\psi \rangle \big|   \label{second ineq j=2} \\
 & \geq -\rho\big|\langle \psi, S_{n+1}^2 \cdots S_1^2\psi \rangle\big|   \label{third ineq j=2}  \\
 & \qquad -\frac{1}{\rho}\big|\langle \psi, S_{n+1}\cdots S_2 \big|\nabla_{x_1}V^{(p)}_N(x_1-x_{n+2},x_1-x_{n+3},\cdots,x_1-x_{n+1+p})\big|^2S_2\cdots S_{n+1}\psi \rangle \big|     \notag  \\
 & \geq -\rho\big|\langle \psi, S_{n+1}^2 \cdots S_1^2\psi \rangle\big|-\frac{1}{\rho}\big\|\nabla_{x_1}V^{(p)}_N\big\|_{L^{\infty}(\mathbb{R}^{dp})}^2 \langle \psi, S_1^2\cdots S_{n+2}^2\psi \rangle      \label{last ineq j=2} \\
 & =-\rho\big|\langle \psi, S_{n+1}^2 \cdots S_1^2\psi\rangle\big| -\frac{CN^{(2pd+2)\beta}}{\rho}\langle \psi, S_1^2\cdots S_{n+2}^2\psi \rangle     \notag
\end{align}
\normalsize
which are dominated by the leading terms in \eqref{leading terms} when $\beta<\frac{1}{2pd+2}$ (which is fine since $p$ is at most $p_0$). The constant $C$ depends on $\|\nabla_{x_1}V^{(p)}\big\|_{L^{\infty}(\mathbb{R}^{dp})}^2$. Here we use the positivity of $V^{(p)}_N$ to obtain \eqref{second ineq j=2}. Note $\dot{S}_j^2=S_j^2-1<S_j^2$, $\rho>0$ in \eqref{third ineq j=2} can be chosen arbitrarily, and in \eqref{last ineq j=2} we have applied \eqref{new lemma} with $l=2$. 

For the term corresponding to $j=3$ in \eqref{intergroup interactions}:
\begin{align}
  &\langle \psi, S_1^2 \cdots S_{n+1}^2V^{(p)}_N(x_1-x_2,x_1-x_{n+2},\cdots,x_1-x_{n+p})\psi \rangle   \notag \\
  &\geq \langle\psi, S_{n+1}\cdots S_3V^{(p)}_N(x_1-x_2,x_1-x_{n+2},\cdots,x_1-x_{n+p})S_3\cdots S_{n+1}\psi\rangle    \label{j=3 positive drop} \\
  &\quad +\langle\psi, S_{n+1}\cdots S_3(\dot{S}_1^2+\dot{S}_2^2)V^{(p)}_N(x_1-x_2,x_1-x_{n+2},\cdots,x_1-x_{n+p})S_3\cdots S_{n+1}\psi\rangle  \label{j=3 can treat as previous} \\
  &\quad +\boxed{\langle\psi, S_{n+1}\cdots S_3\dot{S}_2\dot{S}_1[\dot{S}_1\dot{S}_2, V^{(p)}_N(x_1-x_2,x_1-x_{n+2},\cdots,x_1-x_{n+p})]S_3\cdots S_{n+1}\psi\rangle.}   \label{j=3 to bound new error}
\end{align}
We know \eqref{j=3 positive drop} is positive and thus can be discarded for a lower bound. \eqref{j=3 can treat as previous} can be treated as in the case $j=2$. Note that 
\begin{equation*}
 [\dot{S_1}\dot{S_2}, V^{(p)}_N]=[\dot{S_1}, V^{(p)}_N]\dot{S_2}+\dot{S_1}[\dot{S_2}, V^{(p)}_N]
\end{equation*}
Hence
\small
\begin{align*}
 \eqref{j=3 to bound new error}&\geq -\big|\langle\psi, S_{n+1}\cdots S_3\dot{S}_2\dot{S}_1[\dot{S}_1, V^{(p)}_N(x_1-x_2,x_1-x_{n+2},\cdots,x_1-x_{n+p})]\dot{S}_2S_3\cdots S_{n+1}\psi\rangle\big|  \\
 &\quad - \big|\langle\psi, S_{n+1}\cdots S_3\dot{S}_2\dot{S}_1^2[\dot{S}_2,V^{(p)}_N(x_1-x_2,x_1-x_{n+2},\cdots,x_1-x_{n+p})]S_3\cdots S_{n+1}\psi\rangle\big|  \\
 &\geq -\rho_1\big|\langle\psi, S_{n+1}^2\cdots S_2^2S_1^2\psi\big|-\frac{1}{\rho_1}\|\nabla V^{(p)}_N\|^2_{L^{\infty}(\mathbb{R}^{dp})}\langle\psi, S_1^2\cdots S_{n+2}^2\psi\rangle      \\
 &\quad -\rho_2\big|\langle\psi, S_{n+1}^2\cdots S_2^2S_1^4\psi\big|-\frac{1}{\rho_2}\|V^{(p)}_N\|^2_{W^{1,\infty}(\mathbb{R}^{dp})}\langle\psi, S_1^2\cdots S_{n+1}^2\psi\rangle
\end{align*}
\normalsize
We shall prove the estimate for general terms in \eqref{intergroup interactions} by running a one-step induction in $j$. Note that the $j$-th term $T_j$ in \eqref{intergroup interactions}, with $2\leq j \leq 1+p$, has the coefficient of order $O(N^{n-j+3})$. Assume we have the desired bound for $j$ from $2$ through $j_0$, that is
\begin{align*}
  &T_2\geq -(CN)^{n+1+(2pd+2)\beta}\langle\psi, S_1^2\cdots S_{n+2}^2\psi\rangle,     \\
  &T_3\geq -(CN)^{n+(2pd+2)\beta}\langle\psi, S_1^2\cdots S_{n+2}^2\psi\rangle,     \\
  &\cdots   \\
  &T_{j_0}\geq -(CN)^{n-j_0+3+\delta_{j_0}(\beta)}\langle\psi, S_1^2\cdots S_{n+2}^2\psi\rangle.
\end{align*}
Function $\delta_j(\beta)$ ($2\leq j\leq j_0$) take values in interval $(0,1)$, this small piece of power on $N$ is contributed by appropriate norm of $V^{(p)}_N$. By the cases we have already checked, we know that $j_0\geq 3$. Rewrite the main part of $T_{j_0+1}$ as the following 
\begin{align}
  &\langle \psi, S_1^2 \cdots S_{n+1}^2V^{(p)}_N(x_1-x_2,\cdots,x_1-x_{j_0},x_1-x_{n+2},\cdots,x_1-x_{n+3+p-j_0-1})\psi \rangle     \label{general j} \\
  &=\langle\psi, (1+\dot{S}_1^2)\cdots (1+\dot{S}_{j_0}^2) S_{j_0+1}^2\cdots S_{n+1}^2 V^{(p)}_N \psi\rangle  \notag \\
  &=\langle \psi, S_{n+1}\cdots S_{j_0+1}V^{(p)}_N S_{j_0+1}\cdots S_{n+1}\psi\rangle  \notag \\
  &\quad +\sum_{1\leq r\leq j_0}\langle \psi, \dot{S}_r^2S_{j_0+1}^2\cdots S_{n+1}^2V^{(p)}_N\psi\rangle   \notag \\
  &\quad +\sum_{1\leq r_1<r_2\leq j_0}\langle \psi, \dot{S}_{r_1}^2\dot{S}_{r_2}^2S_{j_0+1}^2\cdots S_{n+1}^2V^{(p)}_N \psi\rangle  \notag  \\
  &\quad +\cdots   \notag \\
  &\quad +\sum_{1\leq r\leq j_0}\langle \psi, \dot{S}_1^2\cdots \hat{\dot{S}}_r^2\cdots \dot{S}_{j_0}^2S_{j_0+1}^2\cdots S_{n+1}^2V^{(p)}_N \psi\rangle   \notag  \\
  &\quad +\boxed{\langle \psi, \dot{S}_1^2\dot{S}_2^2\cdots \dot{S}_{j_0}^2S_{j_0+1}^2\cdots S_{n+1}^2V^{(p)}_N \psi\rangle}   \notag
\end{align}
where a hat denotes a missing term. 
Thanks to the induction assumption we may conclude that the lower bounds of all the terms in the RHS of \eqref{general j} are controlled by the leading terms in \eqref{leading terms} except the last term. By the definition of $\dot{S}_j$, we can prove the following decomposition:
\small
\begin{equation*}
 [\dot{S}_1\cdots \dot{S}_{j-1}, V^{(p)}_N]=[\dot{S}_1, V^{(p)}_N]\dot{S}_2\cdots \dot{S}_{j-1}+\dot{S}_1[\dot{S}_2, V^{(p)}_N]\dot{S}_3\cdots \dot{S}_{j-1}+\cdots+\dot{S}_1\cdots \dot{S}_{j-2}[\dot{S}_{j-1}, V^{(p)}_N]
\end{equation*}
\normalsize
Therefore
\begin{align*}
 &\langle \psi, \dot{S}_1^2\dot{S}_2^2\cdots \dot{S}_{j_0}^2S_{j_0+1}^2\cdots S_{n+1}^2V^{(p)}_N \psi\rangle     \\
 &=\langle \psi, S_{n+1}\cdots S_{j_0+1}\dot{S}_{j_0}\cdots\dot{S}_1[\dot{S}_1\cdots\dot{S}_{j_0}, V^{(p)}_N]S_{j_0+1}\cdots S_{n+1}\psi\rangle    \\
 &=\langle \psi, S_{n+1}\cdots S_{j_0+1}\dot{S}_{j_0}\cdots\dot{S}_1([\dot{S}_1, V^{(p)}_N]\dot{S}_2\cdots \dot{S}_{j_0})S_{j_0+1}\cdots S_{n+1}\psi\rangle      \\
 &\quad +\langle \psi, S_{n+1}\cdots S_{j_0+1}\dot{S}_{j_0}\cdots\dot{S}_2\dot{S}_1(\dot{S}_1[\dot{S}_2, V^{(p)}_N]\dot{S}_3\cdots \dot{S}_{j_0})S_{j_0+1}\cdots S_{n+1}\psi\rangle   \\
 &\quad + \cdots     \\
 &\quad +\langle \psi, S_{n+1}\cdots S_{j_0+1}\dot{S}_{j_0}\cdots\dot{S}_1(\dot{S}_1 \cdots\dot{S}_{j_0-2}[\dot{S}_{j_0-1}, V^{(p)}_N]\dot{S}_{j_0})S_{j_0+1}\cdots S_{n+1}\psi\rangle    \\
 &\quad +\boxed{\langle \psi, S_{n+1}\cdots S_{j_0+1}\dot{S}_{j_0}\cdots\dot{S}_1(\dot{S}_1 \cdots\dot{S}_{j_0-1}[\dot{S}_{j_0}, V^{(p)}_N])S_{j_0+1}\cdots S_{n+1}\psi\rangle}
\end{align*}
Again, by induction assumption all terms in the RHS of the above are bounded as we need except the one in the last line. However, we can reduce it into previous case since (for $j\geq 4$):
\begin{equation}    \label{reduce number of dot S}
 \dot{S}_1\cdots \dot{S}_{j-2}[\dot{S}_{j-1}, V^{(p)}_N]=\dot{S}_1\cdots \dot{S}_{j-3}[\dot{S}_{j-1}, (\dot{S}_{j-2}V^{(p)}_N)]+\dot{S}_1\cdots \dot{S}_{j-3}[\dot{S}_{j-1}, V^{(p)}_N]\dot{S}_{j-2}                                                                                                                             
\end{equation}
Then
\begin{align*}
 &\langle \psi, S_{n+1}\cdots S_{j_0+1}\dot{S}_{j_0}\cdots\dot{S}_1(\dot{S}_1 \cdots\dot{S}_{j_0-1}[\dot{S}_{j_0}, V^{(p)}_N])S_{j_0+1}\cdots S_{n+1}\psi\rangle    \\
 &=\langle \psi, S_{n+1}\cdots S_{j_0+1}\dot{S}_{j_0}\cdots\dot{S}_1\big(\dot{S}_1\cdots \dot{S}_{j_0-2}(\nabla_{j_0}\nabla_{j_0-1}V^{(p)}_N)\big)S_{j_0+1}\cdots S_{n+1}\psi\rangle    \\
 &\quad +\langle \psi, S_{n+1}\cdots S_{j_0+1}\dot{S}_{j_0}\cdots\dot{S}_1\big(\dot{S}_1\cdots \dot{S}_{j_0-2}(i\nabla_{j_0}V^{(p)}_N)\dot{S}_{j_0-1}\big)S_{j_0+1}\cdots S_{n+1}\psi\rangle
\end{align*}
Both terms appear in the previous induction, but with one order higher derivative on $V^{(p)}_N$. Since $\|V^{(p)}_N\|_{W^{j_0-1,\infty}(\mathbb{R}^{dp})}^2 \sim N^{2(pd\beta+(j_0-1)\beta)}\|V^{(p)}\|_{W^{j_0-1,\infty}(\mathbb{R}^{dp})}^2$, we may set $\delta_{j_0}(\beta)=2	pd\beta+2(j_0-1)\beta<1$ (with $j_0 \geq 3$ since \eqref{reduce number of dot S} requires $j\geq 4$). In general, the $j$-th term $T_j$ in \eqref{intergroup interactions} has the following bound:
\begin{equation}
 T_j \geq -N^{n-j+3}N^{2(pd\beta+(j-2)\beta)}\langle\psi, S_1^2\cdots S_{n+2}^2\psi\rangle, \quad \text{for $V^{(p)}_N \in W^{j-2, \infty}$, $3\leq j \leq 1+p$}
\end{equation}
And $T_2 \geq -N^{n+1+(2dp+2)\beta}\langle\psi, S_1^2\cdots S_{n+2}^2\psi\rangle$. Admissible value for $\beta$ will not send the total power of $N$ to be greater than or equal to $n+2$. Thus for each $p\leq p_0$, $\beta$ can take values in $(0,\frac{1}{2dp+2})$, which is actually determined by the base case $j=2$.  

Finally, the term \eqref{group h2 interactions} is actually a special case in \eqref{intergroup interactions} corresponding to $j=2+p$, thus can be handled as above (the highest regularity of the potential is used here). This completes the proof.  
\end{proof}
\vspace{1mm}

\begin{lemma}    \label{lem: sobolev inequality}
 For $d\geq 1$, $m\geq 1$ and $\psi \in L_s^2(\mathbb{R}^{md})$, we have
\begin{equation}     \label{old lemma}
 \langle \psi, V(x_1,\cdots,x_m)\psi \rangle \leq \big\| V \big\|_{L_{x_1,\cdots,x_m}^r} \langle \psi, (1-\Delta_{x_1})\cdots(1-\Delta_{x_m})\psi \rangle    
\end{equation}
for any $r>2$ if $d\leq \frac{2}{m}$, and for $r \geq md$ if $d>\frac{2}{m}$.  Moreover for any $1 \leq l \leq m$, we have
\begin{equation}     \label{new lemma}
 \langle \psi, V(x_1,\cdots,x_m)\psi \rangle \leq \big\| V \big\|_{L_{x_1,\cdots,x_m}^{\infty}} \langle \psi, \prod_{j=1}^l (1-\Delta_{x_j}) \psi \rangle_{L_{x_1,\cdots,x_m}^2}   
\end{equation}
\end{lemma}

\begin{proof}
 By H\"older inequality with $\frac{1}{q}+\frac{1}{r}+\frac{1}{2}=1$ and Sobolev embedding we have
\begin{align*}
    &\quad \langle \psi, V(x_1,\cdots,x_m)\psi \rangle    \\
    &\leq \| V \|_{L_{x_1,\cdots,x_m}^r} \|\psi\|_{L_{x_1,\cdots,x_m}^2} \|\psi\|_{L_{x_1,\cdots,x_m}^q}      \\
    &\leq \| V \|_{L_{x_1,\cdots,x_m}^r} \|\psi\|_{L_{x_1,\cdots,x_m}^2} \|\psi\|_{H_{x_1,\cdots,x_m}^1}      \\
    &\leq \| V \|_{L_{x_1,\cdots,x_m}^r}  \|\psi\|_{H_{x_1,\cdots,x_m}^1}^2      \\
    &=\| V \|_{L_{x_1,\cdots,x_m}^r} \|(1+|\xi_1|^2+|\xi_2|^2+\cdots+|\xi_m|^2)^{\frac{1}{2}}\hat{\psi}\|^2_{L^2}   \\
    &\leq \| V \|_{L_{x_1,\cdots,x_m}^r} \|(1+|\xi_1|^2)^{\frac{1}{2}}(1+|\xi_2|^2)^{\frac{1}{2}}\cdots(1+|\xi_m|^2)^{\frac{1}{2}}\hat{\psi}\|^2_{L^2}   \\
    &= \| V \|_{L_{x_1,\cdots,x_m}^r} \langle \psi, (1-\Delta_{x_1})\cdots(1-\Delta_{x_m})\psi \rangle .
\end{align*}
The Sobolev embedding requires that $q$ is finite and satisfying $2\leq q \leq \frac{2md}{md-2}$, which is equivalent to $2\leq q \leq \frac{2md}{md-2}$ when $d>\frac{2}{m}$ and $2\leq q<\infty$ when $d \leq \frac{2}{m}$. From the H\"older conjugate relations $\frac{1}{r}=\frac{1}{2}-\frac{1}{q}$, we know the constrains on $r$ must be $r>2$ if $d\leq \frac{2}{m}$ and $r \geq md$ if $d>\frac{2}{m}$.

To prove \eqref{new lemma}, choose $q=2, r=\infty$ in the above proof, then replace $L^2$ norm by $H^1$ norm in the first $l$ variables to obtain:
\begin{align*}
    &\quad \langle \psi, V(x_1,\cdots,x_m)\psi \rangle    \\
    &\leq \| V \|_{L_{x_1,\cdots,x_m}^{\infty}} \|\psi\|^2_{L_{x_1,\cdots,x_m}^2}     \\
    &\leq \| V \|_{L_{x_1,\cdots,x_m}^{\infty}} \|\psi\|^2_{H_{x_1,\cdots,x_l}^1 L_{x_{l+1},\cdots,x_m}^2}      \\
    &= \| V \|_{L_{x_1,\cdots,x_m}^{\infty}} \|(1+|\xi_1|^2+|\xi_2|^2+\cdots+|\xi_l|^2)^{\frac{1}{2}}\hat{\psi}\|^2_{L^2_{\xi_1,\cdots,\xi_l}L^2_{x_{l+1},\cdots,x_m}}   \\
    &\leq \| V \|_{L_{x_1,\cdots,x_m}^{\infty}} \|(1+|\xi_1|^2)^{\frac{1}{2}}(1+|\xi_2|^2)^{\frac{1}{2}}\cdots(1+|\xi_l|^2)^{\frac{1}{2}}\hat{\psi}\|^2_{L^2_{\xi_1,\cdots,\xi_l}L^2_{x_{l+1},\cdots,x_m}}   \\
    &=\big\| V \big\|_{L_{x_1,\cdots,x_m}^{\infty}} \langle \psi, (1-\Delta_{x_1})\cdots(1-\Delta_{x_l}) \psi \rangle_{L_{x_1,\cdots,x_m}^2}
\end{align*}
Here the Fourier transform and its inverse transform of $\psi$ are taken only on the first $l$ variables with $1 \leq l \leq m$. 
\end{proof}
\vspace{2mm}

After regularization of the initial data, we have
\begin{cor}[A priori bound]     \label{cor:finite a priori energy bound}
Let $\chi$ be a bump function with support on $[0,1]$ and $\kappa>0$. Define
\begin{equation}        \label{definition of psitilde}
 \tilde\psi_N:=\frac{\chi(\frac{\kappa}{N}H_N)\psi_N}{\big\|\chi(\frac{\kappa}{N}H_N)\psi_N \big\|}
\end{equation}
Let $\tilde\psi_{N,t}=e^{-itH_N}\tilde\psi_N$ and $\tilde\gamma_{N,t}^{(k)}$ be the corresponding $k$-marginal density. Then there exists a constant $\tilde C>0$ depending on $\kappa, p_0, V^{(p)}$ for all $1\leq p\leq p_0$ but independent of $k$, $t$, and there exists an integer $N_0(k)$ for every $k\geq 1$, such that for all $N>N_0(k)$, we have
\begin{equation}     \label{finite a priori energy bound}
 Tr(1-\Delta_{x_1})\cdots(1-\Delta_{x_k})\tilde\gamma_{N,t}^{(k)} \leq \tilde C^k
\end{equation}
\end{cor}

\begin{proof}
 The proof is simple when we have Proposition \ref{prop:finite priori energy bound}, since we have
\begin{equation}
 \begin{split}
  Tr(1-\Delta_{x_1})\cdots(1-\Delta_{x_k})\tilde\gamma_{N,t}^{(k)}&=\langle \tilde\psi_{N,t}^{(k)}, S_1^2\cdots S_k^2 \tilde\psi_{N,t}^{(k)} \rangle     \\
  & \leq \frac{1}{C^k N^k} \langle \tilde\psi_{N,t}^{(k)}, (H_N+N)^k \tilde\psi_{N,t}^{(k)} \rangle    \\
  & \leq \frac{1}{C^k N^k} \langle \tilde\psi_{N,t}^{(k)}, 2^k(H_N^k+N^k) \tilde\psi_{N,t}^{(k)} \rangle    \\
  & =\frac{2^k}{C^k N^k} \langle \tilde\psi_{N,t}^{(k)}, H_N^k \tilde\psi_{N}^{(k)} \rangle+\frac{2^k}{C^k}\|\tilde\psi_{N,t}^{(k)}\|^2  \\
  & =\frac{2^k}{C^k N^k} \langle \tilde\psi_{N}^{(k)}, H_N^k \tilde\psi_{N}^{(k)} \rangle+\frac{2^k}{C^k}  \\
  & \leq  \tilde C^k
 \end{split}
\end{equation}
In the first inequality we use Proposition \ref{prop:finite priori energy bound}, and in the last inequality we use the fact that $\langle \tilde\psi_{N}^{(k)}, H_N^k \tilde\psi_{N}^{(k)} \rangle \leq C^k N^k$ with the constant $C$ depending on $\kappa$ (see Proposition 5.1 in \cite{ESY}). 
\end{proof}

\subsection{Compactness and Convergence}   \label{section: compactness and convergence}
The compactness of the $k$-particle marginal density sequence and the convergence to the infinite hierarchy are established in \cite{ESY},\cite{KSS},\cite{CPquintic}, since the arguments are essentially the same, we outline the main steps here for completeness. 

We introduce the following Banach spaces of density matrices. Denote by $\mathcal{K}_k=\mathcal{K}(L^2(\mathbb{R}^{dk}))$ the space of compact operators on $L^2(\mathbb{R}^{dk})$, equipped with the operator norm topology. And let $\mathcal{L}_k^1=\mathcal{L}^1(L^2(\mathbb{R}^{dk}))$ denote the space of trace operators on $L^2(\mathbb{R}^{dk})$ equipped with the trace class norm. Then we know (see Theorem VI.26 in \cite{RS} for details)
\begin{equation}
 \mathcal{L}_k^1=\mathcal{K}_k^*
\end{equation}
The closed unit ball in $\mathcal{L}_k^1$ is weak$^*$ compact by Banach-Alaoglu theorem, and thus is metrizable in the weak$^*$ topology. Since $\mathcal{K}_k$ is separable, there exists a sequence $\{J_i^{(k)}\}_{i \geq 1} \in \mathcal{K}_k$, with $\|J_i^{(k)}\| \leq 1$, dense in the unit ball of $\mathcal{K}_k$. Then
\begin{equation}
 \eta_k(\gamma^{(k)},\tilde\gamma^{(k)}):=\sum_{i=1}^{\infty} 2^{-i}\big|Tr J_i^{(k)}(\gamma^{(k)}-\tilde\gamma^{(k)})\big|
\end{equation}
is a metric on $\mathcal{L}_k^1$, and the induced topology by $\eta_k$ is equivalent to the weak$^*$ topology on any weak$^*$ compact subset of $\mathcal{L}_k^1$ (Theorem 3.16 in \cite{Rudin}). Therefore a uniformly bounded sequence $\gamma_N^{(k)} \in \mathcal{L}_k^1$ converges to $\gamma^{(k)} \in \mathcal{L}_k^1$ with respect to the weak$^*$ topology if and only if $\eta_k(\gamma_N^{(k)},\gamma^{(k)}) \to 0$ as $N \to \infty$. 
Now fix $T>0$, let $C([0,T], \mathcal{L}_k^1)$ be the space of $\mathcal{L}_k^1$-valued functions of $t \in [0,T]$ which are continuous with respect to the metric $\eta_k$. We definite the following metric $\hat\eta_k$ on $C([0,T], \mathcal{L}_k^1)$ for $k\in \mathbb{N}$:
\begin{equation}
 \hat\eta_k(\gamma^{(k)}(\cdot),\tilde\gamma^{(k)}(\cdot)):=\sup_{t\in [0,T]}\eta_k(\gamma^{(k)}(t),\tilde\gamma^{(k)}(t))
\end{equation}
this induces the product topology $\tau_{prod}$ on $\bigoplus_{k\in \mathbb{N}}C([0,T], \mathcal{L}_k^1)$. 

\begin{prop}  \label{prop: infinite trace bound}
 Let $\tilde\psi_N$ be defined as in \eqref{definition of psitilde}. Then the sequence of marginal densities $\tilde\Gamma_{N,t}=\{\tilde\gamma_{N,t}^{(k)}\}_{k=1}^N \in \bigoplus_{k\in \mathbb{N}}C([0,T], \mathcal{L}_k^1)$ is compact with respect to the product topology $\tau_{prod}$ generated by the metric $\hat\eta_k$. If $\Gamma_{\infty,t}=\{\gamma_{\infty,t}^{(k)}\}_{k\geq 1}$ is an arbitrary subsequential limit point, then its component $\gamma^{(k)}_{\infty,t}$ is non-negative and symmetric under permutations, and 
\begin{equation*}
 Tr \gamma^{(k)}_{\infty,t} \leq 1
\end{equation*}
for every $k \geq 1$.
\end{prop}

\begin{proof}  [Scheme of the proof] 
 The proof is completely analogous to those in \cite{KSS}, \cite{CPquintic}, \cite{ESY}. First of all, by a Cantor diagonal argument, it is sufficient to prove the compactness of $\tilde\gamma_{N,t}^{(k)}$ for some fixed $k$. Thanks to Arzel\`a-Ascoli theorem, this can be done by showing the equicontinuity of $\tilde\gamma_{N,t}^{(k)}$ with respect to the metric $\hat\eta_k$. Then it should be enough to show that for every observable $J^{(k)}$ from a dense subset of $\mathcal{K}_k$ and for every $\epsilon>0$, there exists $\delta=\delta(J^{(k)},\epsilon)$ such that
\begin{equation}   \label{equicontinuity}
  \sup_{N\geq 1} \big|Tr J^{(k)}(\tilde\gamma_{N,t}^{(k)}-\tilde\gamma_{N,s}^{(k)})\big| <\epsilon
\end{equation}
 for all $t,s\in [0,T]$ with $|t-s|\leq \delta$. \\
In order to prove \eqref{equicontinuity}, use \eqref{BBGKY} to rewrite $\tilde\gamma_{N,t}^{(k)}-\tilde\gamma_{N,s}^{(k)}$ in integral form and bound $\big|Tr J^{(k)}(\gamma_{N,t}^{(k)}-\tilde\gamma_{N,s}^{(k)})\big|$, which consists of $p+2$ terms, by the following:
\begin{equation}
 \sup_{N\geq 1}\big|Tr J^{(k)}(\tilde\gamma_{N,t}^{(k)}-\tilde\gamma_{N,s}^{(k)})\big| \leq C \big\vvvert J^{(k)}\big\vvvert |t-s|.
\end{equation}
For this purpose, \cite{CPquintic}, \cite{KSS} introduced an operator norm:
\begin{equation}
 \big\vvvert J^{(k)}\big\vvvert:=\sup_{\textbf{p}'_k} \int d\textbf{p}_k \prod_{j=1}^{k} \langle p_j\rangle \langle p'_j\rangle \big(|\hat J^{(k)}(\textbf{p}_k;\textbf{p}'_k)|+|\hat J^{(k)}(\textbf{p}'_k;\textbf{p}_k)|\big)
\end{equation}
where $\hat J^{(k)}(\textbf{p}_k;\textbf{p}'_k)$ denotes the kernel of the compact operator $J^{(k)}$ in momentum space. Then use the fact that the set of all $J^{(k)} \in \mathcal{K}_k$ with finite norm is dense in $\mathcal{K}_k$ to reach the conclusion. 
\end{proof}

From the above proposition, we know that the sequence $\tilde\Gamma_{N,t}=\{\tilde\gamma_{N,t}^{(k)}\}_{k\geq 1}$ admits at least one limit point in $\bigoplus_{k\in \mathbb{N}}C([0,T], \mathcal{L}_k^1)$ with respect to the product topology $\tau_{prod}$. 

\begin{thm}      \label{thm:converge to the infite hierarchy}
 Let $\tilde\psi_N$ be defined as in \eqref{definition of psitilde}, $\tilde\psi_{N,t}=e^{-itH_N}\tilde\psi_N$ and $\tilde\gamma_{N,t}^{(k)}$ be the corresponding $k$-marginal density. Suppose that $\Gamma_{\infty,t}=\{\tilde\gamma_{\infty,t}^{(k)}\}_{k\geq 1}$ is a limit point of $\tilde\Gamma_{N,t}=\{\tilde\gamma_{N,t}^{(k)}\}_{k=1}^N$ in $\bigoplus_{k\in \mathbb{N}}C([0,T], \mathcal{L}_k^1)$ with respect to the product topology $\tau_{prod}$. Then $\Gamma_{\infty, t}$ is a solution to the infinite hierarchy
\begin{equation}       \label{converge to the infite hierarchy}
 \gamma^{(k)}_{\infty,t}=\mathcal{U}^{(k)}(t)\gamma^{(k)}_{\infty,0}-i\sum_{p=1}^{p_0} b_p\sum_{j=1}^{k}\int_0^t ds  \mathcal{U}^{(k)}(t-s)B_{j;k+1,\dots,k+p}\gamma^{(k+p)}_{\infty,s}
\end{equation}
with initial data $\gamma_{\infty,0}^{(k)}=\Ket{\phi} \Bra{\phi}^{\otimes k}$. $\mathcal{U}^{(k)}(t)$ is the \emph{free evolution operator} defined by $\mathcal{U}^{(k)}(t)\gamma^{(k)}:=e^{it(\Delta_{\mathbf{x}_k}-\Delta_{\mathbf{x}'_k})}\gamma^{(k)}$. 
\end{thm}

\begin{proof}
 We adapt the proof in \cite{CPquintic}, \cite{KSS}. Let $k\geq 1$ be fixed. Up to a subsequence, we can assume that for every $J^{(k)} \in \mathcal{K}_k$
\begin{equation}   \label{sup trace convergence}
 \sup_{t\in [0,T]}Tr J^{(k)}(\gamma_{\infty,t}^{(k)}-\tilde\gamma_{N,t}^{(k)}) \to 0, \quad as \ \ N \to \infty
\end{equation}
It is enough to test \eqref{converge to the infite hierarchy} for observables in a dense subset of $\mathcal{K}_k$. So we choose an arbitrary  $J^{(k)} \in \mathcal{K}_k$ with $\big\vvvert J^{(k)}\big\vvvert < \infty$. We need to prove
\begin{equation}   \label{converge of initial data}
 Tr J^{(k)}\gamma_{\infty,0}^{(k)}=Tr J^{(k)}\Ket{\phi} \Bra{\phi}^{\otimes k}
\end{equation}
and
\begin{equation}   \label{converge of the solution}
 TrJ^{(k)}\gamma_{\infty,t}^{(k)}=TrJ^{(k)}\mathcal{U}^{(k)}(t)\gamma^{(k)}_{\infty,0}-i\sum_{p=1}^{p_0} b_p\sum_{j=1}^{k}\int_0^t ds TrJ^{(k)} \mathcal{U}^{(k)}(t-s)B_{j;k+1,\dots,k+p}\gamma^{(k+p)}_{\infty,s}.
\end{equation}
By the choice of $J^{(k)}$, \eqref{converge of initial data} follows from \eqref{sup trace convergence} and \eqref{trace convergence}:
\begin{equation}     \label{trace convergence}
  Tr J^{(k)}(\tilde\gamma_{N}^{(k)}-\Ket{\phi} \Bra{\phi}^{\otimes k})\to 0, \quad as \ \ N\to \infty
\end{equation}
Here $\tilde\gamma_{N,0}^{(k)}=\tilde\gamma_{N}^{(k)}$. We provide the proof of \eqref{trace convergence} in Appendix A. \\
For \eqref{converge of the solution}, we use the notation $J_t^{(k)}:=J^{(k)}\mathcal{U}^{(k)}(t)$, and go back to the BBGKY hierarchy \eqref{BBGKY} in the integral form as
\small
\begin{align}
  &TrJ^{(k)} \tilde\gamma_{N,t}^{(k)}=     \label{tr bbgky lhs} \\
  &\quad TrJ_t^{(k)}\tilde\gamma_{N,0}^{(k)}   \label{tr bbgky rhs part 1}  \\
  &\quad -\sum_{p=1}^{p_0}\frac{i}{N^{p}}\sum_{1\leq i_1<\cdots<i_{p+1}\leq k}\int_0^t ds TrJ_{t-s}^{(k)}[V^{(p)}_N(x_{i_1}-x_{i_2},\cdots,x_{i_1}-x_{i_{p+1}}),\tilde\gamma_{N,s}^{(k)}]  \label{tr bbgky rhs part 2}  \\
  &\quad -\sum_{p=1}^{p_0}\frac{i(N-k)}{N^{p}}\sum_{1\leq i_1<\cdots<i_{p} \leq k} \int_0^t ds Tr J_{t-s}^{(k)}   \label{tr bbgky rhs part 3}\\ 
  & \qquad \qquad \qquad[V^{(p)}_N(x_{i_1}-x_{i_2},\cdots,x_{i_1}-x_{i_{p}},x_{i_1}-x_{k+1}),\tilde\gamma_{N,s}^{(k+1)}]   \notag \\
  &\quad -\sum_{p=1}^{p_0}\frac{i(N-k)(N-k-1)}{N^{p}} \sum_{1\leq i_1<\cdots<i_{p-1}\leq k}\int_0^t ds TrJ_{t-s}^{(k)}   \label{tr bbgky rhs part 4}\\
  & \qquad \qquad \qquad[V^{(p)}_N(x_{i_1}-x_{i_2},\cdots,x_{i_1}-x_{i_{p-1}},x_{i_1}-x_{k+1},x_{i_2}-x_{k+2}),\tilde\gamma_{N,s}^{(k+2)}]    \notag \\
  &\quad -\cdots    \notag \\
  &\quad -\sum_{p=1}^{p_0}\frac{i(N-k)(N-k-1)\cdots(N-k-p+1)}{N^{p}}\sum_{1 \leq i_1 \leq k}\int_0^t ds TrJ_{t-s}^{(k)}    \label{tr bbgky rhs part p/2+2}\\
  & \qquad \qquad \qquad[V^{(p)}_N(x_{i_1}-x_{k+1},x_{i_1}-x_{k+2},\cdots,x_{i_1}-x_{k+p}),\tilde\gamma_{N,s}^{(k+p)}]    \notag
\end{align}
\normalsize
Let us look at the behavior of the above terms when $N \to \infty$. It is obvious that by \eqref{sup trace convergence}, \eqref{tr bbgky lhs} converges to the LHS of \eqref{converge of the solution}; and \eqref{tr bbgky rhs part 1} converges to the first term on the RHS of \eqref{converge of the solution}.  We also observed that all the terms between \eqref{tr bbgky rhs part 1} and \eqref{tr bbgky rhs part p/2+2} vanish as $N \to \infty$. Therefore, our goal is to show \eqref{tr bbgky rhs part p/2+2} converges to the last term on the RHS of \eqref{converge of the solution}. It suffices to prove that for fixed $T, k$, $J^{(k)}$ and $p$,
\begin{equation}        \label{tr bbgky last part equality}
 \begin{split}
 \sup_{t\in [s,T]}\big|TrJ_{t-s}^{(k)}&\big(V^{(p)}_N(x_{j}-x_{k+1},\cdots,x_{j}-x_{k+p})\tilde\gamma_{N,s}^{(k+p)}    \\
 & -b^{(p)}_0\delta(x_j-x_{k+1})\cdots\delta(x_j-x_{k+p}) \gamma^{(k+p)}_{\infty,s}\big)  \big|\to 0, \quad as \ \ N\to 0.
 \end{split}
\end{equation}
To bound \eqref{tr bbgky last part equality}, we  choose a non-negative probability measure $h$, i.e $h\geq 0$ and  $\int h=1$. Define $h_{\epsilon}(x)=\frac{1}{\epsilon^d}h(\frac{x}{\epsilon}), \epsilon >0$. Then
\small
\begin{align}
 &|TrJ_{t-s}^{(k)}\big(V^{(p)}_N(x_{j}-x_{k+1},\cdots,x_{j}-x_{k+p})\tilde\gamma_{N,s}^{(k+p)}-b^{(p)}_0\delta(x_j-x_{k+1})\cdots\delta(x_j-x_{k+p})\gamma^{(k+p)}_{\infty,s}\big)\big| \notag  \\
 & \leq |TrJ_{t-s}^{(k)}\big(V^{(p)}_N(x_{j}-x_{k+1},\cdots,x_{j}-x_{k+p})-b^{(p)}_0\delta(x_j-x_{k+1})\cdots\delta(x_j-x_{k+p})\big)\tilde\gamma_{N,s}^{(k+p)}\big|   \label{tr bbgky last part 1} \\
 &\quad +b^{(p)}_0\big|TrJ_{t-s}^{(k)}\big(\delta(x_j-x_{k+1})\cdots\delta(x_j-x_{k+p})-h_{\epsilon}(x_j-x_{k+1})\cdots h_{\epsilon}(x_j-x_{k+p})\big)\tilde\gamma_{N,s}^{(k+p)} \big|  \label{tr bbgky last part 2} \\
 &\quad +b^{(p)}_0\big|TrJ_{t-s}^{(k)} h_{\epsilon}(x_j-x_{k+1})\cdots h_{\epsilon}(x_j-x_{k+p})\big(\tilde\gamma_{N,s}^{(k+p)}-\gamma_{\infty,s}^{(k+p)}\big) \big|  \label{tr bbgky last part 3}  \\
 &\quad +b^{(p)}_0\big|TrJ_{t-s}^{(k)} \big(h_{\epsilon}(x_j-x_{k+1})\cdots h_{\epsilon}(x_j-x_{k+p})-\delta(x_j-x_{k+1})\cdots\delta(x_j-x_{k+p}) \big)\gamma_{\infty,s}^{(k+p)} \big|   \label{tr bbgky last part 4}
\end{align}
\normalsize
We conclude that (note $\int \frac{V^{(p)}}{b^{(p)}_0}=1$):
\begin{itemize}
 \item term \eqref{tr bbgky last part 1} converges to 0 as $N\to \infty$ by Lemma \ref{lem: lemma B} and Corollary \ref{cor:finite a priori energy bound}.
 \item term \eqref{tr bbgky last part 2} converges to 0 uniformly in $N$ as $\epsilon\to 0$ by Lemma \ref{lem: lemma B} and Corollary \ref{cor:finite a priori energy bound}.
 \item term \eqref{tr bbgky last part 3} converges to 0 as $N\to \infty$, for every fixed $\epsilon$. (see (6.8) of \cite{KSS}).
 \item term \eqref{tr bbgky last part 4} converges to 0 as $\epsilon\to 0$ by Lemma \ref{lem: lemma B} and \eqref{a priori energy bound}
\end{itemize}
Thus by taking first the limit $N\to \infty$, and then $\epsilon \to 0$, we obtain \eqref{tr bbgky last part equality}. 
\end{proof}

So far, with the uniqueness theorems in Section \ref{section: uniqueness of solutions}, we know that for each fixed $\kappa>0$ and $k\geq 1$, $\hat\eta_k(\tilde\gamma_{N,t}^{(k)}, \Ket{\phi_t} \Bra{\phi_t}^{\otimes k})\to 0$ as $N\to \infty$, Or in other words, 
\begin{equation}  \label{weak star converge}
 \tilde\gamma_{N,t}^{(k)} \to \Ket{\phi_t}\Bra{\phi_t}^{\otimes k}
\end{equation}
in the weak$^*$ topology of $\mathcal{L}_k^1$. It remains to prove that $\gamma_{N,t}^{(k)}$, the $k$-particle marginal density associated with the original wave functions $\psi_N$, converges to $\Ket{\phi_t} \Bra{\phi_t}^{\otimes k}$ as $N\to \infty$. For any given $\epsilon>0$, and compact operator $J^{(k)}\in \mathcal{K}_k$, we can find a small enough $\kappa$ such that (see \eqref{distance between psiN tildepsiN})
\begin{equation}     \label{bound by psiN tildepsiN}
 \big|TrJ^{(k)}(\gamma_{N,t}^{(k)}-\tilde\gamma_{N,t}^{(k)})\big| \leq \big\vvvert J^{(k)}\big\vvvert \|\psi_N-\tilde\psi_N\|<C\kappa^{\frac{1}{2}} \leq \frac{\epsilon}{2}
\end{equation}
uniformly in $N$. With this fixed $\kappa$, by \eqref{weak star converge}, we can pick large enough $N$ to have
\begin{equation}
 \big|TrJ^{(k)}\big(\tilde\gamma_{N,t}^{(k)}-\Ket{\phi_t}\Bra{\phi_t}^{\otimes k}\big)\big| \leq \frac{\epsilon}{2}
\end{equation}
This shows that for any given $\epsilon>0$ and $J^{(k)}\in \mathcal{K}_k$, $\exists N_0>0$ such that
\begin{equation}
  \big|TrJ^{(k)}\big(\gamma_{N,t}^{(k)}-\Ket{\phi_t}\Bra{\phi_t}^{\otimes k}\big)\big| \leq \epsilon
\end{equation}
whenever $N>N_0$. So for each $t\in [0,T]$ and every $k$, $\gamma_{N,t}^{(k)} \to \Ket{\phi_t}\Bra{\phi_t}^{\otimes k}$ in the weak$^*$ topology of $\mathcal{L}_k^1$. Since the limiting hierarchy is an orthogonal projection, the convergence in weak$^*$ topology is equivalent to the trace norm convergence. This concludes Theorem \ref{thm:main result}.

\vspace{2 mm}
\section{A Priori Energy Bounds on the Limiting Hierarchy}

This section is a preparation for proving uniqueness theorems in section \ref{section: uniqueness of solutions} using the approach introduced in \cite{KM}. In order to apply \cite{KM} we have to establish some energy bounds on the limiting hierarchy. The results are stated in theorems. From now on, we denote $S^{(k,\alpha)}$ as:
\begin{equation*}
 S^{(k,\alpha)}=\prod_{j=1}^{k}(1-\Delta_{x_j})^{\frac{\alpha}{2}}(1-\Delta_{x'_j})^{\frac{\alpha}{2}}
\end{equation*}

\begin{thm}[A priori energy bound] \label{thm:priori spacial bound on the limiting hierarchy}
 Suppose that $d \in \{1,2\}$, $0<\beta<\frac{1}{2dp_0+2}$, $p$ satisfies $1\leq p\leq p_0$. If $\Gamma_{\infty,t}=\{\gamma_{\infty,t}^{(k)}\}_{k\geq 1}$ is a limit point of the sequence $\tilde{\Gamma}_{N,t}=\{\tilde{\gamma}_{N,t}^{(k)}\}_{k=1}^N$ with respect to the product topology $\tau_{prod}$, then for every $\alpha <1$ if $d=2$, and every $\alpha \leq 1$ if $d=1$, there exists $C_{\alpha}>0$ (also has $\kappa, p_0, V^{(p)}, d$ dependence) such that 
\begin{equation}  \label{priori spacial bound}
  \Big\|S^{(k,\alpha)}B_{j;k+1,\cdots,k+p}\gamma_{\infty,t}^{(k+p)}\Big\|_{L^2(\mathbb{R}^{dk}\times \mathbb{R}^{dk})} \leq C_{\alpha}^{k+p}
  \end{equation}
for all $k \geq 1$ and all $t \in [0,T]$.
\end{thm}

\begin{proof}
Since the inequality in Corollary \ref{cor:finite a priori energy bound} is uniformly true for all large $N$, we can extract an estimate on limit points $\{\gamma_{\infty,t}^{(k)}\}_{k\geq 1}$ by taking $N \to \infty$:
\begin{equation}    \label{a priori energy bound}
  Tr(1-\Delta_{x_1})\cdots(1-\Delta_{x_k}) \gamma_{\infty,t}^{(k)} \leq C^k
\end{equation}
It is enough to prove that
\begin{equation}     \label{equivalent spacial bound}
  \Big\|S^{(k,\alpha)}B_{j;k+1,\cdots,k+p}\gamma_{\infty,t}^{(k+p)}\Big\|_{L^2(\mathbb{R}^{dk}\times \mathbb{R}^{dk})} \leq Tr(1-\Delta_{x_1})\cdots(1-\Delta_{x_{k+p}}) \gamma_{\infty,t}^{(k+p)} 
\end{equation}
Further, it should be enough to show the case that $k=1$ and $j=1$, since the proof of the argument for other values of $k$, $j$ is extremely similar.  Also, by the definition of the contraction operator $B_{j;k+1,\cdots,k+p}$, we only need to deal with $B^{+}_{j;k+1,\cdots,k+p}$ (same way works for $B^{-}_{j;k+1,\cdots,k+p}$). Switching to the Fourier space we have ($q_i$ and $q'_i$ are Fourier conjugate variables of $x_i$ and $x'_i$ respectively):
\begin{align*}
 &\big(B^{+}_{1;2,\cdots,1+p}\gamma_{\infty,t}^{(1+p)}\big)^\wedge (q_1;q'_1)   \\
 &=\int dx_1 dx'_1 e^{-ix_1\cdot q_1} e^{ix'_1\cdot q'_1} \int dx_2dx'_2 \cdots dx_{1+p}dx'_{1+p}  \\
 & \quad \times  \delta(x_1-x_2) \delta(x_1-x'_2) \delta(x_1-x_3) \delta(x_1-x'_3) \cdots \delta(x_1-x_{1+p}) \delta(x_1-x'_{1+p})     \\
 & \qquad \times \gamma_{\infty,t}^{(1+p)}(x_1,\cdots,x_{1+p};x'_1,\cdots,x'_{1+p})     \\
 &= \int dq_2 dq'_2 \cdots dq_{1+p}dq'_{1+p} \int dx_1dx'_1 \cdots dx_{1+p}dx'_{1+p}     \\
 & \quad \times e^{-ix_1\cdot q_1} e^{ix'_1\cdot q'_1} e^{iq_2(x_1-x_2)} e^{-iq'_2(x_1-x'_2)} \cdots e^{iq_{1+p}(x_1-x_{1+p})} e^{-iq'_{1+p}(x_1-x'_{1+p})}    \\
 & \qquad \times \gamma_{\infty,t}^{(1+p)}(x_1,\cdots,x_{1+p};x'_1,\cdots,x'_{1+p})     \\
 &= \int dq_2 dq'_2 \cdots dq_{1+p}dq'_{1+p} \int dx_1dx'_1 \cdots dx_{1+p}dx'_{1+p}   \\
 &  \quad \times e^{-ix_1\cdot (q_1-q_2+q'_2-\cdots -q_{1+p}+q'_{1+p})} e^{-ix_2\cdot q_2}\cdots e^{-ix_{1+p}\cdot q_{1+p}}  e^{ix'_1\cdot q'_1} e^{ix'_2\cdot q'_2}\cdots e^{ix'_{1+p} q'_{1+p}}    \\
 & \qquad \times \gamma_{\infty,t}^{(1+p)}(x_1,\cdots,x_{1+p};x'_1,\cdots,x'_{1+p}) \\
 &= \int dq_2 dq'_2 \cdots dq_{1+p}dq'_{1+p}    \\
 & \qquad \times \hat{\gamma}_{\infty,t}^{(1+p)} (q_1-q_2+q'_2-\cdots -q_{1+p}+q'_{1+p},q_2,\cdots,q_{1+p};q'_1,q'_2,\cdots,q'_{1+p})
\end{align*}
Thus
\begin{equation}
 \begin{split}
 & \big(S^{(1,\alpha)}B^{+}_{1;2,\cdots,1+p}\gamma_{\infty,t}^{(1+p)}\big)^\wedge (q_1;q'_1)   \\
 &= \langle q_1\rangle^{\alpha} \langle q'_1\rangle^{\alpha} \int dq_2 dq'_2 \cdots dq_{1+p}dq'_{1+p}    \\
 & \qquad \times \hat{\gamma}_{\infty,t}^{(1+p)} (q_1-q_2+q'_2-\cdots -q_{1+p}+q'_{1+p},q_2,\cdots,q_{1+p};q'_1,q'_2,\cdots,q'_{1+p})
\end{split}
\end{equation}
which implies
\begin{equation}    \label{L2 norm in fourier side}
 \begin{split}
 & \big\|S^{(1,\alpha)}B^{+}_{1;2,\cdots,1+p}\gamma_{\infty,t}^{(1+p)}\big\|_{{L^2(\mathbb{R}^{d}\times \mathbb{R}^{d})}}^2  \\
 &=\int dq_1dq'_1 d\tilde{q}_2d\tilde{\tilde{q}}_2 d\tilde{q}'_2d\tilde{\tilde{q}}'_2 \cdots d\tilde{q}_{1+p}d\tilde{\tilde{q}}_{1+p} d\tilde{q}'_{1+p}d\tilde{\tilde{q}}'_{1+p} \langle q_1\rangle^{2\alpha} \langle q'_1\rangle^{2\alpha}    \\
 & \quad \times \hat{\gamma}_{\infty,t}^{(1+p)} (q_1-\tilde{q}_2+\tilde{q}'_2-\cdots -\tilde{q}_{1+p}+\tilde{q}'_{1+p},\tilde{q}_2,\cdots,\tilde{q}_{1+p};\tilde{q}'_1,\tilde{q}'_2,\cdots,\tilde{q}'_{1+p})   \\
 & \quad \times \hat{\gamma}_{\infty,t}^{(1+p)} (q_1-\tilde{\tilde{q}}_2+\tilde{\tilde{q}}'_2-\cdots -\tilde{\tilde{q}}_{1+p}+\tilde{\tilde{q}}'_{1+p},\tilde{\tilde{q}}_2,\cdots,\tilde{\tilde{q}}_{1+p};\tilde{\tilde{q}}'_1,\tilde{\tilde{q}}'_2,\cdots,\tilde{\tilde{q}}'_{1+p}).   
\end{split}
\end{equation}
Note that $\gamma^{(k+p)}$ is non-negative as an operator with trace less than or equal to $1$ (see Proposition \ref{prop: infinite trace bound}). We have the following decomposition
\small
\begin{equation}    \label{gamma decomposition}
 \hat{\gamma}_{\infty,t}^{(1+p)} (q_1,q_2,\cdots,q_{1+p};q'_1,q'_2,\cdots,q'_{1+p})=\sum_j \lambda_j \psi_j(q_1,q_2,\cdots,q_{1+p}) \bar{\psi}_j(q'_1,q'_2,\cdots,q'_{1+p})
\end{equation}
\normalsize
with $\{\psi_j\}$ an orthonormal system, $\lambda_j \geq 0, \forall j$ and $\sum_j \lambda_j \leq 1$. Applying this decomposition in \eqref{L2 norm in fourier side} yields:
\begin{equation}   \label{L2 norm in fourier side after decomposition} 
 \begin{split}
  & \big\|S^{(1,\alpha)}B^{+}_{1;2,\cdots,1+p}\gamma_{\infty,t}^{(1+p)}\big\|_{{L^2(\mathbb{R}^{d}\times \mathbb{R}^{d})}}^2  \\
  &=\sum_{i,j}\lambda_i \lambda_j\int dq_1dq'_1 d\tilde{q}_2d\tilde{\tilde{q}}_2 d\tilde{q}'_2d\tilde{\tilde{q}}'_2 \cdots d\tilde{q}_{1+p}d\tilde{\tilde{q}}_{1+p} d\tilde{q}'_{1+p}d\tilde{\tilde{q}}'_{1+p} \langle q_1\rangle^{2\alpha} \langle q'_1\rangle^{2\alpha}    \\
  & \quad \times \psi_i(q_1-\tilde{q}_2+\tilde{q}'_2-\cdots -\tilde{q}_{1+p}+\tilde{q}'_{1+p},\tilde{q}_2,\cdots,\tilde{q}_{1+p}) \bar{\psi}_i(\tilde{q}'_1,\tilde{q}'_2,\cdots,\tilde{q}'_{1+p})   \\
  & \quad \times \psi_j(q_1-\tilde{\tilde{q}}_2+\tilde{\tilde{q}}'_2-\cdots -\tilde{\tilde{q}}_{1+p}+\tilde{\tilde{q}}'_{1+p},\tilde{\tilde{q}}_2,\cdots,\tilde{\tilde{q}}_{1+p}) \bar{\psi}_j(\tilde{\tilde{q}}'_1,\tilde{\tilde{q}}'_2,\cdots,\tilde{\tilde{q}}'_{1+p})   
 \end{split}
\end{equation}
It is obviously true that
\begin{equation*}
 \langle q_1\rangle^{\alpha} \leq C\Big(\langle q_1-\tilde{q}_2+\tilde{q}'_2-\cdots -\tilde{q}_{1+p}+\tilde{q}'_{1+p}\rangle^{\alpha}+\langle \tilde{q}_2 \rangle^{\alpha}+\langle \tilde{q}'_2 \rangle^{\alpha} +\cdots+\langle \tilde{q}_{1+p}\rangle^{\alpha}+ \langle \tilde{q}'_{1+p}\rangle^{\alpha}\Big)
\end{equation*}
and
\begin{equation*}
  \langle q_1\rangle^{\alpha} \leq C\Big(\langle q_1-\tilde{\tilde{q}}_2+\tilde{\tilde{q}}'_2-\cdots -\tilde{\tilde{q}}_{1+p}+\tilde{\tilde{q}}'_{1+p}\rangle^{\alpha}+\langle \tilde{\tilde{q}}_2 \rangle^{\alpha}+\langle \tilde{\tilde{q}}'_2 \rangle^{\alpha} +\cdots+\langle \tilde{\tilde{q}}_{1+p}\rangle^{\alpha}+ \langle \tilde{\tilde{q}}'_{1+p}\rangle^{\alpha}\Big)
\end{equation*}
multiplying them together we have the following estimate:
\begin{equation}
 \langle q_1\rangle^{2\alpha} \leq C\Big(\langle q_1-\tilde{q}_2+\tilde{q}'_2-\cdots -\tilde{q}_{1+p}+\tilde{q}'_{1+p}\rangle^{\alpha}+\langle \tilde{q}_2 \rangle^{\alpha}+\langle \tilde{q}'_2 \rangle^{\alpha} +\cdots+\langle \tilde{q}_{1+p}\rangle^{\alpha}+ \langle \tilde{q}'_{1+p}\rangle^{\alpha}\Big)
\end{equation}
\begin{equation*}
  \qquad \qquad \times \Big(\langle q_1-\tilde{\tilde{q}}_2+\tilde{\tilde{q}}'_2-\cdots -\tilde{\tilde{q}}_{1+p}+\tilde{\tilde{q}}'_{1+p}\rangle^{\alpha}+\langle \tilde{\tilde{q}}_2 \rangle^{\alpha}+\langle \tilde{\tilde{q}}'_2 \rangle^{\alpha} +\cdots+\langle \tilde{\tilde{q}}_{1+p}\rangle^{\alpha}+ \langle \tilde{\tilde{q}}'_{1+p}\rangle^{\alpha}\Big).
\end{equation*}
After substituting the above bound in \eqref{L2 norm in fourier side after decomposition}, we will obtain $(2p+1)^2$ contributed terms. However, it is enough to illustrate how to control just one of them, since the remaining cases are essentially the same. For instance, the first contribution comes from the replacement of the factor $\langle q_1\rangle^{2\alpha}$ on the RHS of \eqref{L2 norm in fourier side after decomposition} by $\langle q_1-\tilde{q}_2+\tilde{q}'_2-\cdots -\tilde{q}_{1+p}+\tilde{q}'_{1+p}\rangle^{\alpha} \langle q_1-\tilde{\tilde{q}}_2+\tilde{\tilde{q}}'_2-\cdots -\tilde{\tilde{q}}_{1+p}+\tilde{\tilde{q}}'_{1+p}\rangle^{\alpha}$. Using Schwartz inequality we find
\begin{equation}
 \begin{split}
  & \int dq_1dq'_1 d\tilde{q}_2d\tilde{\tilde{q}}_2 d\tilde{q}'_2d\tilde{\tilde{q}}'_2 \cdots d\tilde{q}_{1+p}d\tilde{\tilde{q}}_{1+p} d\tilde{q}'_{1+p}d\tilde{\tilde{q}}'_{1+p}   \\
  & \quad \times \langle q_1-\tilde{q}_2+\tilde{q}'_2-\cdots -\tilde{q}_{1+p}+\tilde{q}'_{1+p}\rangle^{\alpha} \langle q_1-\tilde{\tilde{q}}_2+\tilde{\tilde{q}}'_2-\cdots -\tilde{\tilde{q}}_{1+p}+\tilde{\tilde{q}}'_{1+p}\rangle^{\alpha} \langle q'_1\rangle^{2\alpha}    \\
  & \quad \times \psi_i(q_1-\tilde{q}_2+\tilde{q}'_2-\cdots -\tilde{q}_{1+p}+\tilde{q}'_{1+p},\tilde{q}_2,\cdots,\tilde{q}_{1+p}) \bar{\psi}_i(\tilde{q}'_1,\tilde{q}'_2,\cdots,\tilde{q}'_{1+p})  \\
  & \quad \times \psi_j(q_1-\tilde{\tilde{q}}_2+\tilde{\tilde{q}}'_2-\cdots -\tilde{\tilde{q}}_{1+p}+\tilde{\tilde{q}}'_{1+p},\tilde{\tilde{q}}_2,\cdots,\tilde{\tilde{q}}_{1+p}) \bar{\psi}_j(\tilde{\tilde{q}}'_1,\tilde{\tilde{q}}'_2,\cdots,\tilde{\tilde{q}}'_{1+p})   \\
  &  \leq A+B
 \end{split}
\end{equation}
where
\small
\begin{equation*}
 \begin{split}
  A=&\int dq_1dq'_1 d\tilde{q}_2d\tilde{\tilde{q}}_2 d\tilde{q}'_2d\tilde{\tilde{q}}'_2 \cdots d\tilde{q}_{1+\frac{p}  {2}}d\tilde{\tilde{q}}_{1+p} d\tilde{q}'_{1+p}d\tilde{\tilde{q}}'_{1+p} \langle q'_1\rangle^{2\alpha}\\
  & \times \frac{\langle q_1-\tilde{q}_2+\tilde{q}'_2-\cdots -\tilde{q}_{1+p}+\tilde{q}'_{1+p}\rangle^{2} \langle\tilde{q}_2\rangle^2 \langle \tilde{q}_3\rangle^2 \cdots \langle \tilde{q}_{1+p}\rangle^2 \langle \tilde{\tilde{q}}'_2\rangle^2 \langle \tilde{\tilde{q}}'_3\rangle^2 \cdots \langle \tilde{\tilde{q}}'_{1+p}\rangle^2}{\langle q_1-\tilde{\tilde{q}}_2+\tilde{\tilde{q}}'_2-\cdots -\tilde{\tilde{q}}_{1+p}+\tilde{\tilde{q}}'_{1+p}\rangle^{2-2\alpha} \langle\tilde{\tilde{q}}_2\rangle^2 \langle \tilde{\tilde{q}}_3\rangle^2 \cdots \langle \tilde{\tilde{q}}_{1+p}\rangle^2 \langle \tilde{q}'_2\rangle^2 \langle \tilde{q}'_3\rangle^2 \cdots \langle \tilde{q}'_{1+p}\rangle^2}   \\
  & \times \big|\psi_i(q_1-\tilde{q}_2+\tilde{q}'_2-\cdots -\tilde{q}_{1+p}+\tilde{q}'_{1+p},\tilde{q}_2,\cdots,\tilde{q}_{1+p})\big|^2 \big|\psi_j(\tilde{\tilde{q}}'_1,\tilde{\tilde{q}}'_2,\cdots,\tilde{\tilde{q}}'_{1+p})\big|^2 , \\
\end{split}
\end{equation*}
and
\begin{equation*}
 \begin{split}
  B=&\int dq_1dq'_1 d\tilde{q}_2d\tilde{\tilde{q}}_2 d\tilde{q}'_2d\tilde{\tilde{q}}'_2 \cdots d\tilde{q}_{1+\frac{p}  {2}}d\tilde{\tilde{q}}_{1+p} d\tilde{q}'_{1+p}d\tilde{\tilde{q}}'_{1+p} \langle q'_1\rangle^{2\alpha}\\
  & \times \frac{\langle q_1-\tilde{\tilde{q}}_2+\tilde{\tilde{q}}'_2-\cdots -\tilde{\tilde{q}}_{1+p}+\tilde{\tilde{q}}'_{1+p}\rangle^{2} \langle\tilde{\tilde{q}}_2\rangle^2 \langle \tilde{\tilde{q}}_3\rangle^2 \cdots \langle \tilde{\tilde{q}}_{1+p}\rangle^2 \langle \tilde{q}'_2\rangle^2 \langle \tilde{q}'_3\rangle^2 \cdots \langle \tilde{q}'_{1+p}\rangle^2}{\langle q_1-\tilde{q}_2+\tilde{q}'_2-\cdots -\tilde{q}_{1+p}+\tilde{q}'_{1+p}\rangle^{2-2\alpha} \langle\tilde{q}_2\rangle^2 \langle \tilde{q}_3\rangle^2 \cdots \langle \tilde{q}_{1+p}\rangle^2 \langle \tilde{\tilde{q}}'_2\rangle^2 \langle \tilde{\tilde{q}}'_3\rangle^2 \cdots \langle \tilde{\tilde{q}}'_{1+p}\rangle^2}   \\
  & \times \big|\psi_j(q_1-\tilde{\tilde{q}}_2+\tilde{\tilde{q}}'_2-\cdots -\tilde{\tilde{q}}_{1+p}+\tilde{\tilde{q}}'_{1+p},\tilde{\tilde{q}}_2,\cdots,\tilde{\tilde{q}}_{1+p})\big|^2 \big|\psi_i(\tilde{q}'_1,\tilde{q}'_2,\cdots,\tilde{q}'_{1+p})\big|^2  \\
\end{split}
\end{equation*}
\normalsize
Now let us focus on $A$ below, because $B$ can be handled similarly. Performing integration on $\tilde{\tilde{q}}_2, \tilde{\tilde{q}}_3,\cdots, \tilde{\tilde{q}}_{1+p}$ we obtain:
\small
\begin{equation}      \label{estimation of A}
 \begin{split}
  A\leq &C\int dq_1dq'_1 d\tilde{q}_2 d\tilde{q}'_2d\tilde{\tilde{q}}'_2 \cdots d\tilde{q}_{1+p} d\tilde{q}'_{1+p}d\tilde{\tilde{q}}'_{1+p} \langle q'_1\rangle^{2\alpha}\\
  & \times \frac{\langle q_1-\tilde{q}_2+\tilde{q}'_2-\cdots -\tilde{q}_{1+p}+\tilde{q}'_{1+p}\rangle^{2} \langle\tilde{q}_2\rangle^2 \langle \tilde{q}_3\rangle^2 \cdots \langle \tilde{q}_{1+p}\rangle^2 \langle \tilde{\tilde{q}}'_2\rangle^2 \langle \tilde{\tilde{q}}'_3\rangle^2 \cdots \langle \tilde{\tilde{q}}'_{1+p}\rangle^2}{\langle q_1+\tilde{\tilde{q}}'_2+\cdots+\tilde{\tilde{q}}'_{1+p}\rangle^{2-2\alpha} \langle \tilde{q}'_2\rangle^2 \langle \tilde{q}'_3\rangle^2 \cdots \langle \tilde{q}'_{1+p}\rangle^2}   \\
  & \times \big|\psi_i(q_1-\tilde{q}_2+\tilde{q}'_2-\cdots -\tilde{q}_{1+p}+\tilde{q}'_{1+p},\tilde{q}_2,\cdots,\tilde{q}_{1+p})\big|^2 \big|\psi_j(\tilde{\tilde{q}}'_1,\tilde{\tilde{q}}'_2,\cdots,\tilde{\tilde{q}}'_{1+p})\big|^2   
 \end{split}
\end{equation}
\normalsize
where we used the bound
\begin{equation}     \label{analysis inequality}
 \int_{\mathbb{R}^d} \frac{dy}{\langle W-y \rangle^{2-2\alpha} \langle y \rangle^2} \leq \frac{C}{\langle W \rangle^{2-2\alpha}}
\end{equation}
for all $\alpha <1$ when $d=2$ and for $\alpha \leq 1$ when $d=1$.   \\
Let $\breve{q}_1=q_1-\tilde{q}_2+\tilde{q}'_2-\cdots -\tilde{q}_{1+p}+\tilde{q}'_{1+p}$ in \eqref{estimation of A}. Since $\alpha \leq 1$, we can replace $\langle q'_1\rangle^{2\alpha}$ with $\langle q'_1\rangle^2$ for an upper bound:
\small
\begin{equation}
 \begin{split}
  A \leq &C\int d\breve{q}_1dq'_1 d\tilde{q}_2 d\tilde{q}'_2d\tilde{\tilde{q}}'_2 \cdots d\tilde{q}_{1+p} d\tilde{q}'_{1+p}d\tilde{\tilde{q}}'_{1+p}  \\
  & \times \frac{\langle \breve{q}_1\rangle^2 \langle\tilde{q}_2\rangle^2 \langle \tilde{q}_3\rangle^2 \cdots \langle \tilde{q}_{1+p}\rangle^2  \langle q'_1\rangle^2 \langle \tilde{\tilde{q}}'_2\rangle^2 \langle \tilde{\tilde{q}}'_3\rangle^2 \cdots \langle \tilde{\tilde{q}}'_{1+p}\rangle^2}{\langle \breve{q}_1+\tilde{q}_2-\tilde{q}'_2+\cdots+\tilde{q}_{1+p}-\tilde{q}'_{1+p}+\tilde{\tilde{q}}'_2+\cdots+\tilde{\tilde{q}}'_{1+p}\rangle^{2-2\alpha} \langle \tilde{q}'_2\rangle^2 \langle \tilde{q}'_3\rangle^2 \cdots \langle \tilde{q}'_{1+p}\rangle^2}   \\
  & \times \big|\psi_i(\breve{q}_1,\tilde{q}_2,\cdots,\tilde{q}_{1+p})\big|^2 \big|\psi_j(\tilde{\tilde{q}}'_1,\tilde{\tilde{q}}'_2,\cdots,\tilde{\tilde{q}}'_{1+p})\big|^2    \\
  \leq &C C'_{\alpha} \int d\breve{q}_1 d\tilde{q}_2 d\tilde{q}_3 \cdots d\tilde{q}_{1+p} \langle \breve{q}_1\rangle^2 \langle\tilde{q}_2\rangle^2 \langle \tilde{q}_3\rangle^2 \cdots \langle \tilde{q}_{1+p}\rangle^2 \big|\psi_i(\breve{q}_1,\tilde{q}_2,\cdots,\tilde{q}_{1+p})\big|^2  \\
  & \times \int dq'_1 d\tilde{\tilde{q}}'_2 d\tilde{\tilde{q}}'_3 \cdots d\tilde{\tilde{q}}'_{1+p} \langle q'_1\rangle^2 \langle \tilde{\tilde{q}}'_2\rangle^2 \langle \tilde{\tilde{q}}'_3\rangle^2 \cdots \langle \tilde{\tilde{q}}'_{1+p}\rangle^2 \big|\psi_j(\tilde{\tilde{q}}'_1,\tilde{\tilde{q}}'_2,\cdots,\tilde{\tilde{q}}'_{1+p})\big|^2    \\
 \end{split}
\end{equation}
\normalsize
where $C'_{\alpha}$ is defined as
\begin{equation}     \label{alpha upper bound}
  C'_{\alpha}=\sup_{W\in \mathbb{R}^d} \int \frac{dx_1dx_2\cdots dx_{p}}{\langle W-x_1-x_2-\cdots-x_{p} \rangle^{2-2\alpha} \langle x_1\rangle^2 \langle x_2\rangle^2 \cdots \langle x_{p}\rangle^2}
\end{equation}
For all $\alpha \leq 1$ if $d=1$ and $\alpha <1$ if $d=2$, $C'_{\alpha}<\infty$. 
Now we have the control on one of the $(2p+1)^2$ pieces, and the remaining pieces can be bounded the same way. Recalling \eqref{gamma decomposition} and \eqref{L2 norm in fourier side after decomposition}, we can conclude that
\begin{equation*}
 \begin{split}
  & \big\|S^{(1,\alpha)}B^{+}_{1;2,\cdots,1+p}\gamma_{\infty,t}^{(1+p)}\big\|_{{L^2(\mathbb{R}^{d}\times \mathbb{R}^{d})}}^2  \\
  &\leq C_{\alpha}\sum_{i,j}\lambda_i\lambda_j \int d\breve{q}_1 d\tilde{q}_2 d\tilde{q}_3 \cdots d\tilde{q}_{1+p} \langle \breve{q}_1\rangle^2 \langle\tilde{q}_2\rangle^2 \langle \tilde{q}_3\rangle^2 \cdots \langle \tilde{q}_{1+p}\rangle^2 \big|\psi_i(\breve{q}_1,\tilde{q}_2,\cdots,\tilde{q}_{1+p})\big|^2  \\
  & \qquad \times \int dq'_1 d\tilde{\tilde{q}}'_2 d\tilde{\tilde{q}}'_3 \cdots d\tilde{\tilde{q}}'_{1+p} \langle q'_1\rangle^2 \langle \tilde{\tilde{q}}'_2\rangle^2 \langle \tilde{\tilde{q}}'_3\rangle^2 \cdots \langle \tilde{\tilde{q}}'_{1+p}\rangle^2 \big|\psi_j(\tilde{\tilde{q}}'_1,\tilde{\tilde{q}}'_2,\cdots,\tilde{\tilde{q}}'_{1+p})\big|^2    \\
  &\leq C_{\alpha}\Big(\int d\breve{q}_1 d\tilde{q}_2 d\tilde{q}_3 \cdots d\tilde{q}_{1+p} \langle \breve{q}_1\rangle^2 \langle\tilde{q}_2\rangle^2 \langle \tilde{q}_3\rangle^2 \cdots \langle \tilde{q}_{1+p}\rangle^2    \\
  & \qquad \qquad \times \big|\hat{\gamma}_{\infty,t}^{(1+p)}(\breve{q}_1,\tilde{q}_2,\cdots,\tilde{q}_{1+p};\breve{q}_1,\tilde{q}_2,\cdots,\tilde{q}_{1+p})\big|^2 \Big)^2     \\
  &=C_{\alpha} \Big(Tr(1-\Delta_{x_1})\cdots(1-\Delta_{x_{k+p}}) \gamma_{\infty,t}^{(k+p)}\Big)^2
 \end{split}
\end{equation*}
Therefore \eqref{equivalent spacial bound} follows.
\end{proof}

\begin{thm}  \label{thm:priori induction estimate on the limiting hierarchy}
 Suppose that $d \geq 1$. If $\Gamma_{\infty,t}=\{\gamma_{\infty,t}^{(k)}\}_{k\geq 1}$ is a limit point of the sequence $\tilde{\Gamma}_{N,t}=\{\tilde{\gamma}_{N,t}^{(k)}\}_{k=1}^N$ with respect to the product topology $\tau_{prod}$, then, for every $\alpha >\frac{d}{2}$ there exists a constant $C_{\alpha}$ (also depends on $p_0,d$) such that the estimate
\small
\begin{equation}  \label{priori induction estimate}
    \Big\|S^{(k,\alpha)}B_{j;k+1,\cdots,k+p}\gamma_{\infty,t}^{(k+p)}\Big\|_{L^2(\mathbb{R}^{dk}\times \mathbb{R}^{dk})} \leq  C_{\alpha} \Big\|S^{(k+p,\alpha)}\gamma_{\infty,t}^{(k+p)}\Big\|_{L^2\big(\mathbb{R}^{d(k+p)}\times \mathbb{R}^{d(k+p)}\big)}
  \end{equation}
\normalsize 
holds.
\end{thm}

\begin{proof}
 We will work on the Fourier side of spacial coordinates. Let $(\textbf{u}_k,\textbf{u}'_k)$, $\textbf{q}:=(q_1,q_2,\cdots,q_{p})$ and $\textbf{q}':=(q'_1,q'_2,\cdots,q'_{p})$ be the Fourier conjugate variables corresponding to $(\textbf{x}_k,\textbf{x}'_k)$, $(x_{k+1},x_{k+2},\cdots,x_{k+p})$ and $(x'_{k+1}, x'_{k+2}, \cdots,x'_{k+p})$ respectively.  \\
Assume $j=1$ in $B_{j;k+1,\cdots,k+p}$ without loss of generality, and we replace the contraction operator by its positive part $B^{+}_{j;k+1,\cdots,k+p}$ here, since the negative part is similar. By Plancherel's theorem
\begin{equation}      \label{square of priori induction estimate}
 \begin{split}
  &\Big\|S^{(k,\alpha)}B_{1;k+1,\cdots,k+p}\gamma_{\infty,t}^{(k+p)}\Big\|_{L^2(\mathbb{R}^{dk}\times \mathbb{R}^{dk})}^2    \\
  & = \int d\textbf{u}_kd\textbf{u}'_k \prod_{j=1}^k \langle u_j \rangle^{2\alpha} \langle u'_j \rangle^{2\alpha}   \\
  & \quad \times \big(\int d\textbf{q}d\textbf{q}'\hat{\gamma}^{(k+p)}(u_1+q_1+\cdots+q_{p}-q'_1-\cdots-q'_{p},u_2,\cdots,u_k,\textbf{q};\textbf{u}'_k,\textbf{q}')\big)^2   
 \end{split}
\end{equation}
Cauchy-Schwartz inequality gives us an upper bound
\begin{equation}
 \begin{split}
  \eqref{square of priori induction estimate} \leq & \int d\textbf{u}_k d\textbf{u}'_k F_{\alpha}(\textbf{u}_k,\textbf{u}'_k) \prod_{j=2}^k \langle u_j \rangle^{2\alpha} \prod_{j=1}^k \langle u'_j \rangle^{2\alpha} \int d\textbf{q}d\textbf{q}'    \\
  &  \times \langle u_1+q_1+\cdots+q_{p}-q'_1-\cdots-q'_{p} \rangle^{2\alpha} \langle q_1\rangle^{2\alpha} \cdots \langle q_{p}\rangle^{2\alpha} \langle q'_1\rangle^{2\alpha} \cdots \langle q'_{p}\rangle^{2\alpha}  \\
  &  \times |\hat{\gamma}^{(k+p)}(u_1+q_1+\cdots+q_{p}-q'_1-\cdots-q'_{p},u_2,\cdots,u_k,\textbf{q};\textbf{u}'_k,\textbf{q}')|^2    \\
  \leq & \sup_{\textbf{u}_k,\textbf{u}'_k} F_{\alpha}(\textbf{u}_k,\textbf{u}'_k) \times \Big\|S^{(k+p,\alpha)}\gamma_{\infty,t}^{(k+p)}\Big\|_{L^2\big(\mathbb{R}^{d(k+p)}\times \mathbb{R}^{d(k+p)}\big)}
 \end{split}
\end{equation}
Where 
\small
\begin{equation}  \label{definition of Falpha} 
  F_{\alpha}(\textbf{u}_k,\textbf{u}'_k):=\int\frac{\langle u_1\rangle^{2\alpha} d\textbf{q}d\textbf{q}' }{\langle u_1+q_1+\cdots+q_{p}-q'_1-\cdots-q'_{p} \rangle^{2\alpha} \langle q_1\rangle^{2\alpha} \cdots \langle q_{p}\rangle^{2\alpha} \langle q'_1\rangle^{2\alpha} \cdots \langle q'_{p}\rangle^{2\alpha}}.
\end{equation}
\normalsize
Because of our simplifications at the beginning (specification of $j$ and neglect of the negative part of $B_{j;k+1,\cdots,k+p}$), function $F_{\alpha}(\textbf{u}_k,\textbf{u}'_k)$ only depends on $u_1$ here.  
From
\small
\begin{equation*} 
 \langle u_1 \rangle^{2\alpha} \leq C (\langle u_1+q_1+\cdots+q_{p}-q'_1-\cdots-q'_{p} \rangle^{2\alpha}+\langle q_1\rangle^{2\alpha}+\cdots+\langle q_{p}\rangle^{2\alpha}+\langle q'_1\rangle^{2\alpha}+\cdots+\langle q'_{p}\rangle^{2\alpha})
\end{equation*}
\normalsize
we shift some of the momentum variables to obtain
\begin{equation}  \label{upper bound of Falpha}
  \sup_{\textbf{u}_k,\textbf{u}'_k} F_{\alpha}(\textbf{u}_k,\textbf{u}'_k) \leq C\int \frac{d\textbf{q}d\textbf{q}'}{\langle q_1\rangle^{2\alpha} \cdots \langle q_{p}\rangle^{2\alpha} \langle q'_1\rangle^{2\alpha} \cdots \langle q'_{p}\rangle^{2\alpha}}
\end{equation}
The RHS of \eqref{upper bound of Falpha} is always finite when $\alpha > \frac{d}{2}$. This proves the theorem.
\end{proof}

The above estimate \eqref{priori induction estimate} requires that $\alpha > \frac{d}{2}$. Recall the conditions on $\alpha$ ($\alpha<1$ if $d>2$, $\alpha \leq 1$ if $d=1$) in Theorem \ref{thm:priori spacial bound on the limiting hierarchy}. If we want to use both theorems, only $d=1$ gives us a nonempty intersection of the two conditions, so we cannot afford this when $d>1$. However we need a bound like \eqref{priori induction estimate} for iterative computations in the proof of uniqueness of the limiting hierarchy. We build such a bound in next section.

\vspace{2 mm}
\section{Bounds on the Free Evolution of Infinite Hierarchy}

In this section, we consider the case when the interactions among particles are neglected ($b_0^{(p)}=0$). We will prove a Strichartz estimate that can be used when dealing with recursive Duhamel expansion terms. The approach we followed in this part is exhibited in \cite{KM},\cite{KSS},\cite{CPquintic}. From now on, we will use $\gamma^{(k)}(t,\textbf{x}_{k+p},\textbf{x}'_{k+p})$ to replace $\gamma_{\infty,t}^{(k)}(t,\textbf{x}_{k+p},\textbf{x}'_{k+p})$ for convenience. 

\begin{thm}[Free evolving bound] \label{thm:free evolving bound theorem}
 Assume that $d=2$ and $1-\frac{1}{2(2p_0-1)}<\alpha<1$, $p$ satisfies $1\leq p\leq p_0$.  Let $\gamma^{(k+p)}$ denote the solution of
  \begin{equation} \label{homogeneous equation}
   i\partial_t\gamma^{(k+p)}(t,{\bf x}_{k+p},{\bf x}'_{k+p})+(\Delta_{{\bf x}_{k+p}}-\Delta_{{\bf x}'_{k+p}})\gamma^{(k+p)}(t,{\bf x}_{k+p},{\bf x}'_{k+p})=0
  \end{equation}
 with initial condition
  \begin{equation}
   \gamma^{(k+p)}(0,\cdot)=\gamma_0^{(k+p)} \in \mathcal{H}^{\alpha}
  \end{equation}
where $\mathcal{H}^{\alpha}$ denotes the space of density matrices with finite Hilbert-Schmidt type Sobolev norms:
\begin{equation}
  \mathcal{H}^{\alpha}=\{\gamma^{(k)}: \|S^{(k,\alpha)}\gamma^{(k)}\|_{L^2(\mathbb{R}^{dk} \times \mathbb{R}^{dk})} <\infty \}
\end{equation}
Then, there exists a constant $C=C_{\alpha}$ (also depends on $p_0$) but independent of $j,k$ such that
\begin{equation}  \label{free evolving bound}
  \begin{split}
   & \Big\|S^{(k,\alpha)}B_{j;k+1,\cdots,k+p}\gamma^{(k+p)}\Big\|_{L_{t,{\bf x}_k,{\bf x}'_k}^2(\mathbb{R}\times\mathbb{R}^{2k}\times \mathbb{R}^{2k})}    \\
   & \leq  C_{\alpha} \Big\|S^{(k+p,\alpha)}\gamma_0^{(k+p)}\Big\|_{L_{{\bf x}_{k+p},{\bf x}'_{k+p}}^2(\mathbb{R}^{2(k+p)}\times \mathbb{R}^{2(k+p)})}
  \end{split}
\end{equation}
holds.
\end{thm}

\begin{proof}
Following \cite{CPquintic}, since the two norms are both $L^2$ norms, by Plancherel's theorem, it suffices to prove the estimate \eqref{free evolving bound} for the Fourier transform of functions in both sides. As before, by definition of the contraction operator in \eqref{defintion of Bj+} and \eqref{defintion of Bj-}, we only need to estimate the term in $B^{+}_{j;k+1,\cdots,k+p}$; the term in $B^{-}_{j;k+1,\cdots,k+p}$ can be treated in the same manner. Let $(\tau,\textbf{u}_k,\textbf{u}'_k)$, $\textbf{q}:=(q_1,q_2,\cdots,q_{p})$ and $\textbf{q}':=(q'_1,q'_2,\cdots,q'_{p})$ be the Fourier conjugate variables corresponding to $(t,\textbf{x}_k,\textbf{x}'_k)$, $(x_{k+1},x_{k+2},\cdots,x_{k+p})$ and $(x'_{k+1},x'_{k+2},\cdots,x'_{k+p})$ respectively. For convenience, let
\small
\begin{equation}
 \delta(\cdots):=\delta(\tau+(u_1+q_1+q_2+\cdots+q_{p}-q'_1-q'_2-\cdots-q'_{p})^2+\sum_{j=2}^{k} u^2_j +|\textbf{q}|^2-|\textbf{u}'_k|^2-|\textbf{q}'|^2)
\end{equation}
\normalsize
We may also assume that $j=1$ in $B_{j;k+1,\cdots,k+p}$ without loss of generality. Then
\small
\begin{equation} \label{free evolving bound lhs}
 \begin{split}
  &  \Big\|S^{(k,\alpha)}B_{1;k+1,\cdots,k+p}\gamma^{(k+p)}\Big\|^2_{L_{t,{\bf x}_k,{\bf x}'_k}^2(\mathbb{R}\times\mathbb{R}^{2k}\times \mathbb{R}^{2k})}    \\
  &  =\int_{\mathbb{R}} d\tau \int d\textbf{u}_kd\textbf{u}'_k \prod_{j=1}^k \langle u_j \rangle^{2\alpha} \langle u'_j \rangle^{2\alpha}   \\
  &  \times \big(\int d\textbf{q}d\textbf{q}' \delta(\cdots) \hat{\gamma}^{(k+p)}(\tau,u_1+q_1+\cdots+q_{p}-q'_1-\cdots-q'_{p},u_2,\cdots,u_k,\textbf{q};\textbf{u}'_k,\textbf{q}')\big)^2   
 \end{split}
\end{equation}
\normalsize
Applying the Cauchy-Schwarz inequality, the above integral is further bounded by:
\begin{equation*}
 \begin{split}
  \eqref{free evolving bound lhs} \leq & \int_{\mathbb{R}} d\tau \int d\textbf{u}_k d\textbf{u}'_k I_{\alpha,p}(\tau,\textbf{u}_k,\textbf{u}'_k)\prod_{j=2}^k \langle u_j \rangle^{2\alpha} \prod_{j=1}^k \langle u'_j \rangle^{2\alpha} \int d\textbf{q}d\textbf{q}' \delta(\cdots) \\
  &  \times \langle u_1+q_1+\cdots+q_{p}-q'_1-\cdots-q'_{p} \rangle^{2\alpha} \langle q_1\rangle^{2\alpha} \cdots \langle q_{p}\rangle^{2\alpha} \langle q'_1\rangle^{2\alpha} \cdots \langle q'_{p}\rangle^{2\alpha}  \\
  &  \times |\hat{\gamma}^{(k+p)}(\tau,u_1+q_1+\cdots+q_{p}-q'_1-\cdots-q'_{p},u_2,\cdots,u_k,\textbf{q};\textbf{u}'_k,\textbf{q}')|^2    \\
  \leq & \Big\|S^{(k+p,\alpha)}\gamma_0^{(k+p)}\Big\|_{L_{{\bf x}_{k+p},{\bf x}'_{k+p}}^2(\mathbb{R}^{2(k+p)}\times \mathbb{R}^{2(k+p)})} \times \sup_{\tau,\textbf{u}_k,\textbf{u}'_k} I_{\alpha,p}(\tau,\textbf{u}_k,\textbf{u}'_k)
 \end{split}
\end{equation*}
Where 
\small
\begin{equation}    \label{Ialpha}
 \begin{split}
  &  I_{\alpha,p}(\tau,\textbf{u}_k,\textbf{u}'_k):=   \\
  &  \qquad \int d\textbf{q}d\textbf{q}' \frac{\delta(\cdots)\langle u_1\rangle^{2\alpha}}{\langle u_1+q_1+\cdots+q_{p}-q'_1-\cdots-q'_{p} \rangle^{2\alpha} \langle q_1\rangle^{2\alpha} \cdots \langle q_{p}\rangle^{2\alpha} \langle q'_1\rangle^{2\alpha} \cdots \langle q'_{p}\rangle^{2\alpha}}.
 \end{split}
\end{equation}

\normalsize
If we can show that the supremum  of $I_{\alpha,p}$ over $\tau,\textbf{u}_k,\textbf{u}'_k$ is bounded by a constant (which only depends on $\alpha$) then we are done. Now, observe that 
\small
\begin{equation}    \label{decomp of I}
 \langle u_1 \rangle^{2\alpha} \leq C \big(\langle u_1+q_1+\cdots+q_{p}-q'_1-\cdots-q'_{p} \rangle^{2\alpha}+\langle q_1\rangle^{2\alpha}+\cdots+\langle q_{p}\rangle^{2\alpha}+\langle q'_1\rangle^{2\alpha}+\cdots+\langle q'_{p}\rangle^{2\alpha}\big)
\end{equation}
\normalsize
So we have the following:
\begin{equation}
 I_{\alpha,p}(\tau,\textbf{u}_k,\textbf{u}'_k) \leq \sum_{l=1}^{2p+1} J_l,
\end{equation}
where $J_l$ is obtained by using \eqref{decomp of I} and canceling the corresponding term in the denominator of \eqref{Ialpha}. For example,
\begin{equation}
 J_1 \leq C \int d\textbf{q}d\textbf{q}'  \frac{\delta(\cdots)}{\langle q_1\rangle^{2\alpha} \cdots \langle q_{p}\rangle^{2\alpha} \langle q'_1\rangle^{2\alpha} \cdots \langle q'_{p}\rangle^{2\alpha}} 
\end{equation}
and each $J_l$ for $l=2,3,\cdots,2	p+1$ can be brought into a similar form by appropriately translating one of the momenta $q_j,q'_j$.
Following \cite{KM},\cite{KSS},\cite{CPquintic}, we observe the argument of the $\delta$ distribution equals to
\begin{equation*}
 \begin{split}
  Arg&[\delta]=\tau+(u_1+q_1+\cdots+q_{p}-q'_1-\cdots-q'_{p-1})^2+\sum_{j=2}^{k} u^2_j +|\textbf{q}|^2    \\
 & \qquad -|\textbf{u}'_k|^2-|\textbf{q}'|^2+(q'_{p})^2-2(u_1+q_1+\cdots+q_{p}-q'_1-\cdots-q'_{p-1})\cdot q'_{p}
 \end{split}
\end{equation*}
Then we integrate out the $\delta$ distribution using the component of $q'_{p}$ parallel to $u_1+q_1+\cdots+q_{p}-q'_1-\cdots-q'_{p-1}$, which yields
\small
\begin{equation}
J_1 \leq C_{\alpha} C \int \frac{d\textbf{q}dq'_1\cdots dq'_{p-1} }{|u_1+q_1+\cdots+q_{p}-q'_1-\cdots-q'_{p-1}| \langle q_1\rangle^{2\alpha} \cdots \langle q_{p}\rangle^{2\alpha} \langle q'_1\rangle^{2\alpha} \cdots \langle q'_{p-1}\rangle^{2\alpha}}
\end{equation}
\normalsize
Where
\begin{equation}
 C_{\alpha}:=\int_{\mathbb{R}} \frac{d\zeta}{\langle \zeta \rangle^{2\alpha}} 
\end{equation}
Obviously, $C_{\alpha}$ is finite when $\alpha >\frac{1}{2}$ (Note $\alpha >1-\frac{1}{2(2p_0-1)}\geq 1-\frac{1}{2(2p-1)}\geq \frac{1}{2}$).
Following \cite{CPquintic}, in order to bound $J_1$, we introduce a non-negative spherically symmetric function $h$ with rapid decay away from the unit ball in $\mathbb{R}^2$, such that $\check{h}(x) \geq 0$ decays fast outside the unit ball in $\mathbb{R}^2$, and 
\begin{equation} 
 \frac{1}{\langle q \rangle^{2\alpha}} < \big(h*\frac{1}{|\cdot|^{2\alpha}}\big)(q)
\end{equation}
Such a function $h$ does exists. For example we can take $h(y)=c_1 e^{-c_2 y^2}$ with appropriate $c_1, c_2$. Here we need $\alpha<1$ for $h*\frac{1}{|\cdot|^{2\alpha}}$ to stay in $L^{\infty}(\mathbb{R}^2)$. Then (take $d=2$ below):
\begin{equation}   \label{upper bound of J1}
 \begin{split}
  J_1 & < C_{\alpha} C \langle \big(\frac{1}{|\cdot|}*(h*\frac{1}{|\cdot|^{2\alpha}})\big)*(h*\frac{1}{|\cdot|^{2\alpha}}) \cdots *(h*\frac{1}{|\cdot|^{2\alpha}}),(h*\frac{1}{|\cdot|^{2\alpha}}) \rangle_{L^2(\mathbb{R})} \\
    &= C_{\alpha} C \int dx \big(\frac{1}{|\cdot|}\big)^\vee(x) \Big((h*\frac{1}{|\cdot|^{2\alpha}})^{\vee}(x) \Big)^{2p-1}   \\
    &=C_{\alpha} C' \int dx \frac{1}{|x|^{d-1}}\big(\check{h}(x)\big)^{2p-1} (\frac{1}{|x|^{d-2\alpha}})^{2p-1}  \\
    &=C''_{\alpha}< \infty.
 \end{split}
\end{equation}
Thanks to the decay property of $\check{h}(x)$ outside of the unit ball, the only singularity of the above integral is the origin. Thus \eqref{upper bound of J1} holds if 
\begin{equation}   \label{lower bound on alpha}
 d-1+(2p-1)(d-2\alpha)<d  \quad \Longleftrightarrow \quad \alpha > \frac{d}{2}-\frac{1}{2(2p-1)}
\end{equation}
for all $1\leq p\leq p_0$. 

When $d=2$, we need $\alpha >1-\frac{1}{2(p_0-1)}$ to yield \eqref{upper bound of J1}. Terms $J_2,\cdots,J_{p+1}$ can be bounded in the same manner, thus it suffice to choose $C_{\alpha}=(p_0+1)C''_{\alpha}$. Theorem \ref{thm:free evolving bound theorem} is actually a substitution of Theorem \ref{thm:priori induction estimate on the limiting hierarchy} for high dimensions.
\end{proof}

\vspace{2 mm}
\section{Uniqueness of Solutions}   \label{section: uniqueness of solutions}

We are getting close to prove the conclusions on uniqueness with the results in previous sections. Before doing that we need to introduce some notation appeared in the theorems below. Recall that we use $\gamma^{(k)}(t,\cdot)$ to replace $\gamma^{(k)}_{\infty,t}(\cdot)$ when there is no confusion. The infinite hierarchy \eqref{p-GP hierarchy} can be rewritten in integral form as
\begin{equation}    \label{p-GP hierarchy integral form}
 \gamma^{(k)}(t,\cdot)=\mathcal{U}^{(k)}(t)\gamma^{(k)}(0,\cdot)-i\sum_{p=1}^{p_0}b_p\sum_{j=1}^{k}\int_0^t ds  \mathcal{U}^{(k)}(t-s)B_{j;k+1,\dots,k+p}\gamma^{(k+p)}(s,\cdot)
\end{equation}
Here $b_0^{(p)}=\int_{\mathbb{R}^{pd}}V^{(p)}(x)dx$. Recall the \emph{free evolution operator} $\mathcal{U}^{(k)}(t)$ given by
\begin{equation*}
 \mathcal{U}^{(k)}(t)\gamma^{(k)}=e^{it\Delta_{\pm}^{(k)}}\gamma^{(k)}
\end{equation*}
with $\Delta_{\pm}^{(k)}=\Delta_{\textbf{x}_k}-\Delta_{\textbf{x}'_k}$. \\
Now assume the initial condition $\gamma^{(k)}(0,\cdot)=0$. For fixed positive integer $k$, thanks to Duhamel formula, we can write $\gamma^{(k)}$ in terms of the future iterates $\gamma^{(k+p_1)}$, $\gamma^{(k+p_1+p_2)}, \dotsc, \gamma^{(k+p_1+\cdots+p_n)}$, where $p_1, p_2, \cdots p_n$ are integers chosen from set 
$$S_{p_0}:=\{1, 2, 3, \cdots, p_0\}.$$
Also let $Q_j$ be half of the running sum over $p_1, p_2, p_3, \cdots$:
$$Q_j:=p_1+p_2+\cdots+p_j\leq p_0 j, \qquad j=1, 2, \cdots$$  
Conventionally let $Q_0=0$. Then we have 
\begin{equation} \label{gamma expansion}
 \begin{split}  
  \gamma^{(k)}(t_k,\cdot) &=\sum_{p_1 \in S_{p_0}} b_{p_1}\int_0^{t_k} e^{i(t_k-t_{k+Q_1})\Delta^{(k)}_{\pm}}B^k_{k+Q_1}\big(\gamma^{(k+Q_1)}(t_{k+Q_1})\big) dt_{k+Q_1}  \\
  &=\sum_{p_1, p_2\in S_{p_0}}b_{p_1}b_{p_2}\int_0^{t_k}  e^{i(t_k-t_{k+Q_1}) \Delta^{(k)}_{\pm}}B^k_{k+Q_1}   \\
  &\qquad \Big( \int_0^{t_{k+Q_1}} e^{i(t_{k+Q_1}-t_{k+Q_2})\Delta_{\pm}^{(k+Q_1)}} B^{k+Q_1}_{k+Q_2}\big(\gamma^{({k+Q_2})}(t_{k+Q_2})\big) dt_{k+Q_1} \Big) dt_{k+Q_2}  \\
  &=\cdots   \\
  &=\sum_{p_1,\cdots,p_n \in S_{p_0}} \big(\prod_{j=1}^n b_{p_j}\big) \int_0^{t_k} ... \int_0^{t_{k+Q_{n-1}}} J^k(\underline{t}_{k+Q_n})dt_{k+Q_1} \dots dt_{k+Q_n}
 \end{split}   
\end{equation}

where
\begin{equation*}
  \underline{t}_{k+Q_n}=(t_k,t_{k+Q_1},\dotsc,t_{k+Q_n})
\end{equation*}

\begin{equation*}
 B^{k+Q_{j-1}}_{k+Q_j}:=\sum_{j=1}^{k+Q_{n-1}} B_{j;k+Q_{n-1}+1,k+Q_{n-1}+2,\cdots,k+Q_n}
\end{equation*}

\small
\begin{equation*}
 \begin{split}
  & J^k(\underline{t}_{k+Q_n}) :=    \\
  & \qquad  e^{i(t_k-t_{k+Q_1}) \Delta^{(k)}_{\pm}} B^k_{k+Q_1} \cdots
  e^{i(t_{k+Q_{n-1}}-t_{k+Q_n})\Delta_{\pm}^{(k+Q_{n-1})}} B^{k+Q_{n-1}}_{k+Q_n}\big(\gamma^{(k+Q_n)}(t_{k+Q_n})\big)
 \end{split}
\end{equation*}

\normalsize
Here are the main uniqueness theorems for $d=1,2$:

\begin{thm} \label{thm:1D uniqueness}
 Assume that $d=1$, $t \in [0,T]$ and $\frac{1}{2} < \alpha \leq 1$. The maximal potential constant $b_0=\max\{b_0^{(1)}, b_0^{(2)}, \cdots, b_0^{(p_0)}\}\in (0, \infty)$. Then we have 
\small
\begin{equation}
 \Big\|S^{(k,\alpha)}\gamma^{(k)}(t,\cdot)\Big\|_{L^2(\mathbb{R}^k \times \mathbb{R}^k)} \leq C^k(C_0T)^n
\end{equation}
\normalsize
for arbitrary $n$ and constants $C, C_0$ that are depending on $b_0$, $p_0$, $\kappa$ and $\alpha$, but are independent of $k$ and $T$.
\end{thm}

\begin{thm} \label{thm:2D uniqueness}
 Assume that $d=2$ and $t \in [0,T]$, $1-\frac{1}{2(2p_0-1)} < \alpha \leq 1$. The maximal potential constant $b_0=\max\{b_0^{(1)}, b_0^{(2)}, \cdots, b_0^{(p_0)}\}\in (0, \infty)$. Then we have 
\small
\begin{equation}
 \Big\|S^{(k,\alpha)}\gamma^{(k)}(t,\cdot)\Big\|_{L^2(\mathbb{R}^{2k} \times \mathbb{R}^{2k})} \leq C^k(C_0\sqrt{T})^n
\end{equation}
\normalsize
for arbitrary $n$ and constants $C, C_0$ that are depending on $b_0$, $p_0$, $\kappa$ and $\alpha$, but are independent of $k$ and $T$.
\end{thm}

Based on the above theorems, if we are given sufficiently small $T$, then for all $t\in [0,T]$:
\begin{equation}
 \Big\|S^{(k,\alpha)}\gamma^{(k)}(t,\cdot)\Big\|_{L^2(\mathbb{R}^{dk} \times \mathbb{R}^{dk})}  \to 0\text{  as }n \to 0.
\end{equation}
Which implies that $\gamma^{(k)}(t,\cdot)=0$. Since $k$ is arbitrary, therefore solutions to the infinite hierarchy \eqref{p-GP hierarchy} with zero initial conditions are unique in the above norm.

\begin{proof}[\textbf{Proof of Theorem \ref{thm:1D uniqueness}}]
The idea of the proof is an iterative applications of spacial bound \eqref{priori induction estimate} and Cauchy-Schwarz inequality and at last followed by the use of the bound \eqref{priori spacial bound}. Noticed that $\alpha$ is a constant in $(\frac{1}{2},1]$ and $e^{i(t_k-t_{k+Q_1})\Delta^{(k)}_{\pm}}$ is a unitary operator and commutes with the operator $S^{(k,\alpha)}$, thus we have
\small
\begin{align}
  & \Big\|S^{(k,\alpha)}\int_{0}^{t_k}\dotsi \int_{0}^{t_{k+Q_{n-1}}}J^k(\underline{t}_{k+Q_n})dt_{k+Q_1}\dots  dt_{k+Q_n}\Big\|_{L^2(\mathbb{R}^k \times \mathbb{R}^k)}  \notag  \\
  & \leq \int_{0}^{t_k}\dotsi \int_{0}^{t_{k+Q_{n-1}}} \Big\|S^{(k,\alpha)}e^{i(t_k-t_{k+Q_1}) \Delta^{(k)}_{\pm}} B^k_{k+Q_1} e^{i(t_{k+Q_1}-t_{k+Q_2}) \Delta^{(k+Q_1)}_{\pm}} B^{k+Q_1}_{k+Q_2} \cdots    \notag   \\
  &\qquad \times e^{i(t_{k+Q_{n-1}}-t_{k+Q_n})\Delta_{\pm}^{(k+Q_{n-1})}} B^{{k+Q_{n-1}}}_{k+Q_n}\big(\gamma^{(k+Q_n)}(t_{k+Q_n})\big)\Big\|_{L^2(\mathbb{R}^k \times \mathbb{R}^k)}dt_{k+Q_1}\dots dt_{k+Q_n}   \notag  \\
  & = \int_{0}^{t_k}\dotsi \int_{0}^{t_{k+Q_{n-1}}} \Big\|S^{(k,\alpha)} B^k_{k+Q_1} e^{i(t_{k+Q_1}-t_{k+Q_2}) \Delta^{(k+Q_1)}_{\pm}} B^{k+Q_1}_{k+Q_2} e^{i(t_{k+Q_2}-t_{k+Q_3}) \Delta^{(k+Q_2)}_{\pm}} \cdots    \notag   \\
  &\qquad \times e^{i(t_{k+Q_{n-1}}-t_{k+Q_n})\Delta_{\pm}^{(k+Q_{n-1})}} B^{k+Q_{n-1}}_{k+Q_n}\big(\gamma^{(k+Q_n)}(t_{k+Q_n})\big)\Big\|_{L^2(\mathbb{R}^k \times \mathbb{R}^k)}dt_{k+Q_1}\dots dt_{k+Q_n}   \notag  \\
  & \leq k C_{\alpha} \int_{0}^{t_k}\dotsi \int_{0}^{t_{k+Q_{n-1}}} \Big\|S^{(k+Q_1,\alpha)} e^{i(t_{k+Q_1}-t_{k+Q_2}) \Delta^{(k+Q_1)}_{\pm}} B_{k+Q_2} e^{i(t_{k+Q_2}-t_{k+Q_3}) \Delta^{(k+Q_2)}_{\pm}} \cdots  \label{use induction estimate}     \\
  & \times e^{i(t_{k+Q_{n-1}}-t_{k+Q_n})\Delta_{\pm}^{(k+Q_{n-1})}} B^{k+Q_{n-1}}_{k+Q_n}\big(\gamma^{(k+Q_n)}(t_{k+Q_n})\big)\Big\|_{L^2(\mathbb{R}^{k+Q_1} \times \mathbb{R}^{k+Q_1})}dt_{k+Q_1}\dots dt_{k+Q_n}    \notag  \\
  & \leq \cdots   \notag   \\
  & \leq \prod\limits_{j=0}^{n-2} \Big((k+Q_j)C_{\alpha} \Big) \int_{0}^{t_k}\dotsi \int_{0}^{t_{k+Q_{n-1}}} \Big\|S^{(k+Q_{n-1},\alpha)} e^{i(t_{k+Q_{n-1}}-t_{k+Q_n})\Delta_{\pm}^{(k+Q_{n-1})}}    \label{keep using induction estimate}   \\
  & \qquad \times B^{k+Q_{n-1}}_{k+Q_n}\big(\gamma^{(k+Q_n)}(t_{k+Q_n})\big)\Big\|_{L^2\big(\mathbb{R}^{k+Q_{n-1}} \times \mathbb{R}^{k+Q_{n-1}}\big)}dt_{k+Q_1}\dots dt_{k+Q_n}      \notag  \\
  & = \prod\limits_{j=0}^{n-2} \Big((k+Q_j)C_{\alpha} \Big) \int_{0}^{t_k}\dotsi \int_{0}^{t_{k+Q_{n-1}}} \Big\|S^{(k+Q_{n-1},\alpha)}  \notag  \\ 
  & \qquad  \times B^{k+Q_{n-1}}_{k+Q_n}\big(\gamma^{(k+Q_n)}(t_{k+Q_n})\big)\Big\|_{L^2\big(\mathbb{R}^{k+Q_{n-1}} \times \mathbb{R}^{k+Q_{n-1}}\big)}dt_{k+Q_1}\dots dt_{k+Q_n}  \notag   \\
  & \leq C_{\alpha}^{n-1} \prod\limits_{j=0}^{n-1} (k+p_0 j) \int_{0}^{t_k}\dotsi \int_{0}^{t_{k+Q_{n-1}}} C^{k+p_0 n}dt_{k+Q_1}\dots dt_{k+Q_n}   \label{use spacial bound} \\
  & \leq p_0^n C_{\alpha}^{n-1} \lceil{\frac{k}{p_0}}\rceil (\lceil{\frac{k}{p_0}}\rceil+1)\cdots (\lceil{\frac{k}{p_0}}\rceil+n-1) C^{k+p_0 n} \frac{t_k^n}{n!}   \notag    \\
  & = C^{k+p_0 n} p_0^n C_{\alpha}^{n-1} \binom{\lceil{\frac{k}{p_0}}\rceil+n-1}{n} t_k^n     \notag   \\
  & \leq C^k (p_0 C^{p_0}t_k)^n C_{\alpha}^{n-1} 2^{\lceil{\frac{k}{p_0}}\rceil+n-1}    \notag   
\end{align}
Thus
\begin{align*}
&\qquad  \Big\|S^{(k,\alpha)}\gamma^{(k)}(t_k,\cdot)\Big\|_{L^2(\mathbb{R}^{k} \times \mathbb{R}^{k})}  \\
& \leq \sum_{p_1,\cdots,p_n \in S_{p_0}} \big(\prod_{j=1}^n b_0^{(p_j)}\big) \Big\|S^{(k,\alpha)}\int_{0}^{t_k}\dotsi \int_{0}^{t_{k+Q_{n-1}}}J^k(\underline{t}_{k+Q_n})dt_{k+Q_1}\dots  dt_{k+Q_n}\Big\|_{L^2(\mathbb{R}^k \times \mathbb{R}^k)} \\
&\leq p_0 (b_0)^n C^k (p_0 C^{p_0}t_k)^n C_{\alpha}^{n-1} 2^{\lceil{\frac{k}{p_0}}\rceil+n-1}   \\
&\leq C^k (C_0T)^n  
\end{align*}
\normalsize
where \eqref{use induction estimate} is based on \eqref{priori induction estimate} and we keep using \eqref{priori induction estimate} to obtain \eqref{keep using induction estimate}. Since $e^{i(t_{k+Q_{n-1}}-t_{k+Q_n})\Delta_{\pm}^{(k+Q_{n-1})}}$ is unitary and commutes with $S^{(k+Q_{n-1},\alpha)}$, then after applying Theorem \ref{thm:priori spacial bound on the limiting hierarchy}, we have \eqref{use spacial bound}. $\lceil{x}\rceil$ is the ceiling function. In the last line, choose appropriate $C$ and $C_0$ to finish the proof.                                                                                                                                                                                                                                                                                                                                                                                                                                                                                
\end{proof}

For the proof of Theorem \ref{thm:2D uniqueness}, we run the following combinatorial argument which is inspired by \cite{KM}.

\vspace{2 mm}
\section{Combinatorial Arguments}

\subsection{Graphical Representations}
The key point in the proof of Theorem \ref{thm:2D uniqueness} is to handle the iterative terms from Duhamel formula. Throughout this section, we will prove some lemmas to help us group these terms and also derive some bounds on certain equivalence classes. The technique we used here is an analogous to \cite{KM} and \cite{CPquintic}, but in a much more generalized setting. 

For the reader's convenience, recall some notations we have defined before:
$$\forall 1\leq j\leq n, \quad p_j\in S_{p_0}=\{1, 2, 3, \cdots, p_0\}.$$
$$Q_j:=p_1+p_2+\cdots+p_j, \qquad j=1, 2, \cdots$$ 
Also, $B^{k+Q_{n-1}}_{k+Q_n}=\sum_{j=1}^{k+Q_{n-1}} B_{j;k+Q_{n-1}+1,\cdots,k+Q_n}$, we can rewrite $J^k(\underline{t}_{k+Q_n})$ as the following:
\begin{equation}
 J^k(\underline{t}_{k+Q_n})=\sum\limits_{\mu \in M}J^k(\underline{t}_{k+Q_n};\mu)
\end{equation}
where 
\small
\begin{align*}
 & J^k(\underline{t}_{k+Q_n};\mu)=e^{i(t_k-t_{k+Q_1}) \Delta^{(k)}_{\pm}} B_{\mu(k+1);k+1,\cdots,k+Q_1}e^{i(t_{k+Q_1}-t_{k+Q_2}) \Delta^{(k+Q_1)}_{\pm}} \cdots  \\
 & \quad \times e^{i(t_{k+Q_{n-1}}-t_{k+Q_n})\Delta_{\pm}^{(k+Q_{n-1})}} B_{\mu(k+Q_{n-1}+1);k+Q_{n-1}+1,\cdots,k+Q_n}\big(\gamma^{(k+Q_n)}(t_{k+Q_n})\big)
\end{align*}
\normalsize
and $\mu$ is a map from $\{k+1,k+2,\dots,k+Q_{n-1}+1\}$ to $\{1,2,\dots,k+Q_{n-1}\}$ such that $\mu(2)=1$ and $\mu(j)<j$ for all $j$. $M$ is the set of all these mappings. 

By the definition of $\mu$, we can represent it by highlighting exactly one nonzero entry in each column of a $(k+Q_{n-1}) \times n$ matrix like:

\begin{equation}
 \begin{pmatrix}
  \mathbf{B_{1;k+1,\cdots,k+Q_1}}&B_{1;k+Q_1+1,\cdots,k+Q_2} & \cdots &\mathbf{B_{1;k+Q_{n-1}+1,\cdots,k+Q_n}} \\
  B_{2;k+1,\cdots,k+Q_1} & \mathbf{B_{2;k+Q_1+1,\cdots,k+Q_2}} & \cdots &B_{2;k+Q_{n-1}+1,\cdots,k+Q_n}  \\
  \cdots & \cdots &  \cdots & \cdots  \\
  B_{k;k+1,\cdots,k+Q_1}&B_{k;k+Q_1+1,\cdots,k+Q_2}& \cdots &B_{k;k+Q_{n-1}+1,\cdots,k+Q_n}  \\
  0 & B_{k+1;k+Q_1+1,\cdots,k+Q_2} & \cdots &B_{k+1;k+Q_{n-1}+1,\cdots,k+Q_n}  \\
  \cdots & \cdots & \cdots  & \cdots  \\
  0 & 0 & \cdots & B_{k+Q_{n-1};k+Q_{n-1}+1,\cdots,k+Q_n}  \\
 \end{pmatrix}
\end{equation}

\normalsize
Henceforth we can rewrite \eqref{gamma expansion} as
\begin{equation}
  \gamma^{(k)}(t_k,\cdot)=\int_0^{t_k} ... \int_0^{t_{k+Q_{n-1}}}\sum\limits_{\mu \in M} J^k(\underline{t}_{k+Q_n};\mu)dt_{k+Q_1} \dots dt_{k+Q_n}
\end{equation}

So the basic term of the above sum is the following integral
\begin{equation}  \label{basic integral term}
 I(\mu,\sigma)=\int_{t_k \geq t_{\sigma(k+Q_1)} \geq \cdots \geq t_{\sigma(k+Q_n)}} J^k(\underline{t}_{k+Q_n};\mu)dt_{k+Q_1} \dots dt_{k+Q_n}
\end{equation}
where $\sigma$ is a permutation of ${k+Q_1,k+Q_2,\dots,k+Q_n}$. We will associate the integral $I(\mu,\sigma)$ to the following $(k+Q_{n-1}+1) \times n$ matrix. Matrix \eqref{integral matrix} is also helpful to visualize $B_{\mu(k+Q_{j-1}+1);k+Q_{j-1}+1,\cdots,k+Q_j}$, $j=1,2,\dots,n$ and $\sigma$:
\begin{equation}
  \begin{pmatrix}   \label{integral matrix}
  t_{\sigma^{-1}(k+Q_1)} & t_{\sigma^{-1}(k+Q_2)} & \cdots & t_{\sigma^{-1}(k+Q_n)}  \\
  \mathbf{B_{1;k+1,\cdots,k+Q_1}}&B_{1;k+Q_1+1,\cdots,k+Q_2} & \cdots &\mathbf{B_{1;k+Q_{n-1}+1,\cdots,k+Q_n}} \\
  B_{2;k+1,\cdots,k+Q_1} & \mathbf{B_{2;k+Q_1+1,\cdots,k+Q_2}} & \cdots &B_{2;k+Q_{n-1}+1,\cdots,k+Q_n}  \\
  \cdots & \cdots &  \cdots & \cdots  \\
  B_{k;k+1,\cdots,k+Q_1}&B_{k;k+Q_1+1,\cdots,k+Q_2}& \cdots &B_{k;k+Q_{n-1}+1,\cdots,k+Q_n}  \\
  0 & B_{k+1;k+Q_1+1,\cdots,k+Q_2} & \cdots &B_{k+1;k+Q_{n-1}+1,\cdots,k+Q_n}  \\
  \cdots & \cdots & \cdots  & \cdots  \\
  0 & 0 & \cdots & B_{k+Q_{n-1};k+Q_{n-1}+1,\cdots,k+Q_n}  \\
 \end{pmatrix}
\end{equation}
We label the columns of matrix \eqref{integral matrix} by $1$ through $n$ while rows $0$ through $k+Q_{n-1}$. 

\subsection{Acceptable Moves}
It is an important step to introduce the so called ``acceptable move'' on the set of matrices like \eqref{integral matrix}. In particular, if $\mu(k+Q_j+1)<\mu(k+Q_{j-1}+1)$, we are allowed to do the following changes at the same time:
\begin{itemize}
 \item exchange the highlights in columns $j$ and $j+1$
 \item exchange the highlights in rows $k+Q_{j-1}+1$ and $k+Q_j+1$
 \item exchange the highlights in rows $k+Q_{j-1}+2$ and $k+Q_j+2$
 \item $\cdots$ 
 \item exchange the highlights in rows $k+Q_{j-1}+r_0$ and $k+Q_j+r_0$
 \item exchange $t_{\sigma^{-1}(k+Q_j)}$ and $t_{\sigma^{-1}(k+Q_{j+1})}$
\end{itemize}
with $r_0:=\min\{p_j,p_{j+1}\}.$

For instance, if $k=1,n=4,p_1=2, p_2=1,p_3=3,p_4=2$, then we go from 
\begin{equation*}
  \begin{pmatrix}
  t_{\sigma^{-1}(1+Q_1)} & t_{\sigma^{-1}(1+Q_2)} & t_{\sigma^{-1}(1+Q_3)} & t_{\sigma^{-1}(1+Q_4)} \\
  \mathbf{B_{1;2,3}} & B_{1;4} & B_{1;5,6,7} & B_{1;8,9} \\
  0 & B_{2;4} & \mathbf{B_{2;5,6,7}} & B_{2;8,9}  \\
  0 & \mathbf{B_{3;4}} & B_{3;5,6,7} & B_{3;8,9}  \\
  0 & 0 & B_{4;5,6,7} & \mathbf{B_{4;8,9}}  \\
  0 & 0 & 0 & B_{5;8,9}  \\
  0 & 0 & 0 & B_{6;8,9}  \\
  0 & 0 & 0 & B_{7;8,9}  \\
 \end{pmatrix}
\end{equation*}
to
\begin{equation*}
  \begin{pmatrix}
  t_{\sigma^{-1}(1+Q_1)} & t_{\sigma^{-1}(1+Q_3)} & t_{\sigma^{-1}(1+Q_2)} & t_{\sigma^{-1}(1+Q_4)} \\
  \mathbf{B_{1;2,3}} & B_{1;4} & B_{1;5,6,7} & B_{1;8,9} \\
  0 & \mathbf{B_{2;4}} & B_{2;5,6,7} & B_{2;8,9}  \\
  0 & B_{3;4} & \mathbf{B_{3;5,6,7}} & B_{3;8,9}  \\
  0 & 0 & B_{4;5,6,7} & B_{4;8,9}  \\
  0 & 0 & 0 & \mathbf{B_{5;8,9}}  \\
  0 & 0 & 0 & B_{6;8,9}  \\
  0 & 0 & 0 & B_{7;8,9}  \\
 \end{pmatrix}
\end{equation*}

The reason for taking such moves is explained by the following lemma.

\begin{lemma}  \label{lem:equal integral}
 Let $(\mu,\sigma)$ be transformed into $(\mu',\sigma')$ by an acceptable move. Then, for the corresponding integrals \eqref{basic integral term}, $I(\mu,\sigma)=I(\mu',\sigma')$
\end{lemma}
\begin{proof}
 This is a relatively straightforward proof but somewhat tedious, as in \cite{KM} and \cite{CPquintic}. We modify the proof of Lemma 7.1 in \cite{CPquintic} so that it can be used here. Since there is only one acceptable move between the two integrals, most part of their expressions share the same terms. Let us fix $j\geq 3$, select two integers $i, l$ such that $i<l<j<j+1$ and compare $I(\mu,\sigma)$ and $I(\mu',\sigma')$
\begin{equation}    \label{integral before move}
 \begin{split}
  I(\mu,\sigma) & =\int_{t_k \geq\cdots t_{\sigma(k+Q_j)} \geq t_{\sigma(k+Q_{j+1})} \cdots\geq t_{\sigma(k+Q_n)}\geq 0} J^k(\underline{t}_{k+Q_n};\mu)dt_{k+Q_1} \dots dt_{k+Q_n}  \\
  & = \int_{t_k \geq\cdots t_{\sigma(k+Q_j)} \geq t_{\sigma(k+Q_{j+1})} \cdots\geq t_{\sigma(k+Q_n)}\geq 0} \cdots e^{i(t_{k+Q_{j-1}}-t_{k+Q_j}) \Delta^{(k+Q_{j-1})}_{\pm}}  \\
& \quad \times  B_{l;k+Q_{j-1}+1,\cdots,k+Q_j} e^{i(t_{k+Q_j}-t_{k+Q_{j+1}}) \Delta^{(k+Q_j)}_{\pm}} B_{i;k+Q_j+1,\cdots,k+Q_{j+1}}   \\
  & \quad \times e^{i(t_{k+Q_{j+1}}-t_{k+Q_{j+2}})\Delta_{\pm}^{(k+Q_{j+1})}} (\cdots) dt_{k+Q_1} \dots dt_{k+Q_n}
 \end{split}
\end{equation}
and 
\begin{equation}    \label{integral after move}
 \begin{split}
  I(\mu',\sigma') & =\int_{t_k \geq\cdots t_{\sigma'(k+Q_j)} \geq t_{\sigma'(k+Q_{j+1})} \cdots\geq t_{\sigma'(k+Q_n)}\geq 0} J^k(\underline{t}_{k+Q_n};\mu')dt_{k+Q_1} \dots dt_{k+Q_n}  \\
  & = \int_{t_k \geq\cdots t_{\sigma'(k+Q_j)} \geq t_{\sigma'(k+Q_{j+1})} \cdots\geq t_{\sigma'(k+Q_n)}\geq 0}  \cdots e^{i(t_{k+Q_{j-1}}-t_{k+Q_j}) \Delta^{(k+Q_{j-1})}_{\pm}} \\
  & \quad \times  B_{i;k+Q_{j-1}+1,\cdots,k+Q_j} e^{i(t_{k+Q_j}-t_{k+Q_{j+1}}) \Delta^{(k+Q_j)}_{\pm}} B_{l;k+Q_j+1,\cdots,k+Q_{j+1}}  \\
  & \quad \times e^{i(t_{k+Q_{j+1}}-t_{k+Q_{j+2}})\Delta_{\pm}^{(k+Q_{j+1})}} (\cdots)' dt_{k+Q_1} \dots dt_{k+Q_n}
 \end{split}
\end{equation}
\\
The $\cdots$ in \eqref{integral before move} and \eqref{integral after move} coincide.  \\
For $1\leq r\leq r_0=\min\{p_j,p_{j+1}\}$, $s\geq Q_j$ and index $m$: $j+1\leq m\leq n$, any $B_{k+Q_{j-1}+r;s+1,\cdots,s+p_m}$ (when it is highlighted) in $(\cdots)$ of \eqref{integral before move} will become $B_{k+Q_j+r;s+1,\cdots,s+p_m}$ in $(\cdots)'$ of \eqref{integral after move} and any $B_{k+Q_j+r;s+1,\cdots,s+p_m}$ (when it is highlighted) in $(\cdots)$ of \eqref{integral after move} become $B_{k+Q_{j-1}+r;s+1,\cdots,s+p_m}$ in $(\cdots)'$ of \eqref{integral before move};

All the changes are illustrated in the table below:

\begin{center}
\begin{tabular}{l c r}
  \qquad \qquad $(\cdots)$ & $\longleftrightarrow$ & $(\cdots)'$ \qquad \qquad \\
  \hline
  $B_{k+Q_{j-1}+1;s+1,\cdots,s+p_m}$ & $\leftrightarrow$  & $B_{k+Q_j+1;s+1,\cdots,s+p_m}$ \\
  $B_{k+Q_{j-1}+2;s+1,\cdots,s+p_m}$ & $\leftrightarrow$  & $B_{k+Q_j+2;s+1,\cdots,s+p_m}$ \\
  $\qquad \qquad \vdots$ & $\leftrightarrow$ & $\vdots \qquad \qquad$   \\
  $B_{k+Q_{j-1}+r_0;s+1,\cdots,s+p_m}$ & $\leftrightarrow$  & $B_{k+Q_j+r_0;s+1,\cdots,s+p_m}$ \\
\end{tabular}
\end{center}

In order to prove $I(\mu,\sigma)=I(\mu',\sigma')$ we introduce $P$ and $\tilde{P}$ which are defined as:
\begin{align}
&P=B_{l;k+Q_{j-1}+1,\cdots,k+Q_j} e^{i(t_{k+Q_j}-t_{k+Q_{j+1}}) \Delta^{(k+Q_j)}_{\pm}} B_{i;k+Q_j+1,\cdots,k+Q_{j+1}}  \label{P} \\
&\tilde{P}=B_{i;k+Q_j+1,\cdots,k+Q_{j+1}} e^{-i(t_{k+Q_j}-t_{k+Q_{j+1}}) \tilde{\Delta}^{(k+Q_j)}_{\pm}} B_{l;k+Q_{j-1}+1,\cdots,k+Q_j}  \label{P tilde}
\end{align}
where
\begin{equation*}
\begin{split}
\tilde{\Delta}^{(k+Q_j)}_{\pm}=\Delta^{(k+Q_j)}_{\pm}& -\Delta_{\pm,x_{k+Q_j}}
-\Delta_{\pm,x_{k+Q_j-1}}-\cdots-\Delta_{\pm,x_{k+Q_{j-1}+1}}   \\
& +\Delta_{\pm,x_{k+Q_j+1}}+\Delta_{\pm,x_{k+Q_j+2}}+\cdots+\Delta_{\pm,x_{k+Q_{j+1}}}
\end{split}
\end{equation*}
We've used this notion above: $\Delta_{\pm,x_j}=\Delta_{x_j}-\Delta_{x'_j}$.

We will show that
\begin{equation}     \label{supp equality}
\begin{split}
& e^{i(t_{k+Q_{j-1}}-t_{k+Q_j}) \Delta^{(k+Q_{j-1})}_{\pm}} P e^{i(t_{k+Q_{j+1}}-t_{k+Q_{j+2}})\Delta_{\pm}^{(k+Q_{j+1})}} \\
& = e^{i(t_{k+Q_{j-1}}-t_{k+Q_{j+1}}) \Delta^{(k+Q_{j-1})}_{\pm}} \tilde{P} e^{i(t_{k+Q_j}-t_{k+Q_{j+2}})\Delta_{\pm}^{(k+Q_{j+1})}}
\end{split}
\end{equation}

Indeed in \eqref{P} we can write $\Delta^{(k+Q_j)}_{\pm}=\Delta_{\pm,x_i}+(\Delta^{(k+Q_j)}_{\pm}-\Delta_{\pm,x_i})$. Therefore, 
\begin{equation*}
e^{i(t_{k+Q_j}-t_{k+Q_{j+1}}) \Delta^{(k+Q_j)}_{\pm}}= e^{i(t_{k+Q_j}-t_{k+Q_{j+1}}) \Delta_{\pm,x_i}}e^{i(t_{k+Q_j}-t_{k+Q_{j+1}}) (\Delta^{(k+Q_j)}_{\pm}-\Delta_{\pm,x_i})}
\end{equation*}
Observe that the first term on the RHS of the above equation can be commuted to the left of $B_{l;k+Q_{j-1}+1,\cdots,k+Q_j}$ and the second one to the right of $B_{i;k+Q_j+1,\cdots,k+Q_{j+1}}$, thus after two commutations
\begin{equation}    \label{rewrite P}
\begin{split}
P=&e^{i(t_{k+Q_j}-t_{k+Q_{j+1}})\Delta_{\pm,x_i}} B_{l;k+Q_{j-1}+1,\cdots,k+Q_j} B_{i;k+Q_j+1,\cdots,k+Q_{j+1}}   \\
& \times e^{i(t_{k+Q_j}-t_{k+Q_{j+1}}) (\Delta^{(k+Q_j)}_{\pm}-\Delta_{\pm,x_i})}
\end{split}
\end{equation}
and the LHS of \eqref{supp equality} becomes
\begin{equation}   \label{lhs of supp equality}
\begin{split}
&e^{i(t_{k+Q_{j-1}}-t_{k+Q_j}) \Delta^{(k+Q_{j-1})}_{\pm}} P e^{i(t_{k+Q_{j+2}}-t_{k+Q_{j+2}})\Delta_{\pm}^{(k+Q_{j+1})}}   \\
& = e^{i(t_{k+Q_{j-1}}-t_{k+Q_j}) \Delta^{(k+Q_{j-1})}_{\pm}} e^{i(t_{k+Q_j}-t_{k+Q_{j+1}})\Delta_{\pm,x_i}} B_{l;k+Q_{j-1}+1,\cdots,k+Q_j}  \\
&\quad \times B_{i;k+Q_j+1,\cdots,k+Q_{j+1}} e^{i(t_{k+Q_j}-t_{k+Q_{j+1}})(\Delta^{(k+Q_j)}_{\pm}-\Delta_{\pm,x_i})} e^{i(t_{k+Q_{j+1}}-t_{k+Q_{j+2}})\Delta_{\pm}^{(k+Q_{j+1})}} \\
& = e^{i(t_{k+Q_{j-1}}-t_{k+Q_j}) \Delta^{(k+Q_{j-1})}_{\pm}} e^{i(t_{k+Q_j}-t_{k+Q_{j+1}})\Delta_{\pm,x_i}} B_{i;k+Q_j+1,\cdots,k+Q_{j+1}}  \\
&\quad \times B_{l;k+Q_{j-1},\cdots,k+Q_j} e^{i(t_{k+Q_{j+1}}-t_{k+Q_{j+2}}) (\Delta_{\pm,x_i}+\Delta_{\pm,x_{k+Q_j+1}} + \cdots+\Delta_{\pm,x_{k+Q_{j+1}}})}  \\
& \quad \times e^{i(t_{k+Q_j}-t_{k+Q_{j+2}}) (\Delta_{\pm,x_1}+\cdots+\hat{\Delta}_{\pm,x_i}+\cdots+\Delta_{\pm,x_{k+Q_j}})}
\end{split}
\end{equation}
where a hat denotes a missing term.

Similarly, we can rewrite $\tilde{\Delta}^{(k+Q_j)}_{\pm}$ as
\begin{align*}
\tilde{\Delta}^{(k+Q_j)}_{\pm}& =\Delta^{(k+Q_j)}_{\pm}-\Delta_{\pm,x_{k+Q_j}}
-\cdots-\Delta_{\pm,x_{k+Q_{j-1}+1}}   \\
&\qquad +\Delta_{\pm,x_{k+Q_j+1}}+\cdots+\Delta_{\pm,x_{k+Q_{j+1}}}  \\
& = \Delta^{(k+Q_{j-1})}_{\pm}+\Delta_{\pm,x_{k+Q_j+1}} +\cdots+\Delta_{\pm,x_{k+Q_{j+1}}}  \\
& = (\Delta^{(k+Q_{j-1})}_{\pm}-\Delta_{\pm,x_i}) +(\Delta_{\pm,x_i}+ \Delta_{\pm,x_{k+Q_j+1}}+\cdots+\Delta_{\pm,x_{k+Q_{j+1}}})
\end{align*}

Hence the factor $e^{-i(t_{k+Q_j}-t_{k+Q_{j+1}}) \tilde{\Delta}^{(k+Q_j)}_{\pm}}$ appearing in the definition of $\tilde{P}$ can be rewritten as
\begin{equation*}
 \begin{split}
  & e^{-i(t_{k+Q_j}-t_{k+Q_{j+1}}) \tilde{\Delta}^{(k+Q_j)}_{\pm}}  \\
  & = e^{-i(t_{k+Q_j}-t_{k+Q_{j+1}})(\Delta^{(k+Q_{j-1})}_{\pm}-\Delta_{\pm,x_i})}  \\
  & \times e^{-i(t_{k+Q_j}-t_{k+Q_{j+1}})(\Delta_{\pm,x_i}+\Delta_{\pm,x_{k+Q_j+1}}+ \cdots+\Delta_{\pm,x_{k+Q_{j+1}}})}
\end{split}
\end{equation*}
and consequently,
\begin{equation}    \label{rewrite P tilde}
\begin{split}
\tilde{P}=& e^{-i(t_{k+Q_j}-t_{k+Q_{j+1}}) (\Delta^{(k+Q_{j-1})}_{\pm}-\Delta_{\pm,x_i})} B_{i;k+Q_j+1,\cdots,k+Q_{j+1}} B_{l;k+Q_{j-1}+1,\cdots,k+Q_j}   \\
& \quad \times e^{-i(t_{k+Q_j}-t_{k+Q_{j+1}})(\Delta_{\pm,x_i}+ \Delta_{\pm,x_{k+Q_j+1}}+\cdots+\Delta_{\pm,x_{k+Q_{j+1}}})}
\end{split}
\end{equation}

The RHS of \eqref{supp equality} equals to
\begin{align}
& e^{i(t_{k+Q_{j-1}}-t_{k+Q_{j+1}}) \Delta^{(k+Q_{j-1})}_{\pm}} \tilde{P} e^{i(t_{k+Q_j}-t_{k+Q_{j+2}})\Delta_{\pm}^{(k+Q_{j+1})}}  \\
& = e^{i(t_{k+Q_{j-1}}-t_{k+Q_{j+1}}) \Delta^{(k+Q_{j-1})}_{\pm}} e^{-i(t_{k+Q_j}-t_{k+Q_{j+1}})(\Delta^{(k+Q_{j-1})}_{\pm}-\Delta_{\pm,x_i})}    \notag   \\
&\quad \times B_{i;k+Q_j+1,\cdots,k+Q_{j+1}} B_{l;k+Q_{j-1}+1,\cdots,k+Q_j}      \notag  \\
& \quad \times  e^{-i(t_{k+Q_j}-t_{k+Q_{j+1}})(\Delta_{\pm,x_i}+\Delta_{\pm,x_{k+Q_j+1}}+\cdots+\Delta_{\pm,x_{k+Q_{j+1}}})} e^{i(t_{k+Q_j}-t_{k+Q_{j+2}})\Delta_{\pm}^{(k+Q_{j+1})}}   \notag  \\
& = e^{i(t_{k+Q_{j-1}}-t_{k+Q_j}) \Delta^{(k+Q_{j-1})}_{\pm}} e^{i(t_{k+Q_j}-t_{k+Q_{j+1}})\Delta_{\pm,x_i}} B_{i;k+Q_j+1,\cdots,k+Q_{j+1}}    \notag \\
& \quad \times B_{l;k+Q_{j-1}+1,\cdots,k+Q_j} e^{i(t_{k+Q_{j+1}}-t_{k+Q_{j+2}}) (\Delta_{\pm,x_i}+\Delta_{\pm,x_{k+Q_j+1}} +\cdots+\Delta_{\pm,x_{k+Q_{j+1}}})}    \notag \\
& \quad \times e^{i(t_{k+Q_j}-t_{k+Q_{j+2}}) (\Delta_{\pm,x_1}+\cdots+\hat{\Delta}_{\pm,x_i}+\cdots+\Delta_{\pm,x_{k+Q_j}})}   \notag
\end{align}
which is the same as \eqref{lhs of supp equality}, so \eqref{supp equality} is proved.

Note $r_0=\min\{p_j, p_{j+1}\}$. By the symmetry property of $\gamma^{(k+Q_n)}$, we can perform the following exchanges without changing it's value
\begin{itemize}
 \item exchange $(x_{k+Q_{j-1}+1},x'_{k+Q_{j-1}+1})$ with $(x_{k+Q_j+1},x'_{k+Q_j+1})$
 \item exchange $(x_{k+Q_{j-1}+2},x'_{k+Q_{j-1}+2})$ with $(x_{k+Q_j+2},x'_{k+Q_j+2})$
 \item $\cdots$
 \item exchange $(x_{k+Q_{j-1}+r_0},x'_{k+Q_{j-1}+r_0})$ with $(x_{k+Q_j+r_0},x'_{k+Q_j+r_0})$
\end{itemize}
After performing these exchanges only in the arguments of $\gamma^{(k+Q_n)}$ we can rewrite \eqref{integral before move} based on \eqref{supp equality} as follows:
\begin{align}   
& \quad I(\mu,\sigma)    \label{replaced integral} \\
&=\int_{t_k \geq\cdots t_{\sigma(k+Q_j)} \geq t_{\sigma(k+Q_{j+1})} \cdots\geq t_{\sigma(k+Q_n)}\geq 0}\cdots    \notag \\
& \times e^{i(t_{k+Q_{j-1}}-t_{k+Q_j}) \Delta^{(k+Q_{j-1})}_{\pm}} P e^{i(t_{k+Q_{j+1}}-t_{k+Q_{j+2}})\Delta_{\pm}^{(k+Q_{j+1})}} (\cdots)' dt_{k+Q_1} \dots dt_{k+Q_n}  \notag \\
&=\int_{t_k \geq\cdots t_{\sigma(k+Q_j)} \geq t_{\sigma(k+Q_{j+1})} \cdots\geq t_{\sigma(k+Q_n)}\geq 0}\cdots   \notag \\
& \times e^{i(t_{k+Q_{j-1}}-t_{k+Q_{j+1}}) \Delta^{(k+Q_{j-1})}_{\pm}} \tilde{P} e^{i(t_{k+Q_j}-t_{k+Q_{j+2}})\Delta_{\pm}^{(k+Q_{j+1})}}(\cdots)' dt_{k+Q_1} \dots dt_{k+Q_n}  \notag \\
&=\int_{t_k \geq\cdots t_{\sigma(k+Q_j)} \geq t_{\sigma(k+Q_{j+1})} \cdots\geq t_{\sigma(k+Q_n)}\geq 0}\int_{\mathbb{R}^{(k+Q_{j+1})d}}\cdots  e^{i(t_{k+Q_{j-1}}-t_{k+Q_{j+1}}) \Delta^{(k+Q_{j-1})}_{\pm}}    \notag \\
& \quad \times \delta_{i;k+Q_j+1,\cdots,k+Q_{j+1}} e^{-i(t_{k+Q_j}-t_{k+Q_{j+1}}) \tilde{\Delta}^{(k+Q_j)}_{\pm}} \delta_{l;k+Q_{j-1}+1,\cdots,k+Q_j}   \notag \\
& \quad \times e^{i(t_{k+Q_j}-t_{k+Q_{j+2}})\Delta_{\pm}^{(k+Q_{j+1})}}(\cdots)' dt_{k+Q_1} \dots dt_{k+Q_n}   \notag
\end{align}
in which $\delta_{j;s+1,\cdots,s+p_m}$ denotes the abbreviated kernel of the operator $B_{j;s+1,\cdots,s+p_m}$:

\begin{equation}   \label{contraction kernel}
\begin{split}
 \delta_{j;s+1,\cdots,s+p_m} =& \delta(x_j-x_{s+1}) \delta (x_j-x'_{s+1})\cdots \delta(x_j-x_{s+p_m})\delta (x_j-x'_{s+p_m}) \\
 & - \delta(x'_j-x_{s+1}) \delta (x'_j-x'_{s+1}) \cdots \delta(x'_j-x_{s+p_m}) \delta (x'_j-x'_{s+p_m}).
\end{split}
\end{equation}
In \eqref{replaced integral} we perform the change of variables that
\begin{flushleft}
\begin{itemize}
 \item exchange $(t_{k+Q_{j-1}+1},x_{k+Q_{j-1}+1},x'_{k+Q_{j-1}+1})$ with $(t_{k+Q_j+1},x_{k+Q_j+1},x'_{k+Q_j+1})$
 \item exchange $(t_{k+Q_{j-1}+2},x_{k+Q_{j-1}+2},x'_{k+Q_{j-1}+2})$ with $(t_{k+Q_j+2},x_{k+Q_j+2},x'_{k+Q_j+2})$
 \item $\cdots$
 \item exchange $(t_{k+Q_{j-1}+r_0},x_{k+Q_{j-1}+r_0},x'_{k+Q_{j-1}+r_0})$ with $(t_{k+Q_j+r_0},x_{k+Q_j+r_0},x'_{k+Q_j+r_0})$
\end{itemize}
\end{flushleft}
in the whole integral. 
Under the same change of variables $\tilde{\Delta}^{(k+Q_j)}_{\pm}$ becomes
\begin{equation*}
\begin{split}
\tilde{\Delta}^{(k+Q_j)}_{\pm}&=\Delta^{(k+Q_j)}_{\pm} -\Delta_{\pm,x_{k+Q_j}}
-\cdots-\Delta_{\pm,x_{k+Q_{j-1}+1}}   \\
& \qquad +\Delta_{\pm,x_{k+Q_j+1}}+\cdots+ \Delta_{\pm,x_{k+Q_{j+1}}}  \\
&= \Delta^{(k+Q_{j-1})}_{\pm}+\Delta_{\pm,x_{k+Q_j+1}}+\cdots+\Delta_{\pm,x_{k+Q_{j+1}}}   \\
& \to  \Delta^{(k+Q_{j-1})}_{\pm}+\Delta_{\pm,x_{k+Q_{j-1}+1}}+\cdots + \Delta_{\pm,x_{k+Q_j}}   \\
&=  \Delta^{(k+Q_j)}_{\pm}
\end{split}
\end{equation*}
Note that $\Delta_{\pm}^{(k+Q_{j+1})}$ stay unchanged under this change of variable. Therefore, we obtain:
\begin{equation}
\begin{split}
I(\mu,\sigma)
&=\int_{t_k \geq\cdots t_{\sigma'(k+Q_j)} \geq t_{\sigma'(k+Q_{j+1})} \cdots\geq t_{\sigma'(k+Q_n)}}\cdots  e^{i(t_{k+Q_{j-1}}-t_{k+Q_j}) \Delta^{(k+Q_{j-1})}_{\pm}}    \\
& \quad \times B_{i;k+Q_{j-1}+1,\cdots,k+Q_j} e^{i(t_{k+Q_j}-t_{k+Q_{j+1}}) \Delta^{(k+Q_j)}_{\pm}} B_{l;k+Q_j+1,\cdots,k+Q_{j+1}}   \\
& \quad \times e^{i(t_{k+Q_{j+1}}-t_{k+Q_{j+2}})\Delta_{\pm}^{(k+Q_{j+1})}}(\cdots)' dt_{k+Q_1} \dots dt_{k+Q_n}   \\
&= I(\mu',\sigma')
\end{split}
\end{equation}
where $\sigma'=(k+Q_j,k+Q_{j+1}) \circ \sigma$.  $(k+Q_j,k+Q_{j+1})$ denotes the permutation which reverses $k+Q_j$ and $k+Q_{j+1}$.
\end{proof}

Next, let us consider the subset $\{\mu_s\} \subset M$ of ``special upper echelon" matrices in which each highlighted element of a higher row is to the left of each highlighted element of a lower row. A simple example of a ``special upper echelon" matrix is given below (with $k=1,n=4,p_1=2, p_2=1,p_3=3,p_4=2$)
\begin{equation*}
  \begin{pmatrix}
  \mathbf{B_{1;2,3}} & \mathbf{B_{1;4}} & B_{1;5,6,7} & B_{1;8,9} \\
  0 & B_{2;4} & B_{2;5,6,7} & B_{2;8,9}  \\
  0 & B_{3;4} & B_{3;5,6,7} & B_{3;8,9}  \\
  0 & 0 & \mathbf{B_{4;5,6,7}} & B_{4;8,9}  \\
  0 & 0 & 0 & B_{5;8,9}  \\
  0 & 0 & 0 & B_{6;8,9}  \\
  0 & 0 & 0 & \mathbf{B_{7;8,9}}  \\
 \end{pmatrix}
\end{equation*}

\begin{lemma}
 For each element of $M$ there is a finite number of acceptable moves which brings the matrix to upper echelon form. 
\end{lemma}
\begin{proof}
We start from the first row and take acceptable moves to bring all highlighted entries in the first row in consecutive order. Since our goal is the upper echelon form, the updated highlighted entries will occupy $\bf{B_{1;k+1,\cdots,k+Q_1}}$ through $\bf{B_{1;k+Q_{j_1-1}+1,\cdots,k+Q_{j_1}}}$. Then if there are any highlighted entries on the second row, bring them to positions $\bf{B_{2;k+Q_{j_1}+1,\cdots,k+Q_{j_1+1}}}$ through  $\bf{B_{2;k+Q_{j_2-1}+1,\cdots,k+Q_{j_2}}}$. Here $j_1<j_2$. Noticed that this will not effect the highlighted positions of the first row. If there is no highlighted entire on the second row, just leave it and move to the third row. Keep repeating these steps and we will end up with a special upper echelon matrix after finitely many steps.
\end{proof}

\begin{lemma}   \label{lem:number of echelon matrix}
 Let $C_{k,n}$ be the number of $(k+Q_{n-1}) \times n$ special upper echelon matrices of the type discussed above. Then $C_{k,n} \leq 2^{k+(p_0+1)(n-1)}$.
\end{lemma}
\begin{proof}
The proof consists of two steps. First of all, we dis-assemble the matrix by ``lifting'' all highlighted entries to the first row and put them in the same subsets if they were originally from the same row. In this way, the first row is partitioned into many subsets. Let $P_n$ denote the number of all possible partitions, then 
\begin{equation}    \label{P_n}
P_n=\sum\limits_{i=0}^{n-1}\binom{n-1}{i}=2^{n-1}
\end{equation}
The idea is to put $n-1$ pads in the space among the $n$ elements to separate them. Since we can separate them into different numbers (from $1$ to $n$) of subsets, we can choose to use $0$ pads, $1$ pads, $\cdots$, upto $n-1$ pads. Hence \eqref{P_n} follows. 

The second step is to re-assemble the upper echelon matrix by ``lowering'' the first subset to the first used row, the second subset to the second used row, etc. Note here, we do not require that only the upper triangle matrix is used, which may result in more matrices. This does not matter since we are looking for an upper bound of the number of such matrices. Suppose an arbitrary partition of $n$ has $i$ subsets. Then there will be exactly $\binom{k+Q_{n-1}}{i}$ ways to lower them in an order preserving way to the $k+Q_{n-1}$ available rows. 
Thus
\begin{equation*}
C_{k,n} \leq P_n \sum\limits_{i=0}^{n} \binom{k+Q_{n-1}}{i} \leq 2^{k+Q_{n-1}+n-1} \leq 2^{k+(p_0+1)(n-1)}
\end{equation*}
as desired (since $Q_{n-1}=p_1+p_2+\cdots+p_{n-1} \leq (n-1)p_0$).
\end{proof}
 
\subsection{Equivalence Classes}

\begin{cor}    \label{cor:equal domain integral}
Let $\mu_s$ be a special upper echelon matrix. We write $\mu \sim \mu_s$ if $\mu$ can be transformed to $\mu_s$ in finitely many acceptable moves. Then there exists a subset $D$ of $[0,t_k]^n$ such that
\small
\begin{equation}
\sum\limits_{\mu \sim \mu_s}\int_0^{t_k} ... \int_0^{t_{k+Q_{n-1}}} J^k(\underline{t}_{k+Q_n};\mu)dt_{k+Q_1} \dots dt_{k+Q_n} = \int\limits_{D} J^k(\underline{t}_{k+Q_n};\mu)dt_{k+Q_1} \dots dt_{k+Q_n}
\end{equation}
\end{cor}

\normalsize
\begin{proof}
Consider the following integral
\begin{equation*}
I(\mu,id)=\int_0^{t_k} ... \int_0^{t_{k+Q_{n-1}}} J^k(\underline{t}_{k+Q_n};\mu)dt_{k+Q_1} \dots dt_{k+Q_n}
\end{equation*}
and perform finitely many acceptable moves on the matrix associated to $I(\mu,id)$ until it is transformed to the special upper echelon matrix associated with $I(\mu_s,\sigma)$. By Lemma \ref{lem:equal integral}
\begin{equation*}
I(\mu,id)=I(\mu_s,\sigma).
\end{equation*}
Assume that $(\mu_1,id)$ and $(\mu_2,id)$ with $\mu_1 \neq \mu_2$ yield the same echelon form $\mu_s$, then the corresponding permutations $\sigma_1$ and $\sigma_2$ must be different. Therefore, $D$ can be chosen to be the union of all $\{t_k \geq t_{\sigma(k+Q_1)} \geq t_{\sigma(k+Q_2)} \geq \cdots \geq  t_{\sigma(k+Q_{n-1})} \}$ for all permutations $\sigma$ which occur in a given equivalence class of some $\mu_s$.
\end{proof}

\vspace{2 mm}
\section{Proof of Theorem \ref{thm:2D uniqueness}}

Now we are ready to prove Theorem \ref{thm:2D uniqueness}. 
\begin{proof}[\textbf{Proof of Theorem \ref{thm:2D uniqueness}}]
Fix $t_k$. Recall the expansion of $\gamma^{(k)}$:
\begin{equation}
  \gamma^{(k)}(t_k,\cdot)=\sum\limits_{\mu \in M} \int_0^{t_k} ... \int_0^{t_{k+Q_{n-1}}} J^k(\underline{t}_{k+Q_n};\mu) dt_{k+Q_1} \dots dt_{k+Q_n}
\end{equation}

and $J^k$: 
\begin{align*}
 & J^k(\underline{t}_{k+Q_n};\mu)=e^{i(t_k-t_{k+Q_1}) \Delta^{(k)}_{\pm}} B_{\mu(k+1);k+1,\cdots,k+Q_1}e^{i(t_{k+Q_1}-t_{k+Q_2}) \Delta^{(k+Q_1)}_{\pm}} \cdots  \\
 &  \times e^{i(t_{k+Q_{n-1}}-t_{k+Q_n})\Delta_{\pm}^{(k+Q_{n-1})}} B_{\mu(k+Q_{n-1}+1);k+Q_{n-1}+1,\cdots,k+Q_n}\big(\gamma^{(k+Q_n)}(t_{k+Q_n})\big)
\end{align*}

Thanks to Corollary \ref{cor:equal domain integral} and Lemma \ref{lem:number of echelon matrix} we can write $\gamma^{(k)}(t_k,\cdot)$ as a sum of at most $2^{k+(Q_1+1)(n-1)}$ terms of the form
\begin{equation}
\int\limits_{D} J^k(\underline{t}_{k+Q_n};\mu_s)dt_{k+Q_1} \dots dt_{k+Q_n}.
\end{equation}
Let $I_k^n=\overbrace{[0,t_k] \times [0,t_k] \times \cdots \times
[0,t_k]}^{\text{n copies}}$  \\
and $D_{t_{k+Q_1}}=\{(t_{k+Q_2},\cdots,t_{k+Q_n})|(t_{k+Q_1},t_{k+Q_2},\cdots,t_{k+Q_n}) \in D\}$, then 

\begin{flushleft}
\begin{align*}
& \Big\|S^{(k,\alpha)}\int_{0}^{t_k}\dotsi\int_{0}^{t_{k+Q_{n-1}}} J^k(\underline{t}_{k+Q_n})dt_{k+Q_1}\dots dt_{k+Q_n}\Big\|_{L^2(\mathbb{R}^{2k} \times \mathbb{R}^{2k})} \qquad \qquad \qquad \qquad   \\
& \mathrel{\mathop{\sim}^{\rm \footnotemark}} \Big\|S^{(k,\alpha)} \int_D  e^{i(t_k-t_{k+Q_1}) \Delta^{(k)}_{\pm}} B_{\mu_s(k+1);k+1,\cdots,k+Q_1} e^{i(t_{k+Q_1}-t_{k+Q_2})\Delta^{(k+Q_1)}_{\pm}}   \\
& \qquad \times B_{\mu_s(k+Q_1+1);k+Q_1+1,\cdots,k+Q_2}\cdots dt_{k+Q_1}\cdots dt_{k+Q_n} \Big\|_{L^2(\mathbb{R}^{2k} \times \mathbb{R}^{2k})}  
\end{align*}
\end{flushleft}
\footnotetext{the implicit scalar is proportional to $2^{k+(p_0+1)(n-1)}$ which can be absorbed in $C^k(C_0\sqrt{T})^n$.}

\begin{align*}
& = \Big\|\int_0^{t_k} e^{i(t_k-t_{k+Q_1})\Delta^{(k)}_{\pm}} \Big(\int_{D_{t_{k+Q_1}}} S^{(k,\alpha)}B_{\mu_s(k+1);k+1,\cdots,k+Q_1} e^{i(t_{k+Q_1}-t_{k+Q_2})\Delta^{(k+Q_1)}_{\pm}}   \\
& \qquad \qquad \times B_{\mu_s(k+Q_1+1);k+Q_1+1,\cdots,k+Q_2}\cdots dt_{k+Q_2}\cdots dt_{k+Q_n}\Big)dt_{k+Q_1} \Big\|_{L^2(\mathbb{R}^{2k} \times \mathbb{R}^{2k})}  \\
& \leq \int_0^{t_k}\Big\| e^{i(t_k-t_{k+Q_1})\Delta^{(k)}_{\pm}} \int_{D_{t_{k+Q_1}}} S^{(k,\alpha)}B_{\mu_s(k+1);k+1,\cdots,k+Q_1} e^{i(t_{k+Q_1}-t_{k+Q_2})\Delta^{(k+Q_1)}_{\pm}}   \\
& \qquad \qquad \times B_{\mu_s(k+Q_1+1);k+Q_1+1,\cdots,k+Q_2}\cdots dt_{k+Q_2}\cdots dt_{k+Q_n} \Big\|_{L^2(\mathbb{R}^{2k} \times \mathbb{R}^{2k})} dt_{k+Q_1}   \\
& = \int_0^{t_k}\Big\| \int_{D_{t_{k+Q_1}}} S^{(k,\alpha)}B_{\mu_s(k+1);k+1,\cdots,k+Q_1} e^{i(t_{k+Q_1}-t_{k+Q_2})\Delta^{(k+Q_1)}_{\pm}}   \\
& \qquad \qquad \times B_{\mu_s(k+Q_1+1);k+Q_1+1,\cdots,k+Q_2}\cdots dt_{k+Q_2}\cdots dt_{k+Q_n} \Big\|_{L^2(\mathbb{R}^{2k} \times \mathbb{R}^{2k})} dt_{k+Q_1}   \\
& \leq \int_0^{t_k}\Big( \int_{D_{t_{k+Q_1}}} \Big\|  S^{(k,\alpha)}B_{\mu_s(k+1);k+1,\cdots,k+Q_1} e^{i(t_{k+Q_1}-t_{k+Q_2})\Delta^{(k+Q_1)}_{\pm}}   \\
& \qquad \qquad \times B_{\mu_s(k+Q_1+1);k+Q_1+1,\cdots,k+Q_2}\cdots\Big\|_{L^2(\mathbb{R}^{2k} \times\mathbb{R}^{2k})} dt_{k+Q_2}\cdots dt_{k+Q_n} \Big) dt_{k+Q_1}   \\
& \leq \int_{I_k^n} \Big\|  S^{(k,\alpha)}B_{\mu_s(k+1);k+1,\cdots,k+Q_1} e^{i(t_{k+Q_1}-t_{k+Q_2})\Delta^{(k+Q_1)}_{\pm}} \\
& \qquad \qquad \times B_{\mu_s(k+Q_1+1);k+Q_1+1,\cdots,k+Q_2}\cdots\Big\|_{L^2(\mathbb{R}^{2k} \times\mathbb{R}^{2k})} dt_{k+Q_1}dt_{k+Q_2}\cdots dt_{k+Q_n}   \\
& \mathrel{\mathop{\leq}^{\rm \text{C-S on $t_{k+Q_1}$}}} t_k^{\frac{1}{2}}\int_{I_k^{n-1}} \Big\|  S^{(k,\alpha)}B_{\mu_s(k+1);k+1,\cdots,k+Q_1} e^{i(t_{k+Q_1}-t_{k+Q_2})\Delta^{(k+Q_1)}_{\pm}}  \\
& \qquad \qquad \times \Big(B_{\mu_s(k+Q_1+1);k+Q_1+1,\cdots,k+Q_2}\cdots\Big) \Big\|_{L^2\big((t_{k+Q_1}\in[0,t_k])\times \mathbb{R}^{2k} \times\mathbb{R}^{2k}\big)} dt_{k+Q_2}\cdots dt_{k+Q_n}  \\
& \mathrel{\mathop{\leq}^{\rm \text{\eqref{free evolving bound}}}} C_{\alpha}t_k^{\frac{1}{2}}\int_{I_k^{n-1}} \Big\| S^{(k+Q_1,\alpha)} B_{\mu_s(k+Q_1+1);k+Q_1+1,\cdots,k+Q_2} e^{i(t_{k+Q_2}-t_{k+Q_3})\Delta^{(k+Q_2)}_{\pm}}  \\
& \qquad \qquad \times \Big(B_{\mu_s(k+Q_2+1);k+Q_2+1,\cdots,k+Q_2+Q_1}\cdots\Big) \Big\|_{L^2\big(\mathbb{R}^{2(k+Q_1)} \times\mathbb{R}^{2(k+Q_1)}\big)}dt_{k+Q_2}\cdots dt_{k+Q_n}  \\
& \leq\cdots [\text{repeatedly applying Cauchy-Schwartz inequality and \eqref{free evolving bound}}]  \\
& \leq (C_{\alpha}t_k^{\frac{1}{2}})^{n-1} \int_0^{t_k} \Big\| S^{(k+Q_{n-1},\alpha)} B_{\mu_s(k+Q_{n-1}+1);k+Q_{n-1}+1,\cdots,k+Q_n}    \\
& \qquad \qquad \times \gamma^{k+Q_n}(t_{k+Q_n},\cdot) \Big\|_{L^2\big(\mathbb{R}^{2(k+Q_{n-1})} \times\mathbb{R}^{2(k+Q_{n-1})}\big)}dt_{k+Q_n}    \\
& \mathrel{\mathop{\leq}^{\rm \text{\eqref{priori spacial bound}}}} (C_{\alpha}t_k^{\frac{1}{2}})^{n-1} C^{k+Q_n}   \\
& \leq C^k(C_0\sqrt{T})^n
\end{align*}
\normalsize
we choose appropriate $C$ and $C_0$ to obtain the last line. As we have already seen in the proof of Theorem \ref{thm:1D uniqueness}, the sum over $p_1,p_2,\cdots, p_n$ and product of potential constants $b_0^{(p_1)},b_0^{(p_2)} \cdots, b_0^{(p_n)}$ will contribute extra factor $p_0 (b_0)^n$, which can be absorbed in constant $C$ and $C_0$. This completes the proof. 
\end{proof}

\normalsize
\begin{rmk}
 The main ingredients in the above proof are the free evolution bound \eqref{free evolving bound} and a priori energy bound \eqref{priori spacial bound}. The a priori energy bound usually requires $\alpha \leq 1$ (may see \eqref{alpha upper bound}). While in \eqref{free evolving bound}, we will need $\alpha > \frac{d}{2}-\frac{1}{2(2p-1)}$ which is at least $1$ when $d\geq 3$ (see \eqref{lower bound on alpha}). Therefore only the cases $d=1,2$ give a nonempty intersection for the survival of $\alpha$. Which implies that, under this setting, the method we used here to prove the uniqueness fails for the higher dimensional cases, unless we have better constrains on $\alpha$. Klainerman and Machedon obtained a better estimate (on a different space) here which allows them to prove the uniqueness for the case $d=3$, $p=1$. Actually, we are answering the same questions on the convergence of BBGKY hierarchy to $p$-GP hierarchy as in \cite{KSS} (for $d=2$, $p=1$) and \cite{CPquintic} (for $d\leq 2$, $p=2$), for any 
positive integer $p$. The case when $d=3$, $p=1$ is covered by \cite{CPcubic} recently with a new approach. 
\end{rmk}

\begin{acklg} 
 I sincerely thank my advisor Nata\v{s}a Pavlovi\'c for useful discussions and comments on this topic and her numerous valuable advice and generous help during my study and research. I am also grateful to Thomas Chen and Maja Taskovi\'c for reading the manuscript carefully and for many help suggestions. 
\end{acklg}

\appendix

\section{Approximation of the initial wave function}
Recall the proof of the a priori bound in Corollary \ref{cor:finite a priori energy bound}, we need the expectation of $H_N^k$ to be of the order $N^k$ at time 0. The main idea to obtain this is to approximate the initial wave function with cutoffs. We will prove Lemma \ref{lem: regularization of the initial data} in this section, with which \eqref{trace convergence} is immediate.

\begin{lemma}    \label{lem: regularization of the initial data}
 Suppose $\psi_N \in L^2(\mathbb{R}^{dN})$ with $\|\psi_N\|=1$ is a family of $N$-particle wave functions with the associated marginal densities $\gamma_N^{(k)}$, $k=1,2,\cdots$.  \\
Let $\chi$ be a bump function such that $0 \leq \chi \leq 1$, $\chi(s)=1$ for $s\in [0,1]$ and $\chi(s)=0$ for $s\geq 2$. $\kappa>0$ is a parameter. Define
\begin{equation}   
 \tilde\psi_N:=\frac{\chi(\frac{\kappa}{N}H_N)\psi_N}{\big\|\chi(\frac{\kappa}{N}H_N)\psi_N \big\|}
\end{equation}
We denote by $\tilde\gamma_{N,t}^{(k)}$ the corresponding $k$-marginal density associated with $\tilde\psi_N$.   
We also assume that 
\begin{equation}  \label{first order N bound}
 \langle \psi_N, H_N\psi_N\rangle \leq CN
\end{equation}
and 
\begin{equation}   \label{first marginal convergence}
 \gamma_N^{(1)} \to \Ket\phi \Bra\phi \quad as \ \ N \to \infty
\end{equation}
with $\phi \in H^1(\mathbb{R}^d)$. 
Then, for $\kappa>0$ small enough and for every $k\leq 1$ we have
\begin{equation}   \label{trace convergence needed}
  \lim_{N\to \infty}Tr\big|\tilde\gamma_{N}^{(k)}-\Ket{\phi} \Bra{\phi}^{\otimes k}\big|=0
\end{equation}
\end{lemma}
 
\begin{proof}
 The proof is similar in spirit to the proof for the two-body interactions case, which can be found in \cite{ESY},\cite{ESY06},\cite{ESY09}. Sketch of the key steps are listed below. We just need to show
\begin{equation}   \label{trace convenience to proof}
  Tr \big|\tilde\gamma_{N}^{(1)}-\Ket{\phi} \Bra{\phi}\big|\to 0, \quad as \ \ N\to \infty
\end{equation}
since \eqref{trace convenience to proof} implies \eqref{trace convergence needed} (proved by Lieb and Seiringer in \cite{LS}). Moreover, by the equivalence of weak$^*$ convergence and trace norm convergence, it is enough to prove that for every compact operator $J^{(1)}\in \mathcal{K}_1$ and for every $\epsilon>0$, there exists $N_0=N_0(J^{(1)}, \epsilon)$ such that
\begin{equation}   \label{first trace convergence}
 \big|Tr J^{(1)}\big(\tilde\gamma_{N}^{(1)}-\Ket{\phi} \Bra{\phi}\big)\big|\leq \epsilon, \quad for \ \ N>N_0
\end{equation}
The proof of \eqref{first trace convergence} is divided into five steps.

\emph{Step 1.} By \eqref{first marginal convergence}, we know that there exists a sequence $\xi_N^{(N-1)} \in L^2(\mathbb{R}^{d(N-1)})$, $\|\xi_N^{(N-1)}\|=1$ satisfying
\begin{equation}    \label{step 1 estimate}
 \|\psi_N-\phi\otimes \xi_N^{(N-1)}\|  \to 0, \quad as \ \ N\to 0
\end{equation}
This was proved by Alessandro Michelangeli in \cite{Michelangeli}. The proof in current case is identical to the proof presented in \cite{ESY09}. 

\emph{Step 2.} There exists $\phi_* \in H^2(\mathbb{R}^d)$ with $\|\phi_*\|=1$ such that
\begin{equation}   \label{step 2 estimate}
 \|\phi-\phi_*\| \leq \frac{\epsilon}{32\vvvert J^{(1)}\vvvert}.
\end{equation}

\emph{Step 3.} Let $\Xi=\chi(\frac{\kappa}{N}H_N)$. Then by \eqref{first order N bound}:
\begin{equation}       
 \|(1-\Xi)\psi_N\|^2=\langle\psi_N, (1-\Xi)^2\psi_N\rangle  \leq\langle\psi_N, \mathbf{1}(\kappa H_N\geq N)\psi_N\rangle \leq \frac{\kappa}{N} \langle \psi_N, H_N\psi_N\rangle \leq C\kappa
\end{equation}
is uniformly in $N$. Since $\|\psi_N\|=1$, by triangle inequality we know
\begin{equation}      \label{distance between psiN tildepsiN}
 \|\psi_N-\tilde\psi_N\|=\|\frac{\psi_N}{\|\psi_N\|}-\frac{\Xi\psi_N}{\|\Xi\psi_N\|}\| \leq \frac{2}{\|\psi_N\|}\|\psi_N-\Xi\psi_N\|=2\|(1-\Xi)\psi_N\| \leq C\kappa^{\frac{1}{2}}.
\end{equation}
The above inequality is needed in \eqref{bound by psiN tildepsiN}. One can find $\kappa>0$ small enough such that $\|\Xi\psi_N\|\geq \frac{1}{2}$. Use triangle inequality and note that $\|\Xi\| \leq 1$. We have
\begin{equation}    \label{step 3 estimate}
 \begin{split}
 & \big\|\frac{\Xi\psi_N}{\|\Xi\psi_N\|}-\frac{\Xi(\phi_*\otimes\xi_N^{(N-1)})}{\|\Xi(\phi_*\otimes\xi_N^{(N-1)})\|} \big\|     \\
 & = \big\|\frac{\Xi\psi_N}{\|\Xi\psi_N\|}-\frac{\Xi(\phi_*\otimes\xi_N^{(N-1)})}{\|\Xi\psi_N\|}+\frac{\Xi(\phi_*\otimes\xi_N^{(N-1)})}{\|\Xi\psi_N\|}-\frac{\Xi(\phi_*\otimes\xi_N^{(N-1)})}{\|\Xi(\phi_*\otimes\xi_N^{(N-1)})\|} \big\|    \\
 & \leq \frac{1}{\|\Xi\psi_N\|}\big\|\Xi\psi_N-\Xi(\phi_*\otimes\xi_N^{(N-1)})\big\|+\frac{1}{\|\Xi\psi_N\|}\big\|\|\Xi\psi_N\|-\|\Xi(\phi_*\otimes\xi_N^{(N-1)})\|\big\|      \\
 & \leq \frac{2}{\|\Xi\psi_N\|}\big\|\Xi\big(\psi_N-\phi_*\otimes\xi_N^{(N-1)}\big)\big\|   \\
 & \leq 4\big\|\psi_N-\phi_*\otimes\xi_N^{(N-1)}\big\|   \\
 & \leq 4\big\|\psi_N-\phi\otimes\xi_N^{(N-1)}\big\|+4\big\|\phi\otimes\xi_N^{(N-1)} -\phi_*\otimes\xi_N^{(N-1)}\big\|   \\
 & \leq 4\big\|\psi_N-\phi\otimes\xi_N^{(N-1)}\big\|+4\big\|\phi -\phi_*\big\|   \\
 & \leq \frac{\epsilon}{6\vvvert J^{(1)}\vvvert}
 \end{split}
\end{equation}
for large $N$. Here in the last inequality we used \eqref{step 1 estimate} and \eqref{step 2 estimate}.

\emph{Step 4.} As in \cite{ESY09} and \cite{CPquintic}, we define a similar Hamiltonian after taking into account of the ($p+1$)-particle interactions studied in this paper.
\begin{equation}       \label{breve Hamiltonian}
  \breve{H}_N:=\sum_{i=2}^{N}(-\Delta_{x_i})+\frac{1}{N^{p}}\sum_{2\leq i_1<\cdots<i_{p+1}\leq N}V^{(p)}_N\big(x_{i_1}-x_{i_2},\cdots,x_{i_1}-x_{i_{p+1}}\big)
\end{equation}
Instead of acting on all variables, the new Hamiltonian only acts on the last $N-1$ variables. Let $\breve\Xi=\chi(\frac{\kappa}{N}\breve{H}_N)$. Then by \eqref{step 3 estimate}, we will have
\begin{equation}     \label{step 4 estimate}
 \big\|\frac{\Xi\psi_N}{\|\Xi\psi_N\|}-\frac{\breve{\Xi}(\phi_*\otimes\xi_N^{(N-1)})}{\|\breve{\Xi}(\phi_*\otimes\xi_N^{(N-1)})\|} \big\|  \leq \frac{\epsilon}{3\vvvert J^{(1)}\vvvert}
\end{equation}
We refer the proof of \eqref{step 4 estimate} to \cite{ESY06}.

\emph{Step 5.} For \eqref{first trace convergence}, we define
\begin{equation}
 \breve{\psi}_N:=\frac{\breve{\Xi}(\phi_*\otimes\xi_N^{(N-1)})}{\|\breve{\Xi}(\phi_*\otimes\xi_N^{(N-1)})\|}=\phi_*\otimes \frac{\breve{\Xi}\xi_N^{(N-1)}}{\|\breve{\Xi}\xi_N^{(N-1)}\|}
\end{equation}
since $\breve{\Xi}$ only acts on the last $N-1$ variables and $\|\phi_*\|=1$. Further, we define
\begin{equation}
 \breve{\gamma}_{N}^{(1)}(x_1; x'_1):=\int \breve{\psi}_N(x_1;\textbf{x}_{N-1}) \bar{\breve{\psi}}_N(x'_1,\textbf{x}_{N-1})d{\textbf{x}_{N-1}}
\end{equation}
Note that $\breve{\psi}_N$ is not symmetric in all variables, but it is symmetric in the last $N-1$ variables. Clearly, $\breve{\gamma}_{N}^{(1)}$ is a density matrix and
\begin{equation*}
 \breve{\gamma}_{N}^{(1)}=\Ket{\phi_*} \Bra{\phi_*} 
\end{equation*}
Thus, using $\|\tilde\psi_N-\breve{\psi}_N\| \leq \frac{\epsilon}{3\vvvert J^{(1)}\vvvert}$, which is equivalent to \eqref{step 4 estimate} and $\|\phi-\phi_*\| \leq \frac{\epsilon}{32\vvvert J^{(1)}\vvvert}$ from \eqref{step 2 estimate}, we obtain
\begin{equation}
 \begin{split}
  \big|Tr J^{(1)}\big(\tilde\gamma_{N}^{(1)}-\Ket{\phi} \Bra{\phi}\big)\big|&\leq \big|Tr J^{(1)}\big(\tilde\gamma_{N}^{(1)}-\breve{\gamma}_{N}^{(1)}\big)\big|+\big|Tr J^{(1)}\big(\Ket{\phi_*} \Bra{\phi_*}-\Ket{\phi} \Bra{\phi}\big)\big|   \\
  & \leq 2\big\vvvert J^{(1)}\big\vvvert \|\tilde\psi_N-\breve{\psi}_N\|+2\big\vvvert J^{(1)}\big\vvvert \|\phi_*-\phi\|    \\
  & \leq \epsilon
 \end{split}
\end{equation}
for sufficiently large $N$ with arbitrary $\epsilon$ and small enough $\kappa$. Hence \eqref{first trace convergence} follows.
\end{proof}

\section{A Poincar\'e type inequality}

\begin{lemma}  \label{lem: lemma B}
 Let $h$ be a non-negative probability measure on $\mathbb{R}^d$ satisfying $\int_{\mathbb{R}^d}(1+x^2)^{\frac{1}{2}}h(x)dx<\infty$. Then for $h_{\epsilon}(x)=\frac{1}{\epsilon^d}h(\frac{x}{\epsilon}), \epsilon >0$, and every $0\leq \kappa <1$, there exists a $C>0$ such that
\begin{equation}
 \begin{split}
  &\big|Tr J^{(k)}\big(h_{\epsilon}(x_j-x_{k+1})\cdots h_{\epsilon}(x_j-x_{k+p})-\delta(x_j-x_{k+1})\cdots\delta(x_j-x_{k+p}) \big)\gamma^{(k+p)} \big|    \\
  & \qquad \leq C\epsilon^{\kappa}\big\vvvert J^{(k)}\big\vvvert Tr\big|S_j S_{k+1}\cdots S_{k+p}\gamma^{(k+p)}S_{k+p}\cdots S_{k+1}S_j \big|
 \end{split}
\end{equation}
for all non-negative $\gamma^{(k+p)} \in \mathcal{L}_{k+p}^1$
\end{lemma}

\begin{proof}
 Following \cite{KSS},\cite{CPquintic}, we prove the case $k=1$. For case of $k>1$, the proof is analogous. Since $1\leq j \leq k$, so $j=1$ in current case. By the non-negativity of $\gamma^{(1+p)}$, we can decompose it as  $\gamma^{(1+p)}=\sum_i \lambda_i \Ket{\psi_i} \Bra{\psi_i}$, with $\psi_i \in L^2(\mathbb{R}^{(1+p)d})$ and $\lambda_i\geq 0$, $\sum \lambda_i \leq 1$. Then
\begin{equation}
 \begin{split}
   &Tr J^{(1)}\big(h_{\epsilon}(x_1-x_{2})\cdots h_{\epsilon}(x_1-x_{1+p})-\delta(x_1-x_{2})\cdots\delta(x_1-x_{1+p}) \big)\gamma^{(1+p)}       \\
   & =\sum_i \lambda_i\langle\psi_i, J^{(1)}\big(h_{\epsilon}(x_1-x_{2})\cdots h_{\epsilon}(x_1-x_{1+p})-\delta(x_1-x_{2})\cdots\delta(x_1-x_{1+p}) \big)\psi_i\rangle     \\
   & =\sum_i \lambda_i\langle\Psi_i, \big(h_{\epsilon}(x_1-x_{2})\cdots h_{\epsilon}(x_1-x_{1+p})-\delta(x_1-x_{2})\cdots\delta(x_1-x_{1+p}) \big)\psi_i\rangle     \\ 
 \end{split}
\end{equation}
where $\Psi_i=(J^{(1)}\otimes 1)\psi_i$. Next we switch to Fourier side to obtain
\begin{equation}
 \begin{split}
   &\langle\Psi_i, \big(h_{\epsilon}(x_1-x_{2})\cdots h_{\epsilon}(x_1-x_{1+p})-\delta(x_1-x_{2})\cdots\delta(x_1-x_{1+p}) \big)\psi_i\rangle   \\
   & =\int dq_1\cdots dq_{1+p}dq'_1\cdots dq'_{1+p}\bar{\hat\Psi}_i(q_1,\cdots,q_{1+p})\hat{\psi}_i(q'_1,\cdots,q'_{1+p})    \\
   & \quad \times \int dx_2\cdots dx_{1+p}h(x_2)\cdots h(x_{1+p})(e^{i\epsilon x_2(q_2-q'_2)}\cdots e^{i\epsilon x_{1+p}(q_{1+p}-q'_{1+p})}-1)   \\
   & \qquad \times \delta(q_1+\cdots+q_{1+p}-q'_1-\cdots-q'_{1+p})
 \end{split}
\end{equation}
Since for $x \in \mathbb{R}$, $|e^{ix}-1|=2|\sin \frac{x}{2}|\leq C|x|^{\kappa}$ is always true with arbitrary $0<\kappa<1$ and  constant $C>0$ independent of $\kappa$, we have the following
\begin{equation}
 \begin{split}
  \big|e^{i\epsilon x_2(q_2-q'_2)}\cdots e^{i\epsilon x_{1+p}(q_{1+p}-q'_{1+p})}-1\big| &\leq C\epsilon^{\kappa}\big(\sum_{i=2}^{1+p}|x_i(q_i-q'_i)|\big)^{\kappa}  \\
  & \leq C\epsilon^{\kappa}\sum_{i=2}^{1+p}|x_i(q_i-q'_i)|^{\kappa}     \\
  & \leq C\epsilon^{\kappa}\sum_{i=2}^{1+p}|x_i|^{\kappa}(|q_i|^{\kappa}+|q'_i|^{\kappa})
 \end{split}
\end{equation}
The last inequality follows from $(a+b)^{\kappa}\leq a^{\kappa}+b^{\kappa}$ for $\kappa \in (0,1)$ and $a, b$ both nonnegative. And the second to the last inequality follows in a similar way, but with an implicit constant depending on $p$. Thus
\small
\begin{equation}
 \begin{split}
  &\big|\langle\Psi_i, \big(h_{\epsilon}(x_1-x_{2})\cdots h_{\epsilon}(x_1-x_{1+p})-\delta(x_1-x_{2})\cdots\delta(x_1-x_{1+p}) \big)\psi_i\rangle\big| \\
  & \quad \leq C\epsilon^{\kappa} \int dq_1\cdots dq_{1+p}dq'_1\cdots dq'_{1+p}|\hat\Psi_i(q_1,\cdots,q_{1+p})||\hat{\psi}_i(q'_1,\cdots,q'_{1+p})|    \\
  & \qquad \times \Big(\prod_{i=2}^{1+p}\int |x_i|^{\kappa}h(x_i)dx_i\Big) \Big(\sum_{i=2}^{1+p}|q_i|^{\kappa}+|q'_i|^{\kappa}\Big) \delta(q_1+\cdots+q_{1+p}-q'_1-\cdots-q'_{1+p})
 \end{split}
\end{equation}
\normalsize
Clearly the $p$ copies of integrations involving $h$ are finite by assumption. And the summation term $\sum_{i=2}^{1+p}(|q_i|^{\kappa}+|q'_i|^{\kappa})$ contains a total of $p$ terms. We will show how to control one of them, say $|q_2|^k$. The final upper bound on this part will be the same (up to a constant $p$).
\small 
\begin{align}
  &\int dq_1\cdots dq_{1+p}dq'_1\cdots dq'_{1+p}\delta(q_1+\cdots+q_{1+p}-q'_1-\cdots-q'_{1+p})  \notag \\
  & \qquad \times \big|\hat\Psi_i(q_1,\cdots,q_{1+p})\big|\big|\hat{\psi}_i(q'_1,\cdots,q'_{1+p})\big||q_2|^{\kappa}   \notag \\
  & \quad =\int dq_1\cdots dq_{1+p}dq'_1\cdots dq'_{1+p}\delta(q_1+\cdots+q_{1+p}-q'_1-\cdots-q'_{1+p})  \notag \\
  & \qquad \quad\times\frac{\langle q_1\rangle \langle q_2\rangle\cdots\langle q_{1+p}\rangle}{\langle q'_1\rangle \langle q'_2\rangle\cdots\langle q'_{1+p}\rangle}\big|\hat\Psi_i(q_1,\cdots,q_{1+p})\big|   \frac{\langle q'_1\rangle \langle q'_2\rangle\cdots\langle q'_{1+p}\rangle}{\langle q_1\rangle \langle q_2\rangle^{1-\kappa}\cdots\langle q_{1+p}\rangle}\big|\hat{\psi}_i(q'_1,\cdots,q'_{1+p})\big|  \notag  \\
  & \quad \leq \rho \int dq_1\cdots dq_{1+p}dq'_1\cdots dq'_{1+p}\delta(q_1+\cdots+q_{1+p}-q'_1-\cdots-q'_{1+p})  \notag  \\
  & \qquad \quad\times\frac{\langle q_1\rangle^2 \langle q_2\rangle^2\cdots\langle q_{1+p}\rangle^2}{\langle q'_1\rangle^2 \langle q'_2\rangle^2\cdots\langle q'_{1+p}\rangle^2}\big|\hat\Psi_i(q_1,\cdots,q_{1+p})\big|^2   \notag \\
  & \qquad + \frac{1}{\rho} \int dq_1\cdots dq_{1+p}dq'_1\cdots dq'_{1+p}\delta(q_1+\cdots+q_{1+p}-q'_1-\cdots-q'_{1+p})  \notag  \\
  & \qquad \qquad\times\frac{\langle q'_1\rangle^2 \langle q'_2\rangle^2\cdots\langle q'_{1+p}\rangle^2}{\langle q_1\rangle^2 \langle q_2\rangle^{2(1-\kappa)} \langle q_3\rangle^2\cdots \langle q_{1+p}\rangle^2}\big|\hat\psi_i(q'_1,\cdots,q'_{1+p})\big|^2    \notag \\
  & \quad \leq \rho \langle\Psi_i, S_1^2S_2^2\cdots S_{1+p}^2\Psi_i\rangle \sup_{Q'}\int \frac{dq'_1\cdots dq'_{p}}{\langle q'_1\rangle^2 \langle q'_2\rangle^2\cdots\langle q'_{p}\rangle^2\langle Q'-q'_1-\cdots-q'_{p}\rangle^2}   \label{last 1 integration}  \\
  & \qquad +\frac{1}{\rho} \langle\psi_i, S_1^2S_2^2\cdots S_{1+p}^2 \psi_i\rangle \sup_{Q}\int \frac{dq_1dq_3\cdots dq_{1+p}}{\langle q_1\rangle^2 \langle Q-q_1-q_3-\cdots-q_{p}\rangle^{2(1-\kappa)}\langle q_3\rangle^2\cdots\langle q_{1+p}\rangle^2}    \label{last 2 integration}
\end{align}
\normalsize
for arbitrary $\rho>0$. We can apply \eqref{analysis inequality} to the last two integrations \eqref{last 1 integration} and \eqref{last 2 integration} for all $\kappa \in (0,1)$ to have
\begin{equation}
 \begin{split}
   &\big|Tr J^{(1)}\big(h_{\epsilon}(x_1-x_{2})\cdots h_{\epsilon}(x_1-x_{1+p})-\delta(x_1-x_{2})\cdots\delta(x_1-x_{1+p}) \big)\gamma^{(1+p)}\big|     \\
   & \quad \leq C\epsilon^{\kappa}\big(\rho Tr J^{(1)}S_1^2S_2^2\cdots S_{1+p}^2 J^{(1)}\gamma^{(1+p)}+\frac{1}{\rho}Tr S_1^2S_2^2\cdots S_{1+p}^2\gamma^{(1+p)}\big)      \\
   & \quad \leq C\epsilon^{\kappa}\big(\rho Tr S_1^{-1}J^{(1)}S_1^2 J^{(1)}S_1^{-1}S_1S_2\cdots S_{1+p}\gamma^{(1+p)}S_{1+p}\cdots S_2S_1         \\
   & \qquad \qquad \quad +\frac{1}{\rho}Tr S_1^2S_2^2\cdots S_{1+p}^2 \gamma^{(1+p)}\big)      \\
   &\quad \leq C\epsilon^{\kappa}\big(\rho\|S_1^{-1}J^{(1)}S_1\| \|S_1J^{(1)}S_1^{-1}\|+\frac{1}{\rho}\big)Tr S_1^2S_2^2\cdots S_{1+p}^2 \gamma^{(1+p)}   \\
   & \quad \leq C\epsilon^{\kappa}\vvvert J^{(1)}\vvvert Tr S_1^2S_2^2\cdots S_{1+p}^2 \gamma^{(1+p)}
 \end{split}
\end{equation}
by taking $\rho=\vvvert J^{(1)}\vvvert^{-1}$ in the last inequality.
\end{proof}

\vspace{1mm}
\bibliographystyle{abbrv}
\bibliography{UniqueGP}

\end{document}